\documentclass[12pt]{amsart}
\usepackage[a4paper,margin=1in]{geometry}
\usepackage[utf8]{inputenc}

\usepackage{scalerel}

\usepackage{amsfonts}
\usepackage{amsmath}
\usepackage{amssymb}
\usepackage{amsthm}
\usepackage[pagebackref, colorlinks = true, linkcolor = blue, urlcolor  = blue, citecolor = red]{hyperref}
\usepackage{stmaryrd}

\usepackage{pgfplots}
\pgfplotsset{compat = 1.16}

\usepackage{physics}

\usepackage{permutation}

\usepackage{todonotes}

\newtheorem{theorem}{Theorem}[section]
\newtheorem{definition}[theorem]{Definition}
\newtheorem{proposition}[theorem]{Proposition}

\newtheorem{lemma}[theorem]{Lemma}
\newtheorem{remark}[theorem]{Remark}
\newtheorem{example}[theorem]{Example}

\DeclareMathSymbol{\shortminus}{\mathbin}{AMSa}{"39}

\DeclareMathOperator*{\E}{\mathbb{E}}
\DeclareMathOperator{\id}{id}
\DeclareMathOperator{\Ima}{Im}
\DeclareMathOperator{\Ker}{Ker}

\DeclareMathOperator{\sign}{sign}
\DeclareMathOperator{\Span}{Span}
\DeclareMathOperator{\spec}{Spec}

\newcommand{\T}{\mathsf{T}}
\newcommand{\To}{\mathsf{\Gamma}}

\title{The asymmetric quantum cloning region}

\author{Ion Nechita}
\email{ion.nechita@univ-tlse3.fr}
\address{Laboratoire de Physique Th\'eorique, Universit\'e de Toulouse, CNRS, UPS, France}

\author{Clément Pellegrini}
\email{clement.pellegrini@math.univ-toulouse.fr}
\address{Institut de Mathématiques, Université de Toulouse, CNRS, UPS, France}

\author{Denis Rochette}
\email{denis.rochette@math.univ-toulouse.fr}
\address{Institut de Mathématiques, Université de Toulouse, CNRS, UPS, France}

\date{\today}

\begin{document}

\maketitle

\begin{abstract}
  Quantum cloning is a fundamental protocol of quantum information theory. Perfect universal quantum cloning is prohibited by the laws of quantum mechanics, only imperfect copies being reachable. Symmetric quantum cloning is concerned with case when the quality of the clones is identical. In this work, we study the general case of $1 \to N$ asymmetric cloning, where one asks for arbitrary qualities of the clones. We characterize, for all Hilbert space dimensions and number of clones, the set of all possible clone qualities. This set is realized as the nonnegative part of the unit ball of a newly introduced norm, which we call the $\mathcal Q$-norm.  We also provide a closed form expression for the quantum cloner achieving a given clone quality vector. Our analysis relies on the Schur-Weyl duality and on the study of the spectral properties of partially transposed permutation operators.
\end{abstract}

\tableofcontents

\section{Introduction}

In quantum mechanics, the information is encoded in the state $\vert\psi\rangle$ of a quantum system, that is a unit vector of a complex Hilbert space $\mathcal H$. When thinking of information and communication, a general task of a computer is to "copy and paste" the information. This task, up to a certain precision, is perfectly performed by a classical computer. At the level of quantum mechanics, the duplication of a quantum state is a protocol $\vert\psi\rangle\mapsto\vert\psi\rangle\otimes\vert\psi\rangle$ or $\vert\psi\rangle\mapsto\vert\psi\rangle\otimes\ldots\otimes\vert\psi\rangle=\vert\psi\rangle^{\otimes N}$ when one wants to obtain $N$ copies of $\vert\psi\rangle$. Such protocol carries the name of \emph{quantum cloning}. More generally, one aims at cloning mixed states $\rho\mapsto\rho^{\otimes N}$, where $\rho$ is a density matrix, that is a positive trace one operator on $\mathcal H$. If $\mathcal H=\mathbb C^d$, a quantum cloner is an application $T$ from $\mathcal M_d$ to $\mathcal M_d^{\otimes N}$, where $\mathcal M_d$ is the set of $d\times d$ matrices. As usual, such applications are required to be \emph{quantum channels}, that is completely positive and trace preserving linear maps. This scenario can be relaxed by allowing the output $T(\rho)$ to be a non-product state. In this case, only partial traces are required to be equal to the input,  that is $T_i(\rho)=\rho$, for all clone indices $i=1,\ldots,n$, where $T_i$ is the quantum channel from $\mathcal M_d$ on the $i$-th copy of $\mathcal M_d$ after taking the partial trace on the other copies of $\mathcal M_d$. In this sense $T_i(\rho)$ corresponds to the copy number $i$ of $\rho$. Such protocols are called \emph{perfect quantum cloning}.

Designing perfect quantum cloners, able to clone any state, is known to be impossible; this fundamental fact is known under the name of ``no-cloning theorem'' \cite{dieks1982communication,wootters1982single}. Indeed cloning arbitrary quantum information is forbidden by the rules of quantum mechanics. This no-go result is one of the fundamental differences between classical and quantum information processing, and has received a lot of attention in the recent years.  

Even if perfect cloning is impossible, the story is not over. Since perfect cloning is unreachable, one has to relax the task of obtaining perfect copies of quantum states. This is the central subject of \emph{approximate quantum cloning}. In this context one allows to obtain copies of states which are not perfect. The quality of the copies can be then measured as the distance between the input (i.e the state to be cloned) and the different copies $T_i(\rho)$, ($i=1,\ldots,n$), in terms of a figure of merit called \emph{quantum fidelity}. Two types of approximate cloning can be considered: the \emph{symmetric quantum cloning}, where one asks to have the same quality for each copies or the \emph{asymmetric quantum cloning} where copies are allowed to have different qualities. In terms of quantum fidelity, the non cloning theorem states that the quality can not be one for each copy. 

The quantum cloning problem is a central topic in quantum information theory. Since the pioneering work on universal quantum cloners \cite{buvzek1998universal}, many authors studied the different cloning scenarios : symmetric \cite{werner1998optimal, keyl1999optimal} vs.~asymmetric \cite{cerf2000asymmetric, kay2014optimal, studzinski2014group, hashagen2016universal}, qubit \cite{gisin1997optimal} vs.~qudit \cite{fiurasek2005highly}, phase covariant \cite{bruss2000phase}, probabilistic \cite{duan1998probabilistic}, state dependent \cite{lostaglio2020contextual}, etc \cite{bruss2000approximate, du2005experimental}. For a comprehensive review one can consult \cite{scarani2005quantum, fan2014quantum} where theoretical results as well as practical considerations are exposed. From a concrete point of view where the practical design of concrete quantum cloner are considered one can consult \cite{coyle2020variational} where techniques of machine learning are used to obtain quantum circuits which can clone families of prescribed quantum states. Quantum cloning is an essential ingredient in cryptography protocols: the impossibility of non-cloning prevents a malicious eavesdropper from intercepting a message and copying it without disturbing the original. This main idea is the basis of the BB84 protocol \cite{bennett2020quantum}.

An inherent question is the physically realisable figures of merit for quantum cloners. Briefly speaking, the \emph{asymmetric quantum cloning region} can be described in terms of the quality of the clones produced by quantum channels. For each point in the cloning region, a concrete expression of a quantum cloners corresponding to this figure of merit is of major interest, especially for the optimal quantum cloners. The simplest situation when $N=2$ has been deeply studied with extensive results obtained using different approaches: selective measurements \cite{buvzek1996quantum}, unitary transformations \cite{buvzek1998universal, bruss1998optimal}, prior information \cite{du2005experimental}. The symmetric~$1 \to N$ quantum cloning problem has been fully described by Werner and Keyl in the two papers \cite{werner1998optimal, keyl1999optimal}, where the cloning maps are given by quantum channels. Recently, the asymmetric~$1\to2$ quantum cloning has been revisited in \cite{hashagen2016universal}, where different figures of merit are addressed. In the authors' recent work \cite{nechita2021geometrical}, the complete region has been described in terms of union of ellipses, and the expression of all (not just the optimal) possible cloners has been provided in terms of particular permutation operators.

Several recent papers address the more general problem of asymmetric $1 \to N$ quantum cloning. One the one hand, Kay and his collaborators \cite{kay2009optimal,kay2012optimal, kay2014optimal} consider the optimal quantum fidelity of the different copies is bounded by the largest eigenvalue of an operator. On the other hand, \'Cwikli\'nski, Horodecki, Mozrzymas, and Studzi\'nski \cite{cwiklinski2012region,studzinski2014group} use techniques from group representation theory to analyse the commutant structure of the unitaries operators~$U^{\otimes (N - 1)} \otimes \:  U^*$ and the algebra of partially transposed permutation operators, developed earlier for a more general framework in papers \cite{studzinski2013commutant} and \cite{mozrzymas2014structure, mozrzymas2018simplified}. These papers are principally focused on the optimal quantum cloners. An important result in \cite{kay2012optimal}, describes the optimal quantum cloners in terms of a quadratic relation satisfied by the fildelities of their marginals:
\begin{equation}\label{eq:Kay-formula}
    \sum^N_{i = 1} \big{(} (d + 1) f_i - 1 \big{)} = (d - 1) + \frac{{\Big{(} \sum^N_{i = 1} \sqrt{(d + 1) f_i  - 1} \Big{)}}^2}{N + d - 1}.
\end{equation}
In the equation above, $d$ is the dimension of the Hilbert space, $N$ is the number of clones, and $f_i$ are respectively the optimal achievable fidelities. Using some unitary-equivariant constraints (i.e., commutation relation with all the unitaries $U^{\otimes p} \otimes \bar{U}^{\otimes q}$), Grinko and Ozols \cite{grinko2022linear} have reduced semidefinite optimization problems into linear optimization problems, and in particular the asymmetric cloning problem.

The current paper is devoted to the case of \emph{$1 \to N$  asymmetric quantum cloning}. The originality of our paper lies in the fact that we completely describe, using techniques from functional analysis and representation theory, the cloning region with the help of a newly defined norm $\norm{\cdot}_{\mathcal{Q}}$, which we call the \emph{$\mathcal Q$-norm}.
\begin{definition}
    The $\mathcal{Q}$-norm of a vector $x \in \mathbb{R}^N$ defined by
    \begin{equation*}
        \norm{x}_{\mathcal{Q}} := \frac{d \, \lambda_{\max} \big{(} \sum^{N}_{i = 1} |x_i| \cdot \omega_{(0,i)} \otimes I^{\otimes (N - 1)} \big{)} - \sum_{i=1}^N |x_i|}{d^2 - 1},
    \end{equation*}
    where $\lambda_{\max}(\cdot)$ is the largest eigenvalue.
\end{definition}
More precisely the cloning region
\begin{equation*}
    \mathcal{R}_{N,d} := \bigg{\{} p \in [0,1]^N \: \bigg{|} \: \exists \, T \text{ cloning map s.t. } T_i(\rho) = p_i \!\cdot\! \rho + (1 - p_i) \frac{I}{d} \bigg{\}}
\end{equation*}
appears as the non-negative part of the unit ball of the dual norm $\norm{\cdot}^*_{\mathcal{Q}}$:
$$\mathcal{R}_{N,d} = \bigg{\{} p \in [0,1]^N \: \bigg{|} \: \forall q \in [0,1]^N \, s.t.\, \|q\|_{\mathcal Q} \leq 1, \, \langle p, q\rangle \leq 1\bigg \}.$$

Here, we study the whole set of parameters $p$ reached by quantum cloners, and not only the optimal ones (in some given direction). We naturally recover the result \eqref{eq:Kay-formula} of \cite{kay2012optimal} corresponding to optimal cloners. Surprisingly, the parameters of these optimal quantum cloners do not delimit the whole region, see Figure \ref{fig:flatRegion}. We exhibit points on the border of the cloning region that can not be directly obtained with the quadratic equation \eqref{eq:Kay-formula} of optimality. We start from first principles and use Schur-Weyl duality to design quantum cloners as sum of particular permutation operators. In particular, the technical core of the asymmetric quantum cloning problem is finding the eigenvector associated to the maximal eigenvalue of a certain Hermitian operator appearing in the definition of the $\mathcal Q$-norm. 

In this paper, we shall use the following notation for sets of integers: 
\begin{align*}
\llbracket n \rrbracket &:= \{1,2,\ldots, n\}\\
\llbracket m, n \rrbracket &:= \{m,m+1,\ldots, n\}.
\end{align*}

\bigskip
The paper is structured as follows. Section \ref{sec:asymmetric-cloning} is devoted to present the asymmetric quantum cloning problem. We present our figure of merit approach linked with quantum fidelities. We address the worst case situation and we show that it is equivalent to the average situation. In Section \ref{sec:permutation-operators}, we present the Schur-Weyl ingredient and link the quantum cloners with permutation operators. Next, in Section \ref{sec:optimal-quantum-cloning}, we study in details the eigenvalues and eigenvectors of these permutation operators and we focus on the maximal eigenvalue. This allows us to delimit, in terms of the spectral properties of these operators, the asymmetric quantum cloning region. In Section \ref{sec:cloning-region}, we introduce the $\mathcal Q$-norm, which allows to finally obtain the description of the cloning region in terms of the dual of this norm. In order to streamline the presentation, the proofs of numerous technical lemmas are postponed to appendices.

\section{Asymmetric quantum cloning}\label{sec:asymmetric-cloning}

In this paper, we consider a finite dimensional Hilbert space $\mathcal{H} = \mathbb{C}^d$. We denote by $\mathcal{M}_d$ the corresponding set of complex $d \times d$ matrices and by $\mathcal{U}(d)$ the set of unitary matrices that is $\mathcal{U}(d) = \{ U \in \mathcal{M}_d \: \vert \: U^* U = I \}$, where $I$ denotes the identity matrix. The Hilbert space $\mathcal{H}$ describes a quantum system and the states on $\mathcal{H}$ are the usual density matrices which is denoted by 
\begin{equation*}
    \mathcal{D}_d = \big{\{} \rho \in \mathcal{M}_d \; \big{|} \; \rho \geq 0, \, \Tr(\rho) = 1 \big{\}}.
\end{equation*}
We denote the unnormalized \emph{maximally entangled state} by
\begin{equation}\label{eq:def-omega}
    \omega = \ketbra{\Omega}{\Omega},\quad \textrm{with}\quad\ket{\Omega} = \sum^d_{i = 1} \ket{ii}.
\end{equation}
On a composite systems~$\mathcal{H}^{\otimes N}$, we write~$\rho_{(i)}$ for the state~$\rho$ acting on the Hilbert space $\mathcal{H}$ corresponding to the subsystem~$i$:
$$\rho_{(i)}:= I \otimes I \otimes \cdots \otimes \rho \otimes \cdots \otimes I,$$
where the matrix $\rho$ appears in the $i$-th position. We shall use extensively the notation ~$\ket{\Omega}_{(j, k)}$ which is the unnormalized maximally entangled state between the $j$-th and the $k$-th subsystems, that is~$\ket{\Omega}_{(j, k)} = \sum^d_{i=1} \ket{i}_{(j)} \otimes \ket{i}_{(k)}$. Also, depending on context, we shall index the tensor factors either starting from $0$ or starting from $1$. 

The evolution of mixed states (density matrices) in quantum mechanics is described by a quantum channel $T: \mathcal{M}_d \to \mathcal{M}_{d'}$. Quantum channels are linear maps which are completely positive and trace preserving. In particular they map $\mathcal{D}_d$ into $\mathcal{D}_{d'}$.

When dealing with composite systems, we shall invoke the \emph{partial trace}. Namely, if $T: \mathcal{M}_d \to {(\mathcal{M}_d)}^{\otimes N}$ is a quantum channel, we write 
\begin{equation*}
    T_i(\rho) = \Tr_{\scriptscriptstyle [\![ 1,N ]\!] \setminus \{i\}} \big{[} T(\rho) \big{]},
\end{equation*}
for all $1 \leq i \leq N$ and for all $\rho \in \mathcal{M}_d$. The operation $\Tr_{\scriptscriptstyle [\![ 1,N ]\!] \setminus \{i\}}$ is the partial trace over all copies of $\mathcal{M}_d$ except the copy number $i$. The quantum channels $T_i$ are called the \emph{marginals} of $T$. Similarly, we denote the \emph{partial transpose} of an operator~$M \in \mathcal{M}_d$, on the $i$-th copy of~$\mathcal{M}_d$, by
\begin{equation*}
    M^{\T_i}.
\end{equation*}
When the partial transpose is taken on the first copy of~$\mathcal{M}_d$, we simply write~$$M^{\To}= M^{\T_1}.$$
On the space of complex matrices, we denote by~$\langle \cdot , \cdot \rangle$ the Frobenius inner product:
\begin{equation*}
    \langle A , B \rangle = \Tr \big{[} A^* B \big{]}.
\end{equation*}

\emph{The quantum cloning problem} consists in finding a quantum channel $T: \mathcal{M}_d \to {(\mathcal{M}_d)}^{\otimes N}$ which copies any pure state into \emph{product states}. The \emph{no-cloning theorem} \cite{dieks1982communication,wootters1982single} give a negative answer to this question.
\begin{theorem}[No-cloning theorem]
    There is no quantum channel~$T : \mathcal{M}_d \to {(\mathcal{M}_d)}^{\otimes N}$ such that for all pure quantum states~$\rho \in \mathcal{D}_d$ the following holds
    \begin{equation*}
        T(\rho) = \rho^{\otimes N}.
    \end{equation*}
\end{theorem}

\emph{The quantum broadcasting problem} is a relaxation consisting in finding a quantum channel $T: \mathcal{M}_d \to {(\mathcal{M}_d)}^{\otimes N}$ such that all the \emph{marginals} are the identity. The no-cloning theorem is a corollary of the \emph{no-broadcasting} theorem \cite{barnum1996noncommuting}.
\begin{theorem}[No-broadcasting theorem]
        There is no quantum channel~$T : \mathcal{M}_d \to {(\mathcal{M}_d)}^{\otimes N}$ such that for all mixed quantum states~$\rho \in \mathcal{D}_d$, and for all~$1 \leq i \leq N$ the following holds
    \begin{equation*}
        T_i(\rho) = \rho.
    \end{equation*}
\end{theorem}

In this paper we are going to study an approximate version by allowing the marginals to be of the form
\begin{equation*}
    T_i(\rho) = p_i \! \cdot \! \rho + (1 - p_i) \frac{I}{d},
\end{equation*}
for some~$p_i \in [0,1]$. This problem is known as the \emph{approximate quantum cloning problem}.

\subsection{Quantum Fidelity}

The current section is devoted to rigorously presenting the problem we shall address in this paper. Asking the marginals $T_i$ to be the identity map means that perfect copies of states are required. The no-cloning Theorem expresses that it is not permitted. Then, it is natural to relax this ``search'' of perfection by allowing some ``imperfections''. 

The quantum cloning problem is then given by the description of figures of merit for quantum cloning maps. That is a measure of closeness between the output of the marginal $T_i(\rho)$ and the state $\rho$. We choose the \emph{quantum fidelity} to measure this distance. Recall that the quantum fidelity is defined by \cite{nielsen2010quantum}
\begin{equation*}
    F(\rho , \sigma) = \Tr^2 \bigg{(} \sqrt{\sqrt{\rho} \sigma \sqrt{\rho}} \bigg{)}
\end{equation*}
for all states $(\rho , \sigma) \in (\mathcal{D}_d)^2$. Recall that the quantum fidelity satisfies the following properties (below, $\rho, \sigma$ are density matrices, $x, y$ are unit norm vectors, and $U$ is a unitary operator)
\begin{align*}
    F(\rho , \sigma) &= F(\sigma , \rho) \\
    F(U \rho U^* , U \sigma U^*) &= F(\rho , \sigma) \\
    F \big{(} \rho , \ketbra{x}{x} \big{)} &= \Tr \big{(} \rho \ketbra{x}{x} \big{)} = \bra{x} \rho \ket{x} \\
    F \big{(} \ketbra{x}{x} , \ketbra{y}{y} \big{)} &= \braket{x}{y}^2 \\
    F(\rho , \sigma) &= 1 \Leftrightarrow \rho = \sigma.
\end{align*}
Recall also that the quantum fidelity is jointly concave \cite{nielsen2010quantum}, that is 
\begin{equation*}
    F \bigg{(} \sum^n_{i=1} p_i \cdot \rho_i , \sum^n_{i=1} p_i \cdot \sigma_i \bigg{)} \geq \sum^n_{i=1} p_i \cdot F(\rho_i , \sigma_i)
\end{equation*}
for all probability vectors $p=(p_1 , \ldots , p_n)$ and all density matrices $\rho_i, \sigma_i$ with~$1 \leq i \leq n$.

We shall concentrate on the marginal fidelities $F \big{(} \rho , T_i(\rho) \big{)}$, and we want to be as close to $1$ as possible on the $N$ marginals  according to a distribution $\alpha \in [0,1]^N$.  
More precisely, we are interested in, what we call the the \emph{Quantum cloning problem}. In details, for a quantum channel $T: \mathcal{M}_d \to {(\mathcal{M}_d)}^{\otimes N}$, let us introduce
\begin{equation*}
    F_\alpha(T) = \sum^N_{i=1} \alpha_i \cdot \inf_{\rho \in \mathcal{D}_d} F \big{(} \rho , T_i(\rho) \big{)}.
\end{equation*}
We aim to study the following optimization problem
\begin{equation} \label{eq:optim-worst}
    \sup_{T} F_\alpha(T)
\end{equation}
and find the quantum cloning map which reaches the supremum.

First, let us simplify the problem by showing that we can consider only pure states in the expression of $F_\alpha(T)$. Indeed for a density matrix $\rho$, we can consider its spectral decomposition
\begin{equation*}
    \rho = \sum^d_{i=1} v_i \cdot \ketbra{x_i}{x_i},
\end{equation*}
where $\forall i \in \llbracket d \rrbracket, v_i \geq 0$ and $\sum^d_{i=1} v_i = 1$. Using the concavity of the quantum fidelity, we have that for all quantum channels $\Phi : \mathcal{M}_d \to \mathcal{M}_d$,
\begin{align*}
    F \big{(} \rho , \Phi(\rho) \big{)} &\geq \sum^d_{i=1} v^2_i \cdot F \Big{(} \ketbra{x_i}{x_i} , \Phi \big{(} \ketbra{x_i}{x_i} \big{)} \Big{)} \\
    &\geq \min_i F \Big{(} \ketbra{x_i}{x_i} , \Phi \big{(} \ketbra{x_i}{x_i} \big{)} \Big{)}.
\end{align*}
This obviously implies that for all~$1 \leq i \leq N$
\begin{equation*}
    \inf_{\rho \in \mathcal{D}_d} F \big{(} \rho , T_i(\rho) \big{)} = \inf_{\norm{x} = 1} F \Big{(} \ketbra{x}{x} , T_i \big{(} \ketbra{x}{x} \big{)} \Big{)}.
\end{equation*}
Hence, our problem reduces to the analysis of the following quantity
\begin{equation*}
    F_\alpha(T) = \sum^N_{i=1} \alpha_i \cdot \inf_{\rho \in \mathcal{D}_d} \Tr \big{(} \rho \, T_i(\rho) \big{)},
\end{equation*}
where the infimum is taken on the set of pure states $\rho = \ketbra{x}{x}$.

When the weights $\alpha_i$ are all equal positive scalars, the problem is the so called \emph{symmetric Quantum cloning problem}. Otherwise the problem is called  \emph{Asymmetric Quantum cloning problem} which is the one addressed in this paper. In the next section we shall see that we can restrict the optimization problem \eqref{eq:optim-worst} to the particular class of $\mathcal{U}(d)$ covariant channels.

\subsection{Symmetrized quantum channels}

A particular class of quantum channel will play an important role: the \emph{symmetrized quantum channel}.

\begin{definition}
    For a quantum channel $T: \mathcal{M}_d \to {(\mathcal{M}_d)}^{\otimes N}$ we define $\widetilde{T}$ the symmetrized quantum channel
    \begin{equation*}
        \widetilde{T}(\rho) = \int_{\mathcal{U}(d)} U^{\otimes N} T \big{(} U^* \rho \; U \big{)} {\big{(} U^* \big{)}}^{\otimes N} \mathrm{d}U,
    \end{equation*}
    where the integral is taken with respected to the normalized Haar measure on the unitary group $\mathcal{U}(d)$. 
\end{definition}
In an equivalent way we can define the notion of \emph{$\mathcal{U}(d)$-covariance} for quantum channels.
\begin{definition}
    A quantum channel $T: \mathcal{M}_d \to {(\mathcal{M}_d)}^{\otimes N}$ is called $\mathcal{U}(d)$-covariant if for all $U \in \mathcal{U}(d)$
    \begin{equation*}
        T(U^* \rho \: U) = {\big{(} U^* \big{)}}^{\otimes N} T (\rho) U^{\otimes N}.
    \end{equation*}
\end{definition}
\noindent Note that $T$ is $\mathcal{U}(d)$-covariant if and only if $\widetilde{T} = T$ by the invariance property of the Haar measure.

A straightforward computation shows that the marginals of $\widetilde{T}$ are also $\mathcal{U}(d)$-covariant. Indeed for all marginals $\widetilde{T}_i$, we have
\begin{align*}
    \widetilde{T}_i(\rho) &= \Tr_{\scriptscriptstyle \llbracket N \rrbracket \setminus \{i\}} \Bigg{[} \int_{\mathcal{U}(d)} U^{\otimes N} T \big{(} U^* \rho \: U \big{)} {\big{(} U^* \big{)}}^{\otimes N} \mathrm{d}U \Bigg{]} \\
    &= \int_{\mathcal{U}(d)} U \: Tr_{\scriptscriptstyle \llbracket N \rrbracket \setminus \{i\}} \Big{[} T \big{(} U^* \rho \: U \big{)}  \Big{]} U^* \mathrm{d}U \\
    &= \int_{\mathcal{U}(d)} U \: T_i \big{(} U^* \rho \: U \big{)} \, U^* \mathrm{d}U.
\end{align*}
The marginals of $\widetilde{T}$ are symmetrized, and therefore they are also $\mathcal{U}(d)$-covariant. Now using the jointly concavity of the quantum fidelity we have for all marginals $\widetilde{T}_i$
\begin{equation*}
    \inf_{\rho \in \mathcal{D}_d} F \big{(} \rho , \widetilde{T}_i(\rho) \big{)} \geq \inf_{\rho \in \mathcal{D}_d} F \big{(} \rho , T_i(\rho) \big{)}.
\end{equation*}
where the inﬁma are taken over the pure states~$\rho$. Hence, we have
\begin{equation*}
    \sup_{\widetilde{T}} F_\alpha \big{(} \widetilde{T} \big{)} \geq \sup_T F_\alpha (T).
\end{equation*}
Our optimization problem can be thus be restricted to the simpler set of $\mathcal{U}(d)$-covariant quantum channels.

\subsection{Average vs.~worst fidelity}

The usual \emph{quantum cloning average problem} is defined as follows. For a quantum channel $T: \mathcal{M}_d \to {(\mathcal{M}_d)}^{\otimes N}$ and for a distribution $\alpha\in [0,1]^N$, we define
\begin{equation*}
    \bar{F}_\alpha(T) = \sum^N_{i=1} \alpha_i \E_{\rho \in \mathcal{D}_d} \Big{[} F \big{(} \rho , T_i(\rho) \big{)} \Big{]}
\end{equation*}
where the expectation is taken on the pure states~$\rho$, and
\begin{equation*}
    \E_{\rho \in \mathcal{D}_d} \big{[} f(\rho) \big{]} = \int f \big{(} \ketbra{x}{x} \big{)} \mathrm{d}\nu(x),
\end{equation*}
where $\nu$ is the uniform measure on the set of pure states, that is the (normalized) Lebesgue measure on the unit sphere of $\mathbb C^d$. We have thus restated the question as an optimization problem
\begin{equation*}
    \sup_T \bar{F}_\alpha(T).
\end{equation*}
As a first sight, it is not clear why the average and the worst case (see \eqref{eq:optim-worst}) problems coincide. This relies on the description of which marginals will be involved by solving the quantum cloning problem. In particular using the $\mathcal{U}(d)$-covariance, we shall see later that all the marginals~$T_i$ will be of the form
\begin{equation*}
    T_i(\rho) = p_i \! \cdot \! \rho + (1 - p_i) \frac{I}{d},
\end{equation*}
where the measure of closeness is given by the $p_i \in \mathbb{R}$.
 
 \begin{proposition}
    The worst case problem and the average case problem are equivalent. More precisely, for all $\mathcal{U}(d)$-covariant quantum channel $T : \mathcal{M}_d \to {(\mathcal{M}_d)}^{\otimes N}$ and for all distributions $\alpha \in [0,1]^N$ we have
    \begin{equation*}
        \sup_T \bar{F}_\alpha(T) = \sup_T F_\alpha(T).
    \end{equation*}
 \end{proposition}
 \begin{proof}
    As announced, due to the $\mathcal{U}(d)$-covariance, we shall see in Proposition \ref{prop:decompChoi} that the involved marginals of $T : \mathcal{M}_d \to {(\mathcal{M}_d)}^{\otimes N}$ will be of the form
    \begin{equation*}
        T_i(\rho) = p_i \! \cdot \! \rho + (1 - p_i) \frac{I}{d},
    \end{equation*}
    for all pure states~$\rho$, for some $p_i \in \mathbb{R}$ that does not depend on $\rho$. This way, the figures of merit $F \big{(} \rho , T_i(\rho) \big{)}$ of the marginals, on pure states $\rho = \ketbra{x}{x}$, are given by
    \begin{align*}
        F \big{(} \rho , T_i(\rho) \big{)} &= \bra{x} T_i(\rho) \ket{x} \\
        &= \Big{\langle} x \Big{|} p_i \! \cdot \! \rho + (1 - p_i) \frac{I}{d} \Big{|} x \Big{\rangle} \\
        &= p_i + \frac{1 - p_i}{d}.
    \end{align*}
    In particular they do not depend on the state. It is then clear that for pure states~$\rho$
    \begin{equation*}
        \inf_{\rho \in \mathcal{D}_d} F \big{(} \rho , T_i(\rho) \big{)} = \E_{\rho \in \mathcal{D}_d} \Big{[} F \big{(} \rho , T_i(\rho) \big{)} \Big{]},
    \end{equation*}
    which justifies the equivalence of the worst case and average problem.
\end{proof}

In the rest of the paper all quantum channels $\boldsymbol{T : \mathcal{M}_d \to {(\mathcal{M}_d)}^{\otimes N}}$ are considered to be $\boldsymbol{\mathcal{U}(d)}$-covariant.

\section{Schur-Weyl duality \& Permutation operators}\label{sec:permutation-operators}

This section is divided into two main parts which are about the study of permutations operators. On the one hand, we relate the Choi matrix of a quantum cloner to some particular sum of permutation operators. On the other hand we study the maximal eigenvalues and corresponding eigenvectors of such permutation operators. The crucial tool used here is the so-called \emph{Schur-Weyl duality} \cite{fulton2013representation}. More precisely, we shall use two important ingredients of this duality. The first ingredient concerns the fact that an operator which commutes with all operators of the form $U\otimes\ldots\otimes U$, where $U$ is an arbitrary unitary operator, is a sum of permutation operators. The second ingredient is the decomposition of $\mathcal H^{\otimes n}$ into a direct orthogonal sum of irreducible representations of $\mathfrak{S}_n$.

To be more precise, let $C_T$ be the \emph{Choi matrix} of a quantum channel $T : \mathcal{M}_d \to {(\mathcal{M}_d)}^{\otimes N}$, defined by
\begin{align*}
    C_T &= (\id \otimes \, T) \bigg{(} \sum_{i,j=1}^d \ketbra{i}{j} \otimes \ketbra{i}{j} \bigg{)} \\
    &= (\id \otimes \, T) \ketbra{\Omega}{\Omega},
\end{align*}
where $\ket \Omega$ was defined in \eqref{eq:def-omega}. In particular, we recover $T$ by the formula 
$$\forall \rho \qquad T(\rho) = \Tr_0 \Big{[} C_T \big{(} \rho^\T \otimes I^{\otimes N} \big{)} \big{]}.$$
Recall that a quantum channel $T$ is a completely positive and trace-preserving map. In the Choi matrix form, these conditions become (recall that the Choi matrix is an operator acting on $N+1$ copies of $\mathbb C^d$, indexed by the integers $0,1,2, \ldots, N$, with the index 0 corresponding to the input space) \cite{watrous2018theory}:
\begin{equation} \label{eq:ChoiConditions}
    C_T \geq 0 \quad \text{ and } \quad  \Tr_{\scriptscriptstyle \llbracket 0,N \rrbracket \setminus \{0\}} (C_T) = I.
\end{equation}
Our strategy is then the following. Firstly, by exploiting the $\mathcal U(d)$-covariance of quantum cloners, we shall give the form of admissible corresponding Choi matrices. Secondly, we study the location of the maximal eigenvalues of particular sum of permutations operator which appear in our optimization procedure from Section \ref{subsec:max}

\subsection{Choi matrices of quantum cloners}

In this subsection, we shall use the following version of Schur-Weyl duality \cite{fulton2013representation}. 
\begin{theorem}[Schur–Weyl]
    Let $\mathcal{H}$ be a finite $d$-dimensional Hilbert space, and $T$ a linear map acting on $\mathcal{H}^{\otimes n}$. If $[T , U^{\otimes n}] = 0$ holds for all unitaries $U \in \mathcal{U}(d)$, then $T$ is a linear combination of permutation operators $\Pi_\sigma$ on $\mathcal{H}^{\otimes n}$:
    \begin{equation*}
        T = \sum_{\sigma \in \mathfrak{S}_n} \alpha_\sigma \cdot \Pi_\sigma,
    \end{equation*}
    where $\mathfrak{S}_n$ is the symmetric group on $n$ elements and $\Pi_\sigma$ is the representation defined by
    \begin{equation}\label{eq:action-Sn}
        \Pi_\sigma(v_1 \otimes \cdots \otimes v_n) = v_{\sigma^{\shortminus 1}(1)} \otimes \cdots \otimes v_{\sigma^{\shortminus 1}(n)}.
    \end{equation}
\end{theorem}

As announced we shall combine $\mathcal U(d)$-covariance with the above theorem to obtain the generic form of $\mathcal U(d)$-covariant cloning. This is the content of the following proposition.

\begin{proposition} \label{prop:decompChoi}
    The Choi matrix~$C_T$ of a $\mathcal{U}(d)$-covariant quantum channel $T : \mathcal{M}_d \to {(\mathcal{M}_d)}^{\otimes N}$ is linear combination of partially transposed permutation operators~$\Pi_\sigma$ on $\mathcal{H}^{\otimes (N + 1)}_d$. More precisely, for all $\sigma \in \mathfrak{S}_{N + 1}$ there exists $\beta_\sigma \in \mathbb{C}$ such that
    \begin{equation} \label{eq:decompChoi}
        C_T = \sum_{\sigma \in \mathfrak{S}_{N + 1}} \beta_\sigma \cdot \Pi^{\To}_\sigma
    \end{equation}
    Then for all $1 \leq i \leq N$ there exists $p_i \in \mathbb{R}$ such that the marginals become
    \begin{equation*}
        T_i(\rho) = p_i \! \cdot \!\rho + (1 - p_i) \frac{I}{d},
    \end{equation*}
    for all pure states $\rho$.
\end{proposition}

\begin{proof}
From the $\mathcal{U}(d)$-covariance of $T$ and the formula $\big{(} \bar{U} \otimes U \big{)} \ket{\Omega} = \ket{\Omega}$ we have for all unitaries $U \in \mathcal{U}(d)$:
\begin{align*}
    C_T \big{(} \bar{U} \otimes U^{\otimes N} \big{)} &= \big{[} (\id \otimes \, T) \ketbra{\Omega}{\Omega} \big{]} \big{(} \bar{U} \otimes U^{\otimes N} \big{)} \\
    &= \Big{[} (\id \otimes \, T) \big{(} \bar{U}\otimes U \big{)} \ketbra{\Omega}{\Omega} \big{(} U^\T \otimes \, U^* \big{)} \Big{]} \big{(} \bar{U}\otimes U^{\otimes N} \big{)} \\
    &= \big{(} \bar{U} \otimes U^{\otimes N} \big{)} \big{[} (\id \otimes \, T) \ketbra{\Omega}{\Omega} \big{]} \Big{(} U^\T \otimes \, {\big{(} U^* \big{)}}^{\otimes N} \Big{)} \big{(} \bar{U} \otimes U^{\otimes N} \big{)} \\
    &= \big{(} \bar{U}\otimes U^{\otimes N} \big{)} C_T
\end{align*}
Therefore $\big{[} C_T , \bar{U} \otimes U^{\otimes N} \big{]} = 0$, which is equivalent to $\big{[} C^{\To}_T,U \otimes U^{\otimes N} \big{]} = 0$. Then the decomposition \eqref{eq:decompChoi} is a direct consequence of the fact that $ C^{\To}_T$ commutes with all $U^{\otimes (N + 1)}$ and the Schur-Weyl duality. This holds for all Choi matrices coming from $\mathcal{U}(d)$-covariant channels, which is the case for all marginals $T_i$. In this particular case, the two permutations of~$\mathfrak{S}_2$ are
    \begin{equation*}
        (1)(2) \quad \text{ and } \quad (1 \: 2)
    \end{equation*}
    and their partially transposed permutation operators correspond to the identity and the \emph{maximally depolarizing channels} defined on all~$\rho \in \mathcal{D}_d$ by:
    \begin{equation*}
        \Lambda_\lambda (\rho) = (1 - \lambda)\rho + \lambda \frac{I}{d},
    \end{equation*}
    with~$0 \leq \lambda \leq 1 + \frac{1}{d^2 - 1}$.
    Hence, the marginals $T_i$ are of the form:~$T_i(\rho) = p_i \! \cdot \!\rho + (1 - p_i) \frac{I}{d}$.
\end{proof}

\begin{remark}
    Note that conversely, a linear combination of partially transposed permutation operators is a Choi matrix of a quantum channel if the conditions of Eq.~\eqref{eq:ChoiConditions} hold.
\end{remark}

In the following, we recall that the permutations of the group $\mathfrak{S}_{n+1}$ will start from $0$, i.e. $\sigma \in \mathfrak{S}_{n+1}$ is a permutation of $\{ 0 , 1, \ldots , n - 1, n \}$. Study of permutations operators and their actions on a particular set of vectors wil be crucial in the following. In particular we shall concentrate on a particular operator
\begin{equation} \label{eq:S-op}
    S_x := \sum^N_{i = 1} x_i \cdot \Pi^{\To}_{(0, i)},
\end{equation}
where $x=(x_i)_{i=1}^N\in (\mathbb R^*_+)^N$. The relevance of this operator shall be justified in Section \ref{sec:optimal-quantum-cloning} as it appears in the optimization procedure on fidelities.

\subsection{Permutation operators}

In this subsection, in order to illustrate some examples and lemmas, we shall use a graphical calculus for tensors, derived from the Penrose graphical notation~\cite{penrose1971applications}. This helps to make computations faster and illuminate some otherwise abstract equations. Appendix \ref{app:permutation-operators} is dedicated to many examples illustrated by graphical calculus.
Similar graphical calculi have been developed in the context of tensor network state or categorical quantum information theory, which can be found in~\cite{wood2011tensor, coecke2018picturing, duncan2020graph}. See~\cite[Section 4]{nechita2021geometrical} for a complete description. In particular the permutations are read from right to left.
\bigskip

Note than in $S_x$ appears only transpositions of the form $(0,i)$. It is worth noticing that such permutations belong to $\big{\{} \sigma \in \mathfrak{S}_{N + 1} \; \big{|} \; \sigma(0) \neq 0 \big{\}}$. As it is discussed in Appendix \ref{app:permutation-operators} only such permutations contributes non trivially to the quantum cloning problem. In the sequel we collect information of the permutations of this set and then we plug them into $S$.

A permutation $\sigma \in \big{\{} \sigma \in \mathfrak{S}_{N + 1} \; \big{|} \; \sigma(0) \neq 0 \big{\}}$ can be parametrized by a couple $1 \leq a,b \leq N$ such that $\sigma(0) = a$ and $\sigma(b) = 0$. This gives a partition of
\begin{equation*}
    \big{\{} \sigma \in \mathfrak{S}_{N + 1} \; \big{|} \; \sigma(0) \neq 0 \big{\}} = \bigcup_{1 \leq a,b \leq N} \Sigma_{a,b}    
\end{equation*}
where 
\begin{equation}\label{eq:def-Simga-ab}
    \Sigma_{a,b} = \big{\{} \sigma \in \mathfrak{S}_{N + 1} \; \big{|} \; \sigma(0) = a \text{ and } \sigma(b) = 0 \big{\}}.
\end{equation}
Each set $\Sigma_{a,b}$ contains $(N - 1)!$ permutations. For $1 \leq a,b \leq N$, a generic element of $\Sigma_{a,b}$ will be denoted by $\sigma_{a,b}$. 

Before expressing the relevant result on these permutations, we need to introduce some notations. For two integers $i,j$ we denote by $(i \, : \, j)$ the permutation $\mathfrak{S}_{N-1}$ defined by

\begin{equation*}
    (i \, : \, j) =
    \begin{cases}
        \big{(} (i - 1) \: (i - 2) \: \cdots \: j \big{)} &\text{if } 0 \leq j < i \leq N-1 \\
        \big{(} i \: (i + 1) \: \cdots \: (j - 1) \big{)} &\text{if } 0 \leq i < j \leq N-1 \\
        (0)(1) \cdots (N - 2) &\text{otherwise}
    \end{cases}.
\end{equation*}
Let $1 \leq a,b,c \leq N$ and $\sigma$ in $\mathfrak{S}_{N-1}$, we shall also need the following permutations operators
 \begin{equation*}
        \Pi_{\sigma(a,b,c)} =
        \begin{cases}
            \Pi_{(0 \, : \, (a - 1))} \circ \Pi_{\sigma} \circ \Pi_{((b - 1) \, : \, 0)} &\text{if } b = c \\
            \Pi_{(0 \, : \, (a - 1))} \circ \Pi_{\sigma} \circ \Pi_{((b - 1) \, : \, 0)} \circ \Pi_{((c - 1) \, : \, (b - 1))} &\text{if } b \neq c.
        \end{cases}
    \end{equation*}

The following Lemma gives a similar tractable expression of all the permutation operators $\Pi^{\To}_\sigma$, and two sets which are useful to seek the eigenvectors of the permutation operators.
\begin{lemma}\label{lem:struct}
     Let $1 \leq a,b,c \leq N$, and let $\sigma \in \Sigma_{a,b}$, then
     
     \begin{itemize}
         
     \item There exists a unique $\hat{\sigma} \in \mathfrak{S}_{N - 1}$ such as the partially transposed permutation operator $\sigma^{\To}$ is
    \begin{equation*}
        \Pi_\sigma^{\To} = \Pi_{(1 \: a)} \: (\omega_{(0,1)} \otimes \Pi_{\hat{\sigma}}) \; \Pi_{(1 \: b)}.
    \end{equation*}
    
    \item Let $\phi \in \mathcal{H}^{\otimes (N - 1)}$, then
    \begin{equation*}
        \Pi_\sigma^{\To}(\ket{\Omega}_{(0,c)} \otimes \ket{\phi}) =
        \begin{cases}
            d \cdot \ket{\Omega}_{(0,a)} \otimes \Pi_{\hat{\sigma}(a,b,b)} \ket{\phi} &\text{if } b = c \\[1em]
            \ket{\Omega}_{(0,a)} \otimes \Pi_{\hat{\sigma}(a,b,c)} \ket{\phi} &\text{if } b \neq c.
        \end{cases}
    \end{equation*}
    
    \item In particular
    \begin{align*}
	    \Ima \Pi_\sigma^{\To} &= \Span \Big{\{} \ket{\Omega}_{(0,a)} \otimes \ket{\phi} \; \Big{|} \; \phi \in \mathcal{H}^{\otimes (N - 1)} \Big{\}} \\
	    \Ker \Pi_\sigma^{\To} &\supseteq \Span {\Big{\{} \ket{\Omega}_{(0,k)} \otimes \ket{\phi} \; \Big{|} \; 1 \leq k \leq N, \; \phi \in \mathcal{H}^{\otimes (N - 1)} \Big{\}}}^\perp.
    \end{align*}
    \end{itemize}
\end{lemma}
The proof of this lemma, which is rather lengthy and technical, as well as the graphical formalism used to derive it, is postponed to Appendix~\ref{sec:app-proof-lemma-permutation-operators}, see also Examples~\ref{ex:graphical1}, \ref{ex:graphical2} and \ref{ex:graphical3}. Note that for any $1 \leq i \leq N$, the partially transposed permutation operator associated to the permutation $(0 \: i) \in \Sigma_{i,i}$ is
\begin{equation*}
    \Pi^{\To}_{(0 \: i)} = \omega_{(0, i)} \otimes I^{\otimes (N - 1)}.
\end{equation*}
Let us comment on this lemma. The action of $\Pi_\sigma^{\To}$ onto vectors of the form $\ket{\Omega}_{(0,c)} \otimes \ket{\phi})$ makes appear a permutation operator on $\mathcal H^{\otimes N-1}$, which acts on $\ket{\phi}$. Since this permutation is a representation operator, it is then natural to involve the decomposition of $\mathcal H^{\otimes N-1}$ into irreducible representations in order to make more precise the action of $\Pi_\sigma^{\To}$. This shall be made clear in the next section.

\subsection{Maximal eigenvalue}\label{subsec:max}
This section is devoted to the analysis of the maximal eigenvalue of the operator
\begin{equation} 
    S_x = \sum^N_{i = 1} x_i \cdot \Pi^{\To}_{(0, i)},
\end{equation}
where~$x=(x_i)_{i=1}^N \in (0,\infty)^N$ (in this section, the vector $x$ is fixed and we do not write the subscript $x$). The study of the maximum eigenvalue and the location of corresponding eigenvectors is crucial in the main result of the paper. As announced, due to the structure of the~$\Pi^\To_{(0 \: i)}$, where permutations on~$\mathfrak{S}_{N - 1}$ appear, we shall consider the Schur-Weyl decomposition \cite{fulton2013representation} 
\begin{equation*}
    \mathcal{H}^{\otimes(N-1)} = \bigoplus_{\lambda\,\, \textrm{irrep}\,\,\mathfrak{S}_{N-1}} E^{\lambda}
\end{equation*}
In the sequel, given a permutation $\sigma \in \mathfrak{S}_{N - 1}$, we consider the decomposition of~$\Pi_\sigma$ with respect to this decomposition. Note that this decomposition is orthogonal and all the direct sum in the paper will be considered orthogonal. In particular we denote by $P^\lambda$ the orthogonal projection into the irreducible representation $\lambda$ of~$\mathfrak{S}_{N - 1}$. For $\sigma\in\mathfrak{S}_{N-1}$ note that \begin{equation*}
    \Pi_\sigma^{\lambda}:=P^\lambda \Pi_\sigma P^\lambda = P^\lambda \Pi_\sigma = \Pi_\sigma P^\lambda,
\end{equation*}
is a unitary operator.

Let~$\lambda$ be an irreducible representation of~$\mathfrak{S}_{N - 1}$, and let~${\big{\{} \mathbf{v}^{\lambda}_i \big{\}}}_i$ be an orthonormal basis of the irreducible subspace $E^\lambda$ of~$\mathfrak{S}_{N - 1}$ associated to~$\lambda$. Associated with this orthonormal basis, we introduce the following set of vectors
\begin{equation*}
   \mathcal{V}^{\lambda} = \Big{\{} \ket{\Omega}_{(0,i)} \otimes \ket{\mathbf{v}^{\lambda}_j} \; \Big{|} \; 1 \leq i \leq N, \quad 1 \leq j \leq \dim(E^\lambda) \Big{\}}.
\end{equation*}
The scalar products of the elements of $\mathcal{V}^\lambda$ are determined by the following lemma whose proof is postponed to Appendix \ref{app:scalarproduct}.
\begin{lemma}\label{lem:scalarProduct} 
    Let~$\lambda_1, \lambda_2$ be two distinct irreducible representations of~$\mathfrak{S}_{N - 1}$ on $\mathcal H^{\otimes(N-1)}$, let~$\mathbf{v}_1, \mathbf{v}_2$ be two normalized vectors in the irreducible subspace~$\lambda_1$ and~$\lambda_2$ of~$\mathfrak{S}_{N - 1}$, and let~$1 \leq k,l \leq N$. Then
    \begin{equation*}
        \braket{\Omega_{(0,k)} \otimes\mathbf{v}_1}{\Omega_{(0,l)} \otimes\mathbf{v}_2 } =
        \begin{cases}
            d \cdot \braket{\mathbf{v}_1}{\mathbf{v}_2} &\text{if } k = l \\
            \bra{\mathbf{v}_1} \Pi_{((l - 1) \, : \, (k - 1))} \ket{\mathbf{v}_2} &\text{if } k \neq l.
        \end{cases}
    \end{equation*}
    As a consequence
    \begin{align*}
        \braket{\Omega_{(0,k)} \otimes\mathbf{v}_1}{\Omega_{(0,k)} \otimes\mathbf{v}_1} &= d \\
        \braket{\Omega_{(0,k)} \otimes\mathbf{v}_1}{\Omega_{(0,l)} \otimes\mathbf{v}_1} &= \bra{\mathbf{v}_1} \Pi_{((l - 1) \, : \, (k - 1))} \ket{\mathbf{v}_1},
    \end{align*}
    and
    \begin{equation*}
        \braket{\Omega_{(0,k)} \otimes\mathbf{v}_1}{\Omega_{(0,l)} \otimes\mathbf{v}_2} = 0.
    \end{equation*}
    In particular the vector spaces $\Span \mathcal{V}^\lambda$ are in orthogonal direct sum.
\end{lemma}

In order to seek the maximal eigenvalue of~$S_x$ and the corresponding eigenvector, we shall use the following lemma, which links the range of~$S_x$ and the vector space generated by this~$v^\lambda$.
\begin{lemma} \label{lem:blockEigenvalue}
    We have the following inclusion for $\Ima S_x$ and $\Ker S_x$
    \begin{align*}
        \Ima S_x &\subseteq \bigoplus_\lambda \Span (\mathcal{V}^{\lambda}) \\
        \Ker S_x &\supseteq \Big(\Span (\mathcal{V}^{\lambda}, \lambda )\Big)^\perp 
    \end{align*}
    In particular we will consider the decomposition 
    $$\mathcal H^{\otimes N-1} = \bigoplus_\lambda \Span (\mathcal{V}^{\lambda}) \bigoplus \left(\Span (\mathcal{V}^{\lambda}, \lambda )\right)^\perp .$$
\end{lemma}
\begin{proof}
    The fact that the vector spaces $\Span \mathcal{V}^{\lambda}$ are in orthogonal direct sum is a consequence of Lemma~\ref{lem:scalarProduct}. For the first inclusion, since~$S_x = \sum^N_{i = 1} x_i \cdot \Pi^{\To}_{(0, i)}$ we have
    \begin{equation*}
        \Ima S_x \subseteq \sum^N_{i = 1} \Ima \Pi^{\To}_{(0, i)}.
    \end{equation*}
    Following Lemma~\ref{lem:struct}, for all~$1 \leq i \leq N$,
    \begin{equation*}
        \Ima \Pi^{\To}_{(0, i)} \subseteq \bigoplus_\lambda \Span (\mathcal{V}^{\lambda}).
    \end{equation*}
    Hence, $\Ima S_x \subseteq \bigoplus_\lambda \Span (\mathcal{V}^{\lambda})$. Now we have $S_x = S_x^*$ which yields $\Ker S_x = (\Ima S_x^*)^\perp = (\Ima S_x)^\perp$, then
    \begin{equation*}
        \Ker S_x \supseteq \Span (\mathcal{V}^{\lambda}, \lambda, )^\perp.
    \end{equation*}
    The last point is obvious.
\end{proof}

\begin{remark}
    From Lemma \ref{lem:struct} we have that, for all irrep $\lambda$,
    \begin{equation*}
        \Pi_\sigma^{\To}(\ket{\Omega}_{(0,c)} \otimes \ket{\mathbf{v}^\lambda}) =
        \begin{cases}
            d \cdot \ket{\Omega}_{(0,a)} \otimes \Pi_{\hat{\sigma}(a,b,b)} \ket{\mathbf{v}^\lambda} &\text{if } b = c \\[1em]
            \ket{\Omega}_{(0,a)} \otimes \Pi_{\hat{\sigma}(a,b,c)} \ket{\mathbf{v}^\lambda} &\text{if } b \neq c.
        \end{cases}
    \end{equation*}
    Since $\Pi_{\hat{\sigma}(a,b,c)} \ket{\mathbf{v}^\lambda} = \Pi_{\hat{\sigma}(a,b,c)} P^\lambda \ket{\mathbf{v}^\lambda} = P^\lambda \Pi_{\hat{\sigma}(a,b,c)} P^\lambda \ket{\mathbf{v}^\lambda}= \Pi^\lambda_{\hat{\sigma}(a,b,c)} \ket{\mathbf{v}^\lambda}$, we see that $\Span \mathcal{V}^{\lambda}$ is stabilized by $\Pi^\To_\sigma$ for all $\sigma$. This way, Lemma \ref{lem:blockEigenvalue} expresses the fact that $S_x$ can be block diagonalized with respect to the orthogonal decomposition $\bigoplus_\lambda \Span (\mathcal{V}^{\lambda}) \bigoplus \left(\Span (\mathcal{V}^{\lambda}, \lambda )\right)^\perp$.
    
    As a result of this block diagonalization, for any non-zero eigenvalue of $S$ one can associate an eigenvector of the form~$\chi = \sum_{ij} \beta_{ij} \cdot \ket{\Omega}_{(0,i)} \otimes \ket{\mathbf{v}^{\lambda}_j}$, for some coefficients~$\beta_{ij}$. Note that vectors of this form can also be in the kernel of $S$.
\end{remark}

At this stage, in order to localize the maximal eigenvalue of $S$ one has to study $S_{| \mathcal{V}^\lambda}$ which stands for the restriction of $S$ to the space $ \Span \mathcal{V}^{\lambda}$. Actually this is not straightforward since $\mathcal{V}^{\lambda}$ is not a basis of $\Span \mathcal{V}^{\lambda}$ in general. Nevertheless for the trivial representation $\lambda = (N-1)$ of $\mathfrak{S}_{N - 1}$, the corresponding set, which we denote by $\mathcal{V}^+$, is a basis and we can expressed the action of $S$ on this basis. This is the content of the following proposition.

\begin{proposition}
    The vectors of $\mathcal{V}^+$ are linearly independent. Let us define
    \begin{equation*}
        S_{+} :=
        \begin{pmatrix}
            d \, x_1 & x_1 & x_1 & \ldots & x_1 \\
            x_2 & d \, x_2 & x_2 & \ldots & x_2 \\
            \vdots & \vdots & \vdots & \ddots & \vdots \\
            x_N & x_N & x_N & \ldots & d \, x_N
        \end{pmatrix}.
    \end{equation*}
    Then the restriction of the operator~$S$, defined in Eq.~\eqref{eq:S-op}, onto the vectors~$\mathcal{V}^+$ can be written in block diagonal form
    \begin{equation*}
        S_{|_{\mathcal{V}^+}} =
        S_+^{\oplus \dim \big{(} \vee^{(N - 1)}(\mathcal{H}) \big{)}}.
    \end{equation*}
\end{proposition}

\begin{proof}
    From Lemma~\ref{lem:scalarProduct}, for any~$1 \leq i, j \leq N$ and any distinct~$\mathbf{v}, \mathbf{w}$ in the orthogonal basis of the irreducible subspace of~$\mathfrak{S}_{N - 1}$ associated to the trivial representation, we have the following relation, which holds for any orthonormal vectors $\mathbf v$ and $\mathbf w$ on $\vee^{(N - 1)}(\mathcal{H})$:
     
    \begin{align*}
        \braket{\Omega_{(0,i)} \otimes\mathbf{v}}{\Omega_{(0,j)} \otimes\mathbf{v}} = 
        \begin{cases}
            d &\text{if } i = j \\
            1 &\text{if } i \neq j
        \end{cases},\quad
        \braket{\Omega_{(0,i)} \otimes\mathbf{v}}{\Omega_{(0,j)} \otimes\mathbf{w}} = 0.
    \end{align*}
    Then the Gram matrix~$\tilde G$ of the vectors~$\mathcal V^+$ is block diagonal, where each block is the~$N \times N$ matrix
    \begin{equation*}
       G= \begin{pmatrix}
            d & 1 & \cdots & 1 \\
            1 & d & \cdots & 1 \\
            \vdots & \vdots & \ddots & \vdots \\
            1 & 1 & \cdots & d
        \end{pmatrix}.
    \end{equation*}
    That is
    \begin{equation*}
        \tilde G =
        G^{\oplus \dim \big{(} \vee^{(N - 1)}(\mathcal{H}) \big{)}}.
    \end{equation*}
    Note that each block is invertible. Indeed for all  $v=(v_1,\ldots,v_N)$, we have
    $$\langle v,Gv\rangle=(d-1)\Vert v\Vert^2+\left(\sum_{k=1}^N v_k\right)^2$$
    which clearly shows that $G$ is positive definite since $d>1$. This shows that~$\mathcal{V}^+$ is a linearly independent set.

    Now note that for all $\mathbf v\in\vee^{(N - 1)}(\mathcal{H})$ we have $\Pi_{\hat\sigma}\ket{\mathbf v}=\mathbf{ v}.$ This way, in the basis $\mathcal{V}^+$ of Span $\mathcal{V}^+$, from Lemma \ref{lem:struct} we have for all $i=1,\ldots,N$ and for all $\mathbf{v}\in\vee^{(N - 1)}(\mathcal{H})$
    \begin{equation*}
        \Pi_{(0,i)}^{\To}(\ket{\Omega}_{(0,c)} \otimes \ket{\mathbf{v}}) =
        \begin{cases}
            d \cdot \ket{\Omega}_{(0,i)} \otimes  \ket{\mathbf{v}}= &\text{if } b = c \\[1em]
            \ket{\Omega}_{(0,i)} \otimes \Pi_{\hat{\sigma}(a,b,c)} \ket{\mathbf{v}}=\ket{\Omega}_{(0,i)} \otimes  \ket{\mathbf{v}} &\text{if } b \neq c.
        \end{cases}
    \end{equation*}

Now fix $i\in \llbracket \dim(\vee^{(N - 1)}(\mathcal{H})) \rrbracket$ and consider $\mathbf v_i$ in the orthonormal basis $\vee^{(N - 1)}(\mathcal{H})$, the linear independent set of vectors $\{\ket{\Omega}_{(0,j)} \otimes \ket{\mathbf{v}_i},j=1,\ldots,N\}$ is stable under the action of $S$ and the corresponding matrix is $S^+$. Then reading $\mathcal{V}^+=\{\ket{\Omega}_{(0,1)} \otimes \ket{\mathbf{v}_1},\ldots,\ket{\Omega}_{(0,N)} \otimes \ket{\mathbf{v}_1},\ket{\Omega}_{(0,1)} \otimes \ket{\mathbf{v}_2},\ldots,\ket{\Omega}_{(0,N)} \otimes \ket{\mathbf{v}_2},\ldots\}$, we obtain the block decomposition.
\end{proof}
\begin{remark}
    With a change of basis i.e fix first $\ket{\Omega}_{(0,i)}$ and let $\mathbf v_i$ runs onto the orthonormal basis of $\vee^{(N - 1)}(\mathcal{H})$ , the operator~$S_{|_{\mathcal V^+}}$ can also be written
    \begin{equation*}
        S_{|_{\mathcal{V}^+}} =
        \begin{pmatrix}
            d \, x_1 \cdot I & x_1 \cdot I & \ldots & x_1 \cdot I \\
            x_2 \cdot I & d \, x_2 \cdot I & \ldots & x_2 \cdot I \\
            \vdots & \vdots  & \ddots & \vdots \\
            x_N \cdot I & x_N \cdot I & \ldots & d \, x_N \cdot I
        \end{pmatrix},
    \end{equation*}
    where each block~$I$ is the identity operator of the space~$\vee^{(N - 1)}(\mathcal{H}) \big{)}$.
\end{remark}
Even if $\mathcal{V}^{\lambda}$ is not necessarily a system of linearly independent vectors, one can still define a matrix $\tilde{S}_{| \mathcal{V}^\lambda}$ which describes the action of $S$ onto the $\Span V^\lambda$.
    To this end, let~$\lambda$ be any irreducible representation of~$\mathfrak{S}_{N - 1}$, and let~$\ket{\Omega}_{(0,l)} \otimes \ket{\mathbf{v}^\lambda}$ be a vector in~$V^{\lambda}$. Then for any~$1 \leq k \leq N$ the action of the operator~$\Pi^{\To}_{(0 \: k)}$ on the vector~$\ket{\Omega}_{(0,l)} \otimes \ket{\mathbf{v}^\lambda}$ is given by
    \begin{equation*}
        \Pi^{\To}_{(0 \: k)} \big{(} \ket{\Omega}_{(0,l)} \otimes \ket{\mathbf{v}^\lambda} \big{)} =
        \begin{cases}
            d \cdot \ket{\Omega}_{(1,k)} \otimes \ket{\mathbf{v}^\lambda} &\text{ if } k = l\\
            \ket{\Omega}_{(1,k)} \otimes \Pi^\lambda_{((l - 1) : (k - 1))} \ket{\mathbf{v}^\lambda} &\text{ otherwise}.
        \end{cases}
    \end{equation*}
Indeed for transposition $(0,k)$, the associated  ~$\hat{\sigma} = \id$ and~$\Pi_{(0 \, : \, (k - 1))} \circ \Pi_{\id} \circ \Pi_{((k - 1) \, : \, 0)} = \Pi_{\id}$, we have
    \begin{equation*}
        \Pi_{\id,k,k,l} =
        \begin{cases}
            \Pi_{\id} &\text{if } k = l \\
            \Pi_{((l - 1) \, : \, (k - 1))} &\text{if } k \neq l.
        \end{cases}
    \end{equation*}

    Now at this stage let us introduce $\tilde{S}_{|_{\scriptscriptstyle\mathcal{V}^{\lambda}}}$:
    \begin{equation*}
        \tilde{S}_{|_{\scriptscriptstyle\mathcal{V}^{\lambda}}} =
        \begin{pmatrix}
            d \, x_1 \cdot \Pi^\lambda_{\id} & x_1 \cdot \Pi^\lambda_{((2 - 1) : (1 - 1))} & \ldots & x_1 \cdot \Pi^\lambda_{((N - 1) : (1 - 1))} \\
            x_2 \cdot \Pi^\lambda_{((1 - 1) : (2 - 1))} & d \, x_2 \cdot \Pi^\lambda_{\id} & \ldots & x_2 \cdot \Pi^\lambda_{((N - 1) : (2 - 1))} \\
            \vdots & \vdots & \ddots & \vdots \\
            x_N \cdot \Pi^\lambda_{((1 - 1) : (N - 1))} & x_N \cdot \Pi^\lambda_{((2 - 1) : (N - 1))} & \ldots & d \, x_N \cdot \Pi^\lambda_{\id}
        \end{pmatrix}.
    \end{equation*}
    This matrix would be exactly the matrix of $S_{|_{\scriptscriptstyle\mathcal{V}^{\lambda}}}$ written if $\mathcal{V}^\lambda$ is linearly independent. Writing the set of vectors $V^\lambda$ as a column matrix $ V^\lambda$  we have the relation:$$ S_{|_{\scriptscriptstyle\mathcal{V}^{\lambda}}} V^\lambda = V^\lambda \tilde{S}_{|_{\scriptscriptstyle\mathcal{V}^{\lambda}}}.$$ 
    
    \begin{remark}
        Note that $\tilde{S}_{|_{\scriptscriptstyle \mathcal{V}^+}} = S_{|_{\scriptscriptstyle \mathcal{V}^+}}$.
    \end{remark}

We now arrive at the main result of this section, where we identify the largest eigenvalue and one corresponding eigenvector of the operator $S$ introduced in \eqref{eq:S-op}.

\begin{theorem} \label{thm:largest-eigenspace}
    Let~$S$ be the operator defined in Eq.~\eqref{eq:S-op}, and~$\mu$ the largest eigenvalue of~$S$. Then
    \begin{equation*}
        \mu=\Vert S_+\Vert.
    \end{equation*}
    and a corresponding eigenvector $\chi$ can be chosen of the form
    \begin{equation*}
        \chi = \sum_{i=1}^N \beta_{i} \cdot \ket{\Omega}_{(0,i)} \otimes \ket{\mathbf{v}},
    \end{equation*}
    for some vector $\mathbf v$ in $\vee^{(N - 1)}(\mathcal{H})$. The coefficients $\beta_i$ of the vector $\chi$ satisfy the equation
    \begin{equation*}
        (d - 1) \sum^N_{i = 1} \beta^2_i + \bigg{(} \sum^N_{i = 1} \beta_i \bigg{)}^2 = 1.
    \end{equation*}
\end{theorem}
\begin{proof}
    Note that $S$ is a positive semidefinite operator as a weighted sum of projections, hence its spectrum is non-negative. Thanks to the block diagonalization lemma \ref{lem:blockEigenvalue}, we have
    \begin{equation*}
        \spec (S) \cap \mathbb{R}^*_+ \subseteq \bigcup_\lambda \spec S_{|_{\scriptscriptstyle\mathcal{V}^{\lambda}}}.
    \end{equation*}
    Now since $S_{|_{\scriptscriptstyle\mathcal{V}^{\lambda}}} V^\lambda = V^\lambda \tilde{S}_{|_{\scriptscriptstyle\mathcal{V}^{\lambda}}}$, Lemma \ref{lem:spectrum-inclusion} implies that
    \begin{equation*}
        \spec S_{|_{\scriptscriptstyle\mathcal{V}^{\lambda}}} \subseteq \spec \tilde{S}_{|_{\scriptscriptstyle\mathcal{V}^{\lambda}}}.
    \end{equation*}
    Let $\lambda$ be an irrepr... But all eigenvalue $\mu$ of $\tilde{S}_{|_{\scriptscriptstyle\mathcal{V}^{\lambda}}}$ satisfies $|\mu| \leq \norm{\tilde{S}_{|_{\scriptscriptstyle\mathcal{V}^{\lambda}}}}$. Now by Lemma \ref{lem:norm}
    \begin{equation*}
        \norm{\tilde{S}_{|_{\scriptscriptstyle\mathcal{V}^{\lambda}}}} \leq \norm{{\big{(} \tilde{S}_{|_{\scriptscriptstyle\mathcal{V}^{\lambda}}} \big{)}}^{\prime}},
    \end{equation*}
    where
    \begin{equation*}
        {\big{(} \tilde{S}_{|_{\scriptscriptstyle\mathcal{V}^{\lambda}}} \big{)}}^{\prime} =
        \begin{pmatrix}
            d \, x_1 & x_1  & \ldots & x_1 \\
            x_2  & d \, x_2 & \ldots & x_2 \\
            \vdots & \vdots & \ddots & \vdots \\
            x_N & x_N & \ldots & d \, x_N
        \end{pmatrix} = S^+.
    \end{equation*}
    This holds since $\norm{\Pi^\lambda_{\sigma}} = 1$ as~$\Pi^\lambda_{\sigma}$ is a unitary representation.
    
    Now all eigenvalues $\mu$ of $S_{|_{\scriptscriptstyle\mathcal{V}^{\lambda}}}$, $\mu \leq \norm{S^+}$. This way, let~$\rho(s)$ denote the spectral radius of $S$, which is also the maximal eigenvalue of $S$ (recall that $S$ is a positive operator) then we have
    \begin{align*}
        \rho(S) &\leq \max_\lambda \max_\mu \big{\{} \mu \in \spec S_{|_{\scriptscriptstyle\mathcal{V}^{\lambda}}} \big{\}} \\
        &\leq \norm{S^+}.
    \end{align*}
    So by the Perron–Frobenius theorem~$\rho (S^+) = \norm{S^+}$, since~$S_{|_{\scriptscriptstyle V^+}}$ is a matrix with positive coefficients, which concludes the proof. A straightforward computation based on Lemma \ref{lem:scalarProduct} shows that
$$1=\langle\chi,\chi\rangle=(d - 1) \sum^N_{i = 1} \beta^2_i + {\bigg{(} \sum^N_{i = 1} \beta_i \bigg{)}}^2$$
implies that equation $(d - 1) \sum^N_{i = 1} \beta^2_i + {\bigg{(} \sum^N_{i = 1} \beta_i \bigg{)}}^2 = 1$ is satisfied.
\end{proof}

\section{Optimal Quantum Cloning}\label{sec:optimal-quantum-cloning}

This section concerns the solution of the optimization problem defined previously. We shall proceed in two steps: first, we upper bound the average fidelity by the largest eigenvalue of a matrix. Then we show that our upper bound is reached by explicit quantum cloning channel which shall be called optimal cloning channel. 

\subsection{Upper Bound}

We start with a proposition relating quantum fidelities to the maximal eigenvalue of the following matrix, acting on~$\mathcal{H}^{\otimes (N + 1)}$: 
\begin{eqnarray}\label{eq:def-R-alpha}
    R_\alpha &=& \sum^N_{i=1} \alpha_i \big{(} I_{(0,i)} + \omega_{(0,i)} \big{)} \otimes I^{\otimes (N - 1)} \in \mathcal{M}^{\otimes (N+1)}_d\nonumber\\
    &=&\left(\sum_{i=1}^N\alpha_i\right)I^{\otimes (N+1)}+S_\alpha
\end{eqnarray}
where, $\alpha \in [0,1]^N$ is a weight vector and $S_\alpha$ is defined in \eqref{eq:S-op}. Note that the maximal eigenvector derived for $S_\alpha$ is the same as the one of $R_\alpha$ since any non zero vector is a maximal eigenvector for the identity part of $R_\alpha$.

\begin{proposition} \label{prop:CloningUpperBound}
    Let $T : \mathcal{M}_d \to {(\mathcal{M}_d)}^{\otimes N}$ be a quantum channel having Choi matrix $C_T$. Then, for any weight vector $\alpha \in [0,1]^N$, we have 
    \begin{equation*}
        \bar{F}_\alpha(T) = \frac{1}{d (d + 1)} {\langle}C_T,R_\alpha{\rangle}.
    \end{equation*}
    The following upper bound holds: 
    \begin{align}
         \notag \sup_T \bar{F}_\alpha(T)\leq  \frac{\lambda_{\max}(R_\alpha)}{d + 1},
    \end{align}
    where the $\sup$ is taken over all quantum channels $T$, and $\lambda_{\max}(R_\alpha)$ is the largest eigenvalue of the matrix $R_\alpha$ of Eq.~\eqref{eq:def-R-alpha}.
\end{proposition}
\begin{proof}
    We start from the expression of the average fidelity (below, the expectations are taken with respect to the uniform measure on pure quantum states $\rho$): 
    \begin{align*}
        \bar{F}_\alpha(T) &= \sum^N_{i=1} \alpha_i \cdot \mathbb{E}_\rho \Big{[} F \big{(} \rho , T_i(\rho) \big{)} \Big{]} \\
        &= \sum^N_{i=1} \alpha_i \cdot \mathbb{E}_\rho \Big{[} \Tr \big{(} \rho \: T_i(\rho) \big{)} \Big{]} \\
        &= \sum^N_{i=1} \alpha_i \cdot \mathbb{E}_\rho \Big{[} \big{\langle} T(\rho) , \rho_{(i)} \otimes I \big{\rangle} \Big{]} \\
        &= \sum^N_{i=1} \alpha_i \cdot \mathbb{E}_\rho \Big{[} \big{\langle} C_T , \rho_{(0)}^\T \otimes \rho_{(i)} \otimes I \big{\rangle} \Big{]},
        \\
        &= \sum^N_{i=1} \alpha_i \cdot   \big{\langle} C_T , \mathbb{E}_\rho \left[\rho_{(0)}^\T \otimes \rho_{(i)}\right] \otimes I \big{\rangle}.
    \end{align*}
    To simplify further $\bar{F}_\alpha(T)$, note that taking the integral over the pure states~$\rho$.
    \begin{align*}
       \mathbb{E}_\rho \Big{[}  \rho^\T \otimes \rho  \Big{]} &=\int_{\mathcal{D}_d} \!\!\! \rho^\T \otimes \rho \; \mathrm{d}\rho \\&= {\bigg{[} \int_{\mathcal{D}_d} \!\!\! \rho \otimes \rho \; \mathrm{d}\rho \bigg{]}}^{\To} \\
        &= \frac{2}{d (d + 1)} {\big{(} P^+_{\mathfrak{S}_2} \big{)}}^{\To} \\
        &= \frac{1}{d (d + 1)} {\big{(} \Pi_{(0)(1)} + \Pi_{(0 \: 1)} \big{)}}^{\To}\\
        &= \frac{1}{d (d + 1)} (I + \omega).
    \end{align*}
   Plugging this into the expression above for $\bar{F}_\alpha(T)$ yields the first claim. Let us now show the upper bound in the statement. 
    \begin{align*}
        \sup_T \bar{F}_\alpha(T) &= \sup_T \frac{\langle C_T , R_\alpha \rangle}{d (d + 1)} \\
        &= \sup_{\substack{C \geq 0 \\ \Tr_{[1,N]} (C) = I}} \frac{1}{d (d + 1)} \langle C,R_\alpha \rangle \\
        &\leq \sup_{\substack{C \geq 0 \\ \Tr (C) = d}} \frac{1}{d (d + 1)} \langle C,R_\alpha \rangle \\
        &= \frac{\lambda_\text{max}(R_\alpha)}{d + 1}.
    \end{align*}
    The last equality comes from the fact that for any positive semi-definite matrix~$A$, and for any symmetric matrix~$B$, the following holds:
    \begin{equation*}
        \Tr [A B] \leq \Tr [A] \cdot \lambda_{\text{max}}(B).
    \end{equation*}
    Then taking $C=d\vert \chi\rangle\langle\chi\vert$ for an eigenvector corresponding to $\lambda_\text{max}(R_\alpha)$ from Theorem \ref{thm:largest-eigenspace} shows that the bound is reached.
\end{proof}
\begin{remark}
    The only place where an inequality was used in the derivation above was when we relaxed the condition $\Tr_{\scriptscriptstyle [\![ 0,N ]\!] \setminus \{0\}} (C) = I$ to the (weaker) condition $\Tr (C) = d$. As it turns out, we shall construct in the next Section \ref{sec:lower-bound} examples of quantum channels $T$ matching this upper bound, justifying a posteriori the above manipulation.
\end{remark}

\subsection{Saturation and Construction of the Asymmetric Quantum cloning channel}\label{sec:lower-bound}

In this section we build a Choi matrix candidate for saturating the upper bound of the Proposition~\ref{prop:CloningUpperBound}. We shall prove the following result.
\begin{theorem} \label{thm:candiate}
Let $\alpha\in[0,1]^N$. Let $\chi=\sum_{i=1}^N\beta_i\ket{\Omega_{(0,i}}\otimes \ket{\mathbf v}$ be a norm one eigenvector associated to $\lambda_{max}(R_\alpha)$  from Theorem \ref{thm:largest-eigenspace}. The coefficient vector $\beta$ (which is real and depends on the direction $\alpha$) satisfies the normalization condition \eqref{projectionConditionEquation}). 
Consider the associated quantum channel $T_\beta$ defined by
 \begin{equation} \label{eq:optimalCloning}
    T_\beta(\rho) := \frac{d N (N + d - 1)}{\Tr P^+_{\mathfrak{S}_N}} P_\beta \Big{(} \rho \otimes I^{\otimes (N - 1)} \Big{)} P^\T_\beta,
\end{equation}
where
\begin{equation*}
    P_\beta := \frac{1}{N!} \sum_{\sigma \in \mathfrak{S}_N} \beta_{\scriptscriptstyle \sigma(0) + 1} \cdot \Pi_\sigma.
\end{equation*}
Then $T_\beta$ is an optimal quantum cloning channel in the direction $\alpha$, that is
\begin{equation}\label{eq:equality}
\sum^N_{i=1} \alpha_i \cdot \bar{F} \big{(} {(T_\beta)}_i \big{)} = \frac{\lambda_\text{max}(R_\alpha)}{d + 1}.
\end{equation}
\end{theorem}

First we shall prove that the linear map $T_\beta$ from Eq~\eqref{eq:optimalCloning} is a quantum channel.
To this end we introduce the following Choi matrix which will be the candidate
\begin{equation*}
    C_{\beta} = \frac{d}{\Tr P^+_{\mathfrak{S}_N}} \frac{N + d - 1}{N} \sum_{\substack{1 \leq a,b \leq N \\ \sigma \in \Sigma_{a,b}}} \frac{\beta_a \beta_b}{(N - 1)!} \Pi_\sigma^{\To}
\end{equation*}
for some positive reals ${(\beta_i)}_{1 \leq i \leq N}$ satisfying
\begin{equation} \label{projectionConditionEquation}
    (d - 1) \sum^N_{i = 1} \beta^2_i + {\bigg{(} \sum^N_{i = 1} \beta_i \bigg{)}}^2 = 1.
\end{equation}

The first task is to state that this candidate is indeed a Choi matrix. This is the content of the following proposition whose technical proof is postponed in Lemma \ref{lemm:appchoi0} in Appendix \ref{app:opt}.

\begin{proposition}\label{prop:proj} Consider real numbers ${(\beta_i)}_{1 \leq i \leq N}$ such that there exists $a$ and $b$ with $\beta_a\beta_b\neq 0$. Then, the operator 
$$\widetilde{C}_{\beta}=\sum_{\substack{1 \leq a,b \leq N \\ \sigma \in \Sigma_{a,b}}} \frac{\beta_a \beta_b}{(N - 1)!} \Pi_\sigma^{\To}$$
is an orthogonal projection if and only if the condition \eqref{projectionConditionEquation} is satisfied.

Assume \eqref{projectionConditionEquation} is satisfied. The operator $C_{\beta}$ is a positive operator, such that $$\Tr_{\scriptscriptstyle [\![ 0,N ]\!] \setminus \{0\}} C_{\beta} = I.$$ 
As a consequence the matrix $C_{T_\beta}$ is a Choi Matrix
\end{proposition}

Now that we have designed a Choi matrix, the second task is to derive the associated quantum cloning channel. To this end, assume $\beta$ satisfies condition \eqref{projectionConditionEquation} let define the operator
\begin{equation*}
    P_\beta = \frac{1}{N!} \sum_{\sigma \in \mathfrak{S}_N} \beta_{\scriptscriptstyle \sigma(0) + 1} \cdot \Pi_\sigma,
\end{equation*}
and define for all states~$\rho$,
\begin{equation*}
    T_\beta(\rho) = \frac{d N (N + d - 1)}{\Tr P^+_{\mathfrak{S}_N}} P_\beta \Big{(} \rho \otimes I^{\otimes (N - 1)} \Big{)} P^\T_\beta.
\end{equation*}
This construction is a reminder of the symmetric quantum cloning case (see the en d of this section). Note in general that $P_\beta$ is in general not a projector. In order to show that $T_\beta$ is the corresponding channel associated with $C_{\beta}$ we need the following Lemma whose proof is postponed in Lemma \ref{lemm:appchoi1} in Appendix \ref{app:opt}

\begin{lemma}\label{cor:ChoiToChannel}
For any~$1 \leq a,b \leq N$, and any $\sigma \in \Sigma_{a,b}$, the permutation operator~$\sigma^{\To}$ is the Choi matrix a linear map~$T_{\mu,\nu}: \mathcal{M}_d \to {(\mathcal{M}_d)}^{\otimes N}$ defined by
\begin{equation*}
    T_{\mu,\nu}(X) = \Pi_\mu \big{(} X \otimes I^{\otimes (N - 1)} \big{)} \Pi_\nu
\end{equation*}
for some permutations $\mu$ and $\nu$ in $\mathfrak{S}_N$ such that $\mu(0) = a - 1$ and $\nu(b - 1) = 0$.

As a consequence for some fixed $1 \leq a,b \leq N$,
    \begin{equation*}
        \sum_{\sigma \in \Sigma_{a,b}} \Pi_\sigma^{\To} = \frac{1}{(N - 1)!} C_{T_{a,b}}
    \end{equation*}
    where
    \begin{equation*}
        T_{a,b}(X) = \sum_{\substack{\scriptstyle \mu,\nu \scriptstyle \in \mathfrak{S}_N \\ \scriptstyle \mu(0) \scriptstyle = a - 1 \\ \scriptstyle \nu(b - 1) \scriptstyle = 0}} \Pi_\mu \big{(} X \otimes I^{\otimes (N - 1)} \big{)} \Pi_\nu.
    \end{equation*}

\end{lemma}

Under the light of this Lemma the following proposition expresses the fact that $T_\beta$ has Choi matrix $C_{T_\beta}$.

\begin{proposition} \label{thm:TbQuantumChannel}
Let ${(\beta_i)}_{1 \leq i \leq N}$ satisfying Eq.~\eqref{projectionConditionEquation} and define 
\begin{equation*}
    \widetilde{C}_{T_\beta} = \sum_{1 \leq a,b \leq N} \frac{\beta_a \beta_b}{{(N - 1)!}^2} \: C_{T_{a,b}}
\end{equation*}
with 
\begin{equation*}
    P_\beta = \frac{1}{N!} \sum_{\sigma \in \mathfrak{S}_N} \beta_{\scriptscriptstyle \sigma(0) + 1} \cdot \Pi_\sigma,
\end{equation*}
then $T_\beta$ is a quantum channel whose Choi matrix is given by 
\begin{equation} \label{eq:CTb}
    C_{T_\beta} = \frac{d}{\Tr P^+_{\mathfrak{S}_N}} \frac{N + d - 1}{N} \sum_{\substack{1 \leq a,b \leq N \\ \sigma \in \Sigma_{a,b}}} \frac{\beta_a \beta_b}{(N - 1)!} \Pi_\sigma^{\To}
\end{equation}
\end{proposition}

\begin{proof}
As announced this proposition is a corollary of  lemma~\ref{cor:ChoiToChannel} noting that
\begin{equation*}
    \widetilde{C}_{T_\beta} = \sum_{1 \leq a,b \leq N} \frac{\beta_a \beta_b}{{(N - 1)!}^2} \: C_{T_{a,b}}
\end{equation*}
we find
\begin{equation*}
    T_\beta(\rho) = \frac{d N (N + d - 1)}{\Tr P^+_{\mathfrak{S}_N}} P_\beta \Big{(} \rho \otimes I^{\otimes (N - 1)} \Big{)}  P^\T_\beta.
\end{equation*}
and the proof is complete.
\end{proof}

Before expressing our final Theorem, which states the optimality of $T_\beta$, we need to compute the marginal of this channel. To this end, in the final proof, we shall need the following lemma whose proof is postponed in Lemma \ref{lemm:appchoi2} in Appendix \ref{app:opt}

\begin{lemma}
    Let an orthogonal projector~$\widetilde{C}_{T_\beta}$ for some reals~${(\beta_i)}_{1 \leq i \leq N}$ satisfying Eq.~\eqref{projectionConditionEquation}. Then for any~$1 \leq i \leq N$, we have
    \begin{equation*}
        \widetilde{C}_{T_\beta} \: \omega_{(0,i)} = \frac{1}{(N - 1)!} \sum_{\substack{1 \leq a \leq N \\ \sigma \in \Sigma_{a,i}}} \beta_a \bigg{(} (d - 1) \beta_i + \sum_{1 \leq b \leq N} \beta_b \bigg{)} \Pi_\sigma^{\To}.
\end{equation*}

Let Choi matrix~$C_{T_\beta}$ for some reals~${(\beta_i)}_{1 \leq i \leq N}$ satisfying Eq.~\eqref{projectionConditionEquation}. Then for any~$1 \leq i \leq N$, we have
    \begin{equation*}
        \Tr \Big{[} C_{T_\beta} \: \omega_{(0,i)} \Big{]} = d {\bigg{(} (d - 1) \beta_k + \sum^N_{j = 1} \beta_j \bigg{)}}^2.
    \end{equation*}    
\end{lemma}
Now we are in the position to give the proof of the main theorem of this section.

\begin{proof}[Proof of Theorem \ref{thm:candiate}]
The fact that the linear map $T_\beta$ is a quantum channel have been proved in Theorem~\ref{thm:TbQuantumChannel}. We now show its optimally in the direction $\alpha$.

Recall that we have
\begin{align*}
    \bar{F}_\alpha(T_\beta)&=\sum_i\alpha_i \cdot \bar{F}(T_i)\\
  &\leq\frac{\lambda_{max}(R_\alpha)}{d+1}.  
\end{align*}
Then $T_\beta$ is optimal if we have equality in the above expression, that is 
\begin{equation}
\sum^N_{i=1} \alpha_i \cdot \bar{F} \big{(} {(T_\beta)}_i \big{)} = \frac{\lambda_\text{max}(R_\alpha)}{d + 1}.
\end{equation}

Denote now $R_i=\big{(} I_{(0,i)} + \omega_{(0,i)} \big{)} \otimes I^{\otimes (N - 1)}$, using Lemma \ref{lemm:appchoi2}, for all marginal fidelities, we have ($C_{T_\beta}$ is the Choi matrix of the map $T_\beta$, see Eq~\eqref{eq:CTb})
\begin{align*}
    \bar{F}({(T_\beta)}_i) &= \mathbb{E}_\rho \bigg{[} F \Big{(} \rho , {(T_\beta)}_i(\rho) \Big{)} \bigg{]} \\
    &= \mathbb{E}_\rho \bigg{[} \Big{\langle} C_{T_\beta},\rho^\T_{(i)} \otimes \rho_{(i)} \otimes I \Big{\rangle} \bigg{]} \\
    &= \frac{1}{d (d + 1)} \big{\langle} C_{T_\beta} , R_i \big{\rangle} \\
    &= \frac{\Tr \big{[} C_{T_\beta} \big{]} + \Tr \big{[} C_{T_\beta} \: \omega_{(0,i)} \big{]}}{d (d + 1)} \\
    &= \frac{1 + {\left( (d - 1) \beta_i + \displaystyle{\sum^N_{j=1}} \beta_j \right)}^2}{d + 1}
\end{align*}
that is the equality \eqref{eq:equality} is satisfied if
\begin{align*}
    \sum^N_{i = 1} \alpha_i {\bigg{(} (d - 1) \beta_i + \sum^N_{j = 1} \beta_j \bigg{)}}^2 &= \lambda_\text{max} \bigg{(} \sum^N_{i = 1} \alpha_i \cdot \omega_{(0,i)} \bigg{)}=\lambda_\text{max}(S_\alpha) .\\
\end{align*}
In the preceding equality we have withdrawn the useless part concerning the identity operator. At this stage, since $\chi$ corresponds also to a maximal eigenvector of $S_\alpha$ from Theorem \ref{thm:largest-eigenspace}, we have
\begin{equation*}
    \bigg{\langle} \chi \bigg{|} \sum^N_{i = 1} \alpha_i \cdot \omega_{(0,i)} \bigg{|} \chi \bigg{\rangle} = \lambda_\text{max} \bigg{(} S_\alpha \bigg{)},
\end{equation*}
Now using Lemma \ref{lem:scalarProduct}, we have 
\begin{align*}
    \bigg{\langle} \chi \bigg{|} \sum^N_{i = 1} \alpha_i \cdot \omega_{(0,i)} \bigg{|} \chi \bigg{\rangle} &=  \sum^N_{i = 1} \alpha_i {\bigg{(} (d - 1) \beta_i + \sum^N_{j = 1} \beta_j \bigg{)}}^2
\end{align*}
which proves finally the optimality of $T_\beta$ and conclude the proof of Theorem \ref{thm:candiate}.
\end{proof}

Let us finish by recalling the \emph{symmetric quantum cloning problem} and showing the analogy with the known results in this context and the one we have provided. Indeed, for $\beta = \frac{1}{N (N + d - 1)}$ the following matrix is a projector:
\begin{equation*}
     \widetilde{C}_{T} = \sum_{\sigma \in \mathfrak{S}_{N + 1}} \frac{\beta}{(N - 1)!} \Pi_\sigma^{\To}.
\end{equation*}
and hence,~$\widetilde{C}_{T} \geq 0$. We define the matrix~$C_T$ to be
\begin{equation*}
    C_T = \frac{d}{\Tr P^+_{\mathfrak{S}_N}} \frac{N + d - 1}{N} \: \widetilde{C}_T,
\end{equation*}
then~$C_T \geq 0$ and~$Tr_{\scriptscriptstyle [\![ 0,N ]\!] \setminus \{1\}} (C_T) = I$. That is~$C_T$ is the Choi matrix of a quantum channel. It can be seen that it is actually the Choi matrix associated to the Optimal Symmetric Quantum cloning channel from \cite{werner1998optimal}, defined on all pure state~$\rho$ by:
\begin{equation*}
    T_{\text{opt}}(\rho) = \frac{d}{\Tr P^+_{\mathfrak{S}_N}} P^+_{\mathfrak{S}_N} \Big{(} \rho \otimes I^{\otimes (N - 1)} \Big{)} P^+_{\mathfrak{S}_N}.
\end{equation*}

\subsection{Figures of merit}

Given a set~$(p_i)_{1 \leq i \leq N}$ in~$[0,1]$ one might ask whether it exists a quantum cloning map~$T_\beta$ for some positive reals~$(\beta_i)_{1 \leq i \leq N}$ such that for all marginal~$(T_\beta)_i$ and for all (pure) states~$\rho$:
\begin{equation*}
    (T_\beta)_i (\rho) = p_i \! \cdot \! \rho + (1 - p_i) \frac{I}{d}.
\end{equation*}
We give in the next proposition a relation providing an equation satisfied by the probability tuple $(p_i)$ corresponding to the channels $T_\beta$. This relation has already been noted in the series of papers~\cite{kay2009optimal, kay2014optimal,kay2012optimal}.
\begin{proposition}\label{prop:merit}
    Let~$T_\beta$ be a~$1 \to N$ quantum cloning map, for some positive reals~$(\beta_i)_{1 \leq i \leq N}$ such that
    for all marginal~$(T_\beta)_i$ and for all pure state~$\rho$:
    \begin{equation*}
        (T_\beta)_i (\rho) = p_i \! \cdot \! \rho + (1 - p_i) \frac{I}{d}.
    \end{equation*}
    Then the~$(p_i)_{1 \leq i \leq N}$ must satisfy
    \begin{equation}\label{eq:equality-Kay}
            N + (d^2 - 1) \sum^N_{i = 1} p_i = d (d - 1) + \frac{{\Big{(} \sum^N_{i=1} \sqrt{(d^2 - 1) p_i + 1}\Big{)}}^2}{N + d - 1}.
    \end{equation}
\end{proposition}
\begin{proof}
    The figures of merit $p_i$'s for each marginal~$(T_\beta)_i$ are related to the marginal fidelities $f_i$ by $p_i = \frac{d f_i - 1}{d - 1}$ with
    \begin{equation*}
        f_i = \frac{1 + {\big{(} (d - 1) \beta_i + \sum^N_{j = 1} \beta_j \big{)}}^2}{d + 1}.
    \end{equation*}
    By summing over all $i$, we get that
    \begin{equation*}
        \sum^N_{i = 1} \sqrt{(d + 1) f_i  - 1} = (N + d - 1) \sum^N_{i = 1} \beta_i,
    \end{equation*}
    and finally
    \begin{equation*}
        \beta_i = \frac{1}{d - 1} \bigg{(} \sqrt{(d + 1) f_i  - 1} - \frac{1}{N + d - 1} \sum^N_{j = 1} \sqrt{(d + 1) f_j - 1} \bigg{)}.
    \end{equation*}
    The Eq.~\eqref{projectionConditionEquation} yields
    \begin{align*}
        \sum^N_{i = 1} \big{(} (d + 1) f_i - 1 \big{)} &= \sum^N_{i = 1} {\bigg{(} (d - 1) \beta_i + \sum^N_{j = 1} \beta_j \bigg{)}}^2 \\
        &= (d - 1) \Bigg{(} (d - 1) \sum^N_{i = 1} \beta^2_i + {\bigg{(} \sum^N_{j = 1} \beta_j \bigg{)}}^{\!\!\!\!2} \: \Bigg{)} + (N + d - 1) {\bigg{(} \sum^N_{i = 1} \beta_i \bigg{)}}^2 \\
        &= (d - 1) + (N + d - 1) {\bigg{(} \sum^N_{i = 1} \beta_i \bigg{)}}^2 \\
        &= (d - 1) + \frac{{\Big{(} \sum^N_{i = 1} \sqrt{(d + 1) f_i  - 1} \Big{)}}^2}{N + d - 1}.
    \end{align*}
    In terms of $p_i$'s, we finally find
    \begin{equation} \label{eq:piCondition}
        N + (d^2 - 1) \sum^N_{i = 1} p_i = d (d - 1) + \frac{{\Big{(} \sum^N_{i=1} \sqrt{(d^2 - 1) p_i + 1}\Big{)}}^2}{N + d - 1}.
    \end{equation}
\end{proof}
\begin{remark}
    In particular we recover the standard result of Werner \cite{werner1998optimal} on the optimal $1 \to N$ Symmetric Quantum cloning $p_{\text{opt}} = \frac{d + N}{N (d + 1)}$ if we set all the $p_i$'s equal.
\end{remark}

\begin{example}[$N = 2$, see also \cite{nechita2021geometrical}]
    The Choi matrix~$C_T$ of an optimal Asymmetric Quantum cloning $T: \mathcal{M}_d \to {(\mathcal{M}_d)}^{\otimes 2}$ is a linear combination of four permutation operators:
    \begin{equation} \label{fourPermutationOperatorsEquation:1}
        C_T =  c_1 \cdot \Pi^{\To}_{(1 \: 2)} + c_2 \cdot \Pi^{\To}_{(1 \: 3)} + c_3 \Big{(} \Pi^{\To}_{(1 \: 2 \: 3)} + \Pi^{\To}_{(3 \: 2 \: 1)} \Big{)}
    \end{equation}
    with $c_1, c_2, c_3 \geq 0$ such that $c_3 = \sqrt{c_1 c_2}$ and $d \, (c_1 + c_2) + 2 \: c_3 = 1$. The two marginals of $T$ are
    \begin{align*}
        T_1(\rho) &= (d \: c_1 + 2 \: c_3) \!\cdot\! \rho + d \: c_2 \cdot I \\
        T_2(\rho) &= (d \: c_2 + 2 \: c_3) \!\cdot\! \rho + d \: c_1 \cdot I
    \end{align*}
    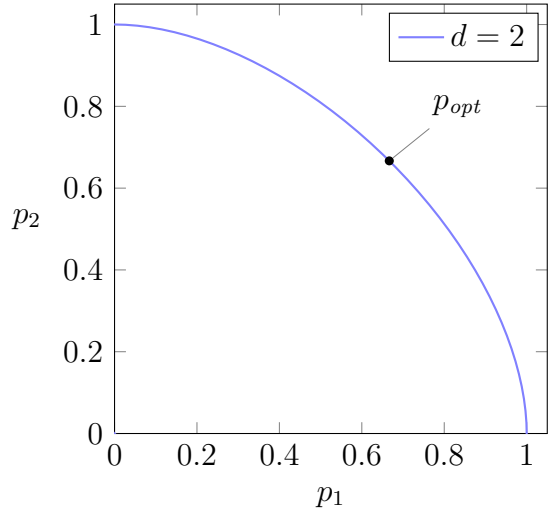
\begin{figure}[!htb]
        \centering
        \begin{tikzpicture}
            \begin{axis}[view = {0} {90},
                         xlabel = $p_1$,
                         ylabel = $p_2$,
                         ylabel style = {rotate = -90},
                         xmin = 0,
                         xmax = {1 + 0.05},
                         ymin = 0,
                         ymax = {1 + 0.05},
                         axis equal image]
                \begin{scope}[rotate around = {-45:(0,0)}]
                    \draw[draw,blue!50!white,thick] (0,{sqrt(2) / 3})
                        ellipse [x radius = {sqrt(2 / 3)},y radius = {sqrt(2) / 3}];
                \end{scope}
                \node[fill,
                      circle,
                      scale = 0.3,
                      pin = 45:{$p_{\text{opt}}$}]
                    at ({(2 + 2) / (2 * (2 + 1))},{(2 + 2) / (2 * (2 + 1))})
                    {};
                \addlegendimage{line legend,blue!50!white,thick}
                \addlegendentry{$d = 2$}
            \end{axis}
        \end{tikzpicture}
        \caption{Figures of merit of the optimal $1 \to 2$ Asymmetric Quantum cloning.}
    \end{figure}
    If we relax the optimality condition, a Choi matrix $C_T$ in the form of Eq.~\eqref{fourPermutationOperatorsEquation:1} is in the admissible region of the Asymmetric Quantum cloning problem if $C_T \geq 0$ and $\Tr_{\scriptscriptstyle [\![ 0,N ]\!] \setminus \{1\}} (C_T) = I$. These conditions, in terms of the $c_i$'s, become $c_3 \leq \sqrt{c_1 c_2}$ and $d \, (c_1 + c_2) + 2 \: c_3 = 1$, see Figure~\ref{fig:adimissibleRegion4}. In general the Choi matrix of a~$1 \to 2$ quantum cloning map is a linear combination of the six permutation operators of $\mathfrak{S}_3$, see Figure~\ref{fig:adimissibleRegion6} for the general admissible region.
    \begin{figure}[!htb]
        \begin{minipage}{0.45\textwidth}
            \begin{tikzpicture}
                \begin{axis}[scale only axis,
                             width = 0.9\textwidth,
                             view = {0} {90},
                             xlabel = $p_1$,
                             ylabel = $p_2$,
                             ylabel style = {rotate = -90},
                             xmin =  {-1 / (2^2 - 1) - 0.05},
                             xmax = {1 + 0.05},
                             ymin = {-1 / (2^2 - 1) - 0.05},
                             ymax = {1 + 0.05},
                             axis equal image,
                             axis on top = true,
                             extra x ticks = 0,
    	                     extra y ticks = 0,
    	                     extra tick style = {grid = major},
    	                     legend pos = south west]
                    \begin{scope}[rotate around = {-45:(0,0)}]
                        \draw[draw = blue!50!white,
                              fill = blue!15!white] (0,{sqrt(2) / 3})
                            ellipse [x radius = {sqrt(2 / 3)},y radius = {sqrt(2) / 3}];
                        \draw[draw = red!50!white,
                              fill = red!15!white] (0,{7 / (8 * sqrt(2))})
                            ellipse [x radius = {3 / 4},y radius = {3 / (8 * sqrt(2))}];
                        \draw[draw = orange!50!white,
                              fill = orange!15!white] (0,{(7 * sqrt(2)) / 15})
                            ellipse [x radius = {2 * sqrt(2 / 15)},y radius = {(2 * sqrt(2)) / 15}];
                    \end{scope}
                    \addlegendimage{area legend,draw = blue!50!white,fill = blue!15!white}
                    \addlegendentry{$d = 2$}
                    \addlegendimage{area legend,draw = red!50!white,fill = red!15!white}
                    \addlegendentry{$d = 3$}
                    \addlegendimage{area legend,draw = orange!50!white,fill = orange!15!white}
                    \addlegendentry{$d = 4$}
                \end{axis}
            \end{tikzpicture}
            \caption{Admissible regions with $4$ permutation operators.}
            \label{fig:adimissibleRegion4}
        \end{minipage}
        \hfill
        \begin{minipage}{0.45\textwidth}
            \begin{tikzpicture}
                \begin{axis}[scale only axis,
                             width = 0.9\textwidth,
                             view = {0} {90},
                             xlabel = $p_1$,
                             ylabel = $p_2$,
                             ylabel style = {rotate = -90},
                             xmin =  {-1 / (2^2 - 1) - 0.05},
                             xmax = {1 + 0.05},
                             ymin = {-1 / (2^2 - 1) - 0.05},
                             ymax = {1 + 0.05},
                             axis equal image,
                             axis on top = true,
                             extra x ticks = 0,
    	                     extra y ticks = 0,
    	                     extra tick style = {grid = major}]
                    \begin{scope}[rotate around = {-45:(current axis.origin)}]
                        \draw[fill,blue!15!white] (0,{sqrt(2) / 3})
                            ellipse [x radius = {sqrt(2 / 3)},y radius = {sqrt(2) / 3}];
                    \end{scope}
                    \draw[fill,blue!15!white] ({-1 / 3},{-1 / 3}) -- ({-1 / 3},{2 / 3}) -- ({2 / 3},{-1 / 3})-- ({-1 / 3},{-1 / 3});
                    \addlegendimage{area legend,fill,blue!15!white}
                    \addlegendentry{$d = 2$}
                \end{axis}
            \end{tikzpicture}
            \caption{Admissible regions with $6$ permutation operators.}
            \label{fig:adimissibleRegion6}
        \end{minipage}
    \end{figure}
\end{example}

We finish this section by showing that optimal cloning of $1\to(N-1)$ does not provide in general optimal cloning of $1\to N$.

More precisely, let $p \in [0,1]^{(N-1)}$ satisfying the optimal condition of Eq.~\eqref{eq:piCondition} for the $1 \to (N - 1)$ asymmetric quantum cloning. Under which condition $(p,0) \in [0,1]^N$ satisfies the optimal condition of Eq.~\eqref{eq:piCondition} for the $1 \to N$ asymmetric quantum cloning. The next proposition is a corrollary of Proposition \ref{prop:merit}

\begin{proposition}
Let $p \in [0,1]^{(N-1)}$ satisfying the optimal condition of Eq.~\eqref{eq:piCondition} for the $1 \to (N - 1)$ asymmetric quantum cloning. Then $(p,0) \in [0,1]^N$ satisfies the optimal condition of Eq.~\eqref{eq:piCondition} if and only if $p$ is of the form
$$p^{(i)}=e_i=(0,\ldots,1,\ldots,0),i=1,\ldots,N-1,$$ where $1$ is in position $i$ and the vector is completed with $0$
\end{proposition}

\begin{proof} Let $p=(p_i)_{1 \leq i \leq N-1} \in [0,1]^{(N-1)}$ satisfying the optimal condition of Eq.~\eqref{eq:piCondition} and
$$A=\sum^{N - 1}_{i = 1} \sqrt{(d^2 - 1) p_i + 1}.$$
We have from Eq.~\eqref{eq:piCondition}
\begin{equation}\label{eq:1toN-1}N-1 + (d^2 - 1) \sum^{N - 1}_{i = 1} p_i = d (d - 1) + \frac{A^2}{N-1 + d - 1}.
\end{equation}
Suppose now that~$(p,0) \in [0,1]^N$ satisfies the optimal condition of Eq.~\eqref{eq:piCondition} for the $1 \to N$ asymmetric quantum cloning we have
$$N + (d^2 - 1) \sum^{N - 1}_{i = 1} p_i = d (d - 1) + \frac{(A + 1)^2}{N + d - 1},$$
this yields
$$1 = \frac{(A + 1)^2}{N + d - 1} - \frac{A^2}{N - 1 + d - 1},$$
Solving this quadratic equation in terms of $A$ and selecting the non negtaive solution yiels~$A = (N+d-2)$. Plugging it into \eqref{eq:1toN-1} we get
$$\sum_{i=1}^{N-1} p_i=1$$
Now consider the function
$$p\mapsto f(p)=\left(\sum_{k=1}^{N-1} \sqrt{(d^2 - 1)p_k + 1}\right)$$
which is strictly concave. If $\sum_{i=1}^{N-1} p_i=1$ and if we denote $e_i=(0,\ldots,0,1,0,\ldots,0)$
\begin{eqnarray*}
(N+d-2)=f(p)&=&f\left(\sum_{i=1}^{N-1}p_ie_i\right)\\
&\geq& \sum_{i=1}^{N-1}p_if(e_i)\\
&=&\sum_{i=1}^{N-1} p_i(N+d-2)\\
&=&(N+d-2)
\end{eqnarray*}
Therefore the only possibility to have equality in the above equations
is that $p=e_i$ for some $i=1,\ldots, N-1.$
\end{proof}

\section{The cloning region}\label{sec:cloning-region}

\subsection{The \texorpdfstring{$\mathcal Q$}{Q}-norm}

Having related the average fidelities $\bar{F}_\alpha(T)$ to the largest eigenvalue of the matrix $R_\alpha$ from Eq.~\eqref{eq:def-R-alpha}, we show next that this quantity is further related to a norm on $\mathbb{R}^N$, which has very interesting properties. We introduce now a quantity which is crucial for our work.

\begin{definition}\label{def:Q-norm}
For a vector $x \in \mathbb{R}^N$, define its \emph{$\mathcal Q$-norm}
\begin{equation} \label{eq:def-norm-Q}
    \norm{x}_{\mathcal{Q}} := \frac{d \, \lambda_{\max}(S_x) - \norm{x}_1}{d^2 - 1},
\end{equation}
where $\norm{x}_{1} = \sum_{i=1}^N |x_i|$ is the $\ell_1$ norm of the vector $x$ and the matrix $S_x$ is given by    
\begin{equation}\label{eq:def-S-x}
    S_x = \sum^{N}_{i = 1} |x_i| \cdot \omega_{(0,i)} \otimes I^{\otimes (N - 1)} \in \mathcal M_d^{\otimes (N+1)}.
\end{equation}
The fact that the quantity above is indeed a norm in $\mathbb R^N$ will be shown in Theorem \ref{thm:Q-norm}.
\end{definition}

\begin{remark}
    The $\mathcal Q$-norm of a vector $x \in \mathbb R^N$ depends on the dimension parameter $d$.
\end{remark}

\begin{remark}
   The fact that we subtract the one norm $\Vert.\Vert_1$ implies that it is highly non trivial to show that $\Vert.\Vert_{\mathcal Q}$ is a norm. Note that on its own $\lambda_{max}$ allows to define a norm (essentially this comes from the facts that for two matrices $0\leq A\leq B$ we have $\lambda_{max}(A)\leq\lambda_{max}(B)$ and $\lambda_{max}(A+B)\leq\lambda_{max}(A)+\lambda_{max}(B)$). In general unless trivial examples subtracting two norms don't provide a new norm. This would be the content of a theorem in our context.
\end{remark}

Note that the matrix $S_x$ defined above is closely related to both the matrix $R_\alpha$ from Eq.~\eqref{eq:def-R-alpha} and to the \emph{star Hamiltonian} considered in \cite[Eq.~(3)]{kay2009optimal} and \cite{kay2012optimal, kay2014optimal}. We shall prove in Theorem \ref{thm:Q-norm} that the quantity from Eq.~\eqref{eq:def-norm-Q} is a norm; to do so, we need the following preliminary lemma which is proved in Lemma \ref{lemma:qnormapp0},\ref{lemma:qnormapp1},\ref{lemma:qnormapp2} Appendix \ref{app:norm}. Note that the proofs are again based on the design of a maximal eigenvector. 

\begin{lemma}\label{lemm:qnormfull}
The quantity $\|\cdot\|_{\mathcal Q}$ has the following properties:
\begin{itemize}
\item For all $x\in \mathbb R^N$, we have
$$\frac 1d\Vert x\Vert_1\leq\lambda_{\max}(S_x)\leq d\Vert x\Vert_1.$$

\item For all~$x, y \in \mathbb{R}^N_+$,~$\norm{x + y}_{\mathcal{Q}} \leq \norm{x}_{\mathcal{Q}} + \norm{y}_{\mathcal{Q}}$.

\item For all~$t \in [0,1]^N$ and~$x \in \mathbb{R}^N_+$,~$\norm{t \cdot x}_{\mathcal{Q}} \leq \norm{x}_{\mathcal{Q}}$.
\end{itemize}
    
\end{lemma}
In some sense this lemma gathers all the minimal requirements for a quantity to be a norm. In particular, the first point expresses the fact that the $\|\cdot\|_{\mathcal Q}$ quantity is non-negative. The first inequality is saturated if and only if the matrix $S_x$ is a (non-negative) multiple of the identity. The second inequality is saturated if and only if $x$ has at most one non-zero entry.


Under the light of the previous lemma we can prove the important result establishing that the quantity $\Vert.\Vert_{\mathcal Q}$ is a norm.

\begin{theorem}\label{thm:Q-norm}
	For all $N \geq 1$, the quantity $\norm{\cdot}_{\mathcal{Q}}$ from Eq.~\eqref{eq:def-norm-Q} is a norm on $\mathbb{R}^N$.
\end{theorem}
\begin{proof}
    The absolute homogeneity is clear since both~$\lambda_{\max}(S_x)$ and $\|x\|_1$ are also absolutely homogeneous. Let~$x \in \mathbb{R}^N$ such that~$\norm{x}_{\mathcal{Q}} = 0$, then~$$d \cdot \lambda_{\max}(S_x) =  \norm{x}_1=\tr(S_x).$$ In particular all eigenvalues of~$\sum^{N}_{i = 1} |x_i| \cdot (\omega_{(0,k)} \otimes I^{\otimes (N - 1)})$ are equal to $\lambda_{\max}(S_x)$, and hence~$\sum^{N}_{i = 1} |x_i| \cdot (\omega_{(0,k)} \otimes I^{\otimes (N - 1)})$ is a constant multiple of the identity:
    \begin{equation*}
        \sum^{N}_{i = 1} |x_i| \cdot \big{(} \omega_{(0,k)} \otimes I^{\otimes (N - 1)} \big{)} = c \cdot I^{\otimes (N + 1)}
    \end{equation*}
    for a~$c=\lambda_{\max}(S_x) \geq 0$. Let~$i, j=1,\ldots,d$ such that~$i \neq j$. Let~$\ket{\psi}$ be defined by
    \begin{equation*}
        \ket{\psi} = \big{|} i \underbrace{j \cdots j}_{N \text{ times}} \big{\rangle}.
    \end{equation*}
    Then
    \begin{align*}
        \bigg{\langle} \psi \bigg{|} \sum^{N}_{i = 1} |x_i| \cdot \big{(} \omega_{(0,k)} \otimes I^{\otimes (N - 1)} \big{)} \bigg{|} \psi \bigg{\rangle} &= 0 \\
        \bra{\psi} c \cdot I^{\otimes (N + 1)} \ket{\psi} &= c.
    \end{align*}
    Finally~$c = 0$ implies~$x = 0$, and~$\norm{\cdot}_{\mathcal{Q}}$ is positive-definite. Let~$x, y \in \mathbb{R}^N$, such that both~$\norm{x}_{\mathcal{Q}} \leq 1$ and~$\norm{y}_{\mathcal{Q}} \leq 1$. Let us show that~$\norm{\frac{x + y}{2}}_{\mathcal{Q}} \leq 1$, which is equivalent to the triangle inequality. We can assume that~$x + y \in \mathbb{R}^N_+$, otherwise multiply the~$3$ vectors by~$\sign(x + y)$. Let~$x^{\prime}$ defined by
    \begin{equation*}
        x^{\prime} =
        \begin{cases}
            [x,\frac{x+y}{2}] \: \cap \: \partial \, \mathbb{R}^N_+ &\text{ if } \cap \neq \emptyset \\
            x &\text{ otherwise}
        \end{cases}
    \end{equation*}
    and similarly for~$y^{\prime}$. Then~$x^{\prime}$,~$y^{\prime}$ and~$\frac{x+y}{2}$ are in~$\mathbb{R}^N_+$, with
    \begin{equation*}
        \frac{x + y}{2} \in [x^{\prime} , y^{\prime}],
    \end{equation*}
    then there exists $\lambda\in[0,1]$ such that $\frac{x+y}{2}=\lambda x'+(1-\lambda)y'$, then by the second point of Lemma \ref{lemm:qnormfull} and homogeneity,~$$\norm{\frac{x+y}{2}}_{\mathcal{Q}} \leq \lambda\Vert x'\Vert_{\mathcal{Q}}+(1-\lambda)\Vert y'\Vert_{\mathcal{Q}} \leq \max(\norm{x^{\prime}}_{\mathcal{Q}}, \norm{y^{\prime}}_{\mathcal{Q}}).$$ Now there exists $t_x\in[0,1]^N$ and $t_y\in[0,1]^N$ such that $x'=t_xx$ and $y'=t_yy$, then by the last point of Lemma \ref{lemm:qnormfull} both~$\norm{x^{\prime}}_{\mathcal{Q}} \leq 1$ and~$\norm{y^{\prime}}_{\mathcal{Q}} \leq 1$, such that~$\norm{\frac{x+y}{2}}_{\mathcal{Q}} \leq 1$.
\end{proof}

With the help of the $\mathcal Q$-norm, Proposition \ref{prop:CloningUpperBound} can be reformulated as: for all quantum channels $T$,
\begin{equation}\label{eq:reformulation-UB-Q}
\forall \alpha \in [0,1]^N, \qquad \bar F_\alpha(T) \leq \frac 1 d \|\alpha\|_1 + \left( 1 - \frac 1 d\right) \|\alpha\|_{\mathcal Q}.
\end{equation}

We end this section by defining the dual norm:
\begin{equation}\label{eq:def-dual-Q-norm}
    \norm{y}^*_{\mathcal{Q}} := \sup_{x \neq 0} \frac{\langle y,x \rangle}{\;\:\norm{x}_\mathcal{Q}}
\end{equation}
which will play an important role later, when discussing the set admissible cloning probabilities $\mathcal R(N,d)$, see Eq.~\eqref{eq:def-R}. 
\begin{example}[$N = 1$]
    Let~$x \in \mathbb{R}$, in one dimension the~$\mathcal{Q}$-norm of~$x$ introduced in Definition~\ref{def:Q-norm} is
    \begin{equation*}
        \norm{x}_{\mathcal{Q}} := \frac{d \, \lambda_{\max} \big{(} |x| \cdot \omega \big{)} - |x|}{d^2 - 1}, 
    \end{equation*}
    with~$\lambda_{\max} \big{(} |x| \cdot \omega \big{)} = d |x|$, such that the~$\mathcal{Q}$-norm reduces to~$\norm{x}_{\mathcal{Q}} = |x|$, as well as the dual norm defined in Equation~\eqref{eq:def-dual-Q-norm}:
    \begin{align*}
        \norm{y}^*_{\mathcal{Q}} &= \sup_{x \neq 0} \frac{\langle y,x \rangle}{\;\:\norm{x}_\mathcal{Q}} \\
        &= \sup_{\norm{x}_{\mathcal{Q}} = 1} \langle y,x \rangle \\
        &= \sup_{x = \pm 1} \langle y,x \rangle \\
        &= |y|.
    \end{align*}
\end{example}
\begin{example}[$N = 2$]
    Let~$x_1, x_2 \in \mathbb{R}$, and let~$S_x$ be defined as in Equation~\eqref{eq:def-S-x}:
    \begin{equation*}
        S_x = |x_1| \cdot \omega_{(0,1)} \otimes I + |x_2| \cdot \omega_{(0,2)} \otimes I.
    \end{equation*}
    That is
    \begin{equation*}
        S_x = |x_1| \smallPermutationOperator[1]{1/2,2/1,3/3} + |x_2| \smallPermutationOperator[1]{1/3,2/2,3/1}.
    \end{equation*}
    Then a normalized eigenvector for the largest eigenvalue of~$S_x$ is
    \begin{equation*}
        \chi = \beta_1 \cdot \ket{\Omega}_{(0,1)} \otimes \ket{\mathbf{v}} + \beta_2 \cdot \ket{\Omega}_{(0,2)} \otimes \ket{\mathbf{v}},
    \end{equation*}
    that is
    \begin{equation*}
        \chi = \beta_1 \smallUnboxedPermutationOperator[1]{1/2,2/0,3/3} \hspace{-1.6em} \smallUnboxedLabeledDiagram{1/0,2/0,3/3}{\mathbf{v}}{vector}{3}\kern-.7em+ \beta_2 \smallUnboxedPermutationOperator[1]{1/3,2/2,3/0} \hspace{-1.6em} \smallUnboxedLabeledDiagram{1/0,2/2,3/0}{\mathbf{v}}{vector}{2}\kern-.7em,
    \end{equation*}
    for some coefficients~$\beta_1, \beta_2$, and~$\mathbf{v} = \frac{1}{\sqrt{d}} \sum_{1 \leq i < d} \ket{i}$. Then the action of~$S_x$ on the eigenvector~$\chi$ is given by
    \begin{align*}
        S_x (\chi) &= \bigg{[} |x_1| \smallPermutationOperator[1]{1/2,2/1,3/3} + |x_2| \smallPermutationOperator[1]{1/3,2/2,3/1} \bigg{]} \Big{(} \beta_1 \smallUnboxedPermutationOperator[1]{1/2,2/0,3/3} \hspace{-1.6em} \smallUnboxedLabeledDiagram{1/0,2/0,3/3}{\mathbf{v}}{vector}{3}\kern-.7em+ \beta_2 \smallUnboxedPermutationOperator[1]{1/3,2/2,3/0} \hspace{-1.6em} \smallUnboxedLabeledDiagram{1/0,2/2,3/0}{\mathbf{v}}{vector}{2}\kern-.7em\Big{)} \\
        &= |x_1| \Big{(} \beta_1 \smallUnboxedPermutationOperator[1]{1/2,2/1,3/3} \hspace{-1.6em} \smallUnboxedPermutationOperator[1]{1/2,2/0,3/3} \hspace{-1.6em} \smallUnboxedLabeledDiagram{1/0,2/0,3/3}{\mathbf{v}}{vector}{3}\kern-.7em+
        \beta_2 \smallUnboxedPermutationOperator[1]{1/2,2/1,3/3} \hspace{-1.6em} \smallUnboxedPermutationOperator[1]{1/3,2/2,3/0} \hspace{-1.6em} \smallUnboxedLabeledDiagram{1/0,2/2,3/0}{\mathbf{v}}{vector}{2}\kern-.7em\Big{)}
        + |x_2| \Big{(} \beta_1 \smallUnboxedPermutationOperator[1]{1/3,2/2,3/1} \hspace{-1.6em} \smallUnboxedPermutationOperator[1]{1/2,2/0,3/3} \hspace{-1.6em} \smallUnboxedLabeledDiagram{1/0,2/0,3/3}{\mathbf{v}}{vector}{3}\kern-.7em+
        \beta_2 \smallUnboxedPermutationOperator[1]{1/3,2/2,3/1} \hspace{-1.6em} \smallUnboxedPermutationOperator[1]{1/3,2/2,3/0} \hspace{-1.6em} \smallUnboxedLabeledDiagram{1/0,2/2,3/0}{\mathbf{v}}{vector}{2}\kern-.7em\Big{)} \\
        &= |x_1| (d \beta_1 + \beta_2) \smallUnboxedPermutationOperator[1]{1/2,2/0,3/3} \hspace{-1.6em} \smallUnboxedLabeledDiagram{1/0,2/0,3/3}{\mathbf{v}}{vector}{3}\kern-.7em+ |x_2| (\beta_1 + d \beta_2) \smallUnboxedPermutationOperator[1]{1/3,2/2,3/0} \hspace{-1.6em} \smallUnboxedLabeledDiagram{1/0,2/2,3/0}{\mathbf{v}}{vector}{2}\kern-.7em.
    \end{align*}
    Solving the system of equations gives
    \begin{equation*}
        \lambda_{\max}(S_x) = \frac{1}{2} \Big{(} d \, (|x_1| + |x_2|) + \sqrt{d^2 (|x_1| - |x_2|)^2 + 4 |x_1| \, |x_2|} \Big{)},
        \end{equation*}
        and finally the~$\mathcal{Q}$-norm becomes (see Figure~\ref{fig:qNorm}):
        \begin{equation*}
            \norm{x}_{\mathcal{Q}} = \frac{\frac{d}{2} \Big{(} d \, (|x_1| + |x_2|) + \sqrt{d^2 (|x_1| - |x_2|)^2 + 4 |x_1| \, |x_2|} \Big{)} - \big{(} |x_1| + |x_2| \big{)}}{d^2 - 1}.
        \end{equation*}
\end{example}
\begin{figure}[ht]
    \centering
    \includegraphics[width=0.5\textwidth]{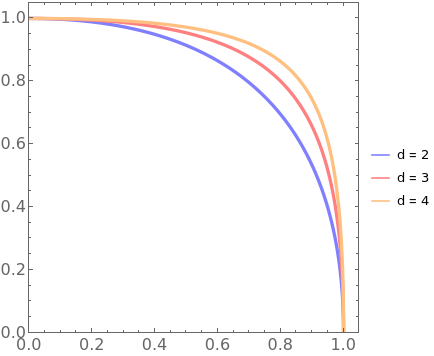}
    \caption{Some unit spheres of~$\mathbb{R}^2_+$ with the~$\mathcal{Q}$-norm.}
    \label{fig:qNorm}
\end{figure}

\subsection{Admissible region}

This section is devoted to present one of our main result, namely the cloning area. More precisely we make precise the distribution $p \in [0,1]^N$ reachable by Quantum cloning map $T : \mathcal{M}_d \to {(\mathcal{M}_d)}^{\otimes N}$, that is such that for all pure state~$\rho$ $$T_i(\rho)=p_i\!\cdot\!\rho+(1-p_i)\frac{I}{d}.$$

Let $\mathcal{R}_{N,d}$ be the admissible region of distribution $p \in [0,1]^N$ for the $1 \to N$ Asymmetric Quantum cloning problem, that is,
\begin{equation}\label{eq:def-R}
    \mathcal{R}_{N,d} := \bigg{\{} p \in [0,1]^N \: \bigg{|} \: \exists \, T \text{ cloning map s.t. } T_i(\rho) = p_i \!\cdot\! \rho + (1 - p_i) \frac{I}{d} \bigg{\}}
\end{equation}
\begin{theorem}\label{thm:region}
    $\mathcal{R}_{N,d}$ is the non-negative part of the unit ball of the dual norm $\norm{\cdot}^*_{\mathcal{Q}}$. In other words, a distribution $p \in [0,1]^N$ is in $\mathcal{R}_{N,d} $ if and only if $\norm{p}^*_{\mathcal{Q}} \leq 1$.
\end{theorem}
\begin{proof}
    Let $p \in \mathcal{R}_{N,d} $ and a cloning map $T_p : \mathcal{M}_d \to {(\mathcal{M}_d)}^{\otimes N}$ such that for all~$\rho \in \mathcal{D}_d$ and all~$1 \leq i \leq N$ we have $${\big{(} T_p \big{)}}_i (\rho) = p_i \!\cdot\! \rho + (1 - p_i) \frac{I}{d},$$ then for all $\alpha \in [0,1]^N$, $\alpha \neq 0$
    \begin{align*}
        \sup_T \bar{F}_\alpha(T) &\geq \sum_i \alpha_i \cdot \bar{F} \big{(} {(T_p)}_i \big{)} \\
        &= \sum_i \alpha_i \cdot \mathbb{E}_\rho \bigg{[} \Big{\langle} \rho\:,\:p_i \!\cdot\! \rho + (1 - p_i) \frac{I}{d} \Big{\rangle} \bigg{]} \\
        &= \sum_i \alpha_i \cdot \mathbb{E}_\rho \Big{[} p_i + \frac{1 - p_i}{d} \Big{]} \\
        &= \langle p,\alpha \rangle + \frac{\norm{\alpha}_1 - \langle p,\alpha \rangle}{d}\\
       &=  \frac 1 d \norm{\alpha}_1  + \left( 1- \frac 1 d\right)  \langle p,\alpha \rangle.
    \end{align*}
    This inequality together with the upper bound from Proposition~\ref{prop:CloningUpperBound} in the form of Eq.~\eqref{eq:reformulation-UB-Q}
    \begin{equation*}
       \sup_T \bar F_\alpha(T) \leq \frac 1 d \|\alpha\|_1 + \left( 1 - \frac 1 d\right) \|\alpha\|_{\mathcal Q},
    \end{equation*}
    give us $\langle p,\alpha \rangle \leq \|\alpha\|_{\mathcal Q}$ for all $\alpha \in [0,1]^N$. For arbitrary $\alpha \in \mathbb R^N$, note that
    $$\langle p,\alpha \rangle \leq \langle p,|\alpha| \rangle \leq  \||\alpha|\|_{\mathcal Q} =  \|\alpha\|_{\mathcal Q},$$
    showing the inclusion 
    $$\mathcal R_{N,d} \subseteq \{ p \in [0,1]^N \, : \, \|p\|_{\mathcal Q}^* \leq 1\}.$$

    
    Let $p \in [0,1]^N$ such that $\norm{p}^*_{\mathcal{Q}} = 1$, there is $\alpha \in [0,1]^N$ such that $\norm{\alpha}_{\mathcal{Q}} = \langle p,\alpha \rangle$, that is
    \begin{align*}
        \frac{d \, \lambda_{\max}(S_\alpha) - \norm{\alpha}_1}{d^2 - 1} &= \langle p,\alpha \rangle \\
        \frac{\lambda_{\max}(S_\alpha) + \norm{\alpha}_1}{d + 1} &= \frac{d - 1}{d} \langle p,\alpha \rangle + \frac{1}{d} \norm{\alpha}_1 \\
        \frac{\lambda_\text{max} \big{(} R_\alpha \big{)}}{d + 1} &= \langle f,\alpha \rangle
    \end{align*}
    where in the last equation we set $f_i = \frac{(d - 1)p_i + 1}{d}$. Then, from Theorem \ref{thm:candiate}, there are some ${(\beta_i)}_{1 \leq i \leq N}$ satisfying Eq.~\eqref{projectionConditionEquation} and a cloning map $T_\beta$ such that
    \begin{equation*}
        \sum_i \alpha_i \cdot \bar{F} \big{(} {(T_\beta)}_i \big{)} = \frac{\lambda_\text{max}(R_\alpha)}{d + 1}
    \end{equation*}
    Then the marginals of $T_\beta$ are ${(T_\beta)}_i(\rho) = p_i \!\cdot\! \rho + (1 - p_i) \frac{I}{d}$, and $p$ is in $\mathcal{R}_{N,d}$. Since this is true for all $p \in [0,1]^N$ such that $\norm{p}^*_{\mathcal{Q}} = 1$, this is true for all $p \in [0,1]^N$ such that $\norm{p}^*_{\mathcal{Q}} \leq 1$ by convexity.
\end{proof}

Importantly, since the maps $T_\beta$ from Theorem \ref{thm:candiate} are valid cloning maps, the points $p$ corresponding to these maps (see Eq.~\eqref{eq:equality-Kay}) are elements of $\mathcal R_{N,d}$.

\subsection{Convex region}

From Theorem~\ref{thm:region}, we know that the admissible region~$\mathcal{R}_{N,d}$ is a convex set delimited by a family of hyperplanes:
\begin{equation*}
    \Big{\{} p \in \mathbb{R}^N \Big{|} \langle \alpha, p \rangle = \norm{\alpha}_{\mathcal{Q}} \Big{\}}.
\end{equation*}
The asymmetric quantum cloning is an optimization problem defined in Eq~\eqref{eq:optim-worst} by
\begin{equation*}
    \sup_T \sum^N_{i=1} \alpha_i \cdot \inf_{\rho \in \mathcal{D}_d} F \big{(} \rho , T_i(\rho) \big{)},
\end{equation*}
for a distribution~$\alpha \in [0,1]^N$, and where the supremum is taken on the quantum channels and the infimum on the pure states. 
From Proposition~\ref{prop:decompChoi}, the marginals of a~$\mathcal{U}(d)$-covariant quantum channels~$T: \mathcal{M}_d \to {\big{(} \mathcal{M}_d \big{)}}^{\otimes N}$ are of the form
\begin{equation*}
    T_i(\rho) = p_i \! \cdot \! \rho + (1 - p_i) \frac{I}{d},
\end{equation*}
for all pure states~$\rho$, where the~$p_i$'s become our figures of merit. Within this formulation of the problem, if we restrict the problem to $\mathcal{U}(d)$-covariant quantum channels, we do not ask the~$p_i$'s to be collinear with the~$\alpha_i$'s. Instead we want to maximize
\begin{equation*}
    \sum^N_{i = 1} \alpha_i \, \bigg{(} p_i + \frac{1 - p_i}{d} \bigg{)}.
\end{equation*}
The cloning maps defined in Section~\ref{sec:lower-bound} can indeed give~$p_i$'s in a different direction than the~$\alpha_i$'s, specially when the direction of the~$\alpha_i$ does not intersect an extreme point of~$\mathcal{R}_{N,d}$. As a consequence, the cloning maps $T_\beta$ defined in Section~\ref{sec:lower-bound} do not fill the boundary of~$\mathcal{R}_{N,d}$, since some points in this boundary are not optimal with respect to the optimization problem. However these points can always be reached by convex combination of the cloning maps defined in Section~\ref{sec:lower-bound}.
\begin{example}[$d = 2$ \& $N = 3$]
    Let us consider the~$1 \to 3$ quantum cloning problem where we ask the first two copies are the same (have the same fidelity). That is we have to find the set
    \begin{equation*}
        \Big{\{} (p, p, q) \: \Big{|} \: p, q \in [0, 1] \Big{\}} \bigcap \mathcal{R}_{3,2}.
    \end{equation*}
    When~$q = 0$ we are in the setting of the symmetric~$1 \to 2$ quantum cloning problem, where the best (i.e.~largest) $p$ given by~\cite{werner1998optimal} is~$p_{\text{opt}} = \frac{2}{3}$; in our setting, this means that 
    $$\|(2/3, 2/3, 0)\|_{\mathcal Q}^* = 1.$$
    
    However, the cloning map defined in Section~\ref{sec:lower-bound} cannot produce the point~$(\frac{2}{3}, \frac{2}{3}, 0)$, instead we get~$(\frac{2}{3}, \frac{2}{3}, \frac{1}{9})$ (the black dot on Figure~\ref{fig:flatRegion}), since~$(\frac{2}{3}, \frac{2}{3}, 0)$ is not optimal with respect to the optimization problem. These cloning maps are the best cloning maps for this problem, but doesn't reach all the boundary of the set~$\mathcal{R}_{N, d}$. These unreachable flat region can be achieved by convexity (the flat region is in orange on Figure~\ref{fig:flatRegion}).
    \begin{figure}[!htb]
        \begin{minipage}{0.45\textwidth}
            \includegraphics[width=\textwidth]{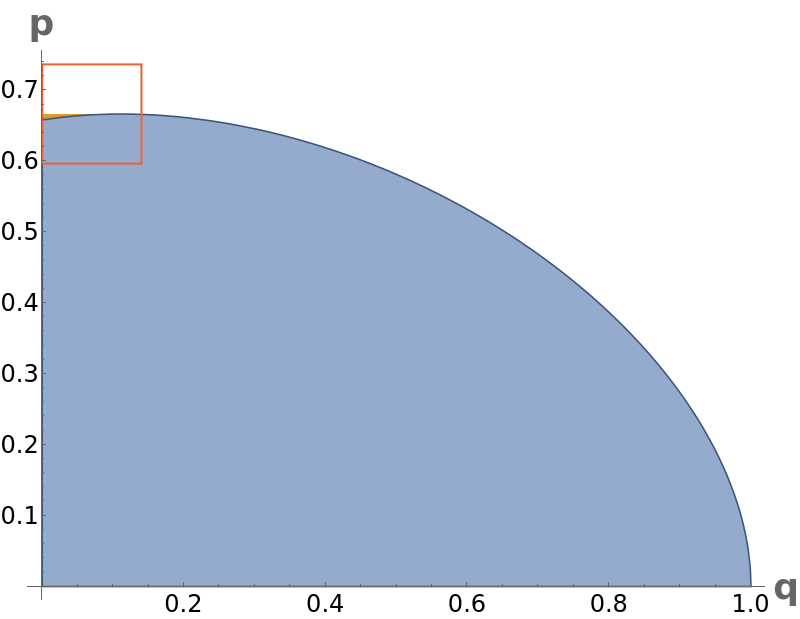}
        \end{minipage}
        \hfill
        \begin{minipage}{0.45\textwidth}
            \includegraphics[width=\textwidth]{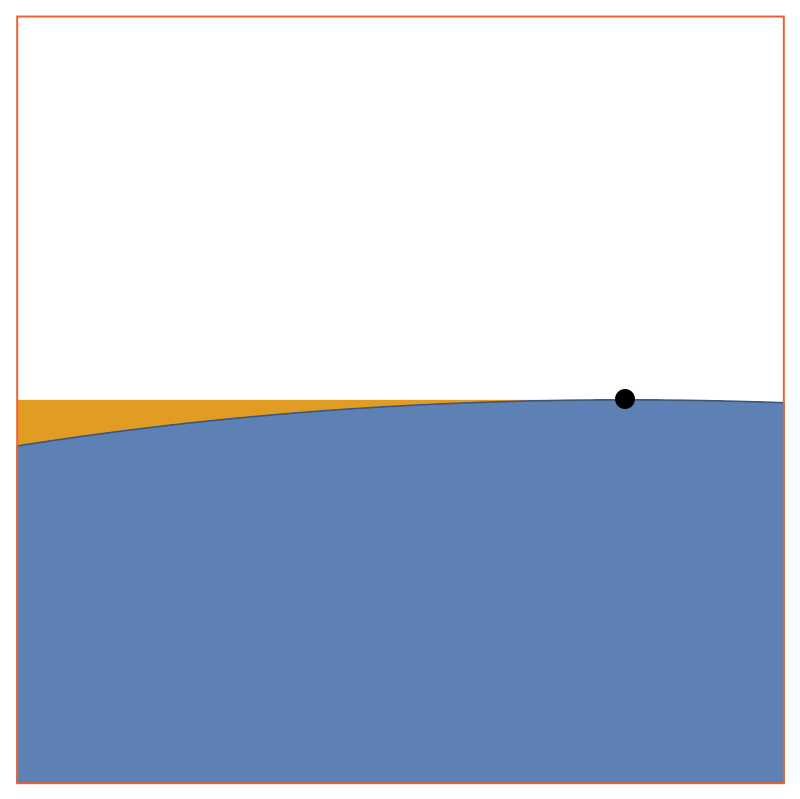}
        \end{minipage}
        \caption{The section $\{(p,p,q)\}$ of the~$1 \to 3$ quantum cloning region $\mathcal R_{3,2}$. On the axis $q=0$, the largest $p$ such that $(p,p,0) \in \mathcal R_{3,2}$ is $p=2/3$. However, the only $q$ such that $(2/3,2/3,q)$ satisfies Eq.~\eqref{eq:equality-Kay} is $q=1/9$. In the left panel (and in the zoomed region in the right panel), we have plotted in blue the points $(p,p,q)$ satisfying Eq.~\eqref{eq:equality-Kay} with an inequality, and in orange the slice of the set $\mathcal R_{3,2}$. The visible orange region in the right panel corresponds to points in $\mathcal R_{3,2}$ which do not satisfy the inequality \eqref{eq:equality-Kay}.}
        \label{fig:flatRegion}
    \end{figure}
\end{example}

In conclusion, as the previous example shows, given a tuple $p \in [0,1]^N$, in order to decide whether $p$ is admissible for the problem of asymmetric quantum cloning, it does not suffice to check whether the Eq.~\eqref{eq:equality-Kay} holds as an inequality; one needs to check the condition from Theorem \ref{thm:region}, that is $\|p\|_{\mathcal Q}^* \leq 1$. This is one of the main contributions of this paper, the realization that the convex region $\mathcal R(N,d)$ is \emph{not} described by the Eq.~\eqref{eq:equality-Kay}, seen as an inequality.  

\bigskip

\noindent\textbf{Acknowledgements.} The authors were supported by the ANR project \href{https://esquisses.math.cnrs.fr/}{ESQuisses}, grant number ANR-20-CE47-0014-01 as well as by the PHC programs \emph{Sakura} (Random Matrices and Tensors for Quantum Information and Machine Learning) and \emph{Procope} (Entanglement Preservation in Quantum Information Theory). I.N.~has also received support from the ANR project \href{https://www.math.univ-toulouse.fr/~gcebron/STARS.php}{STARS}, grant number ANR-20-CE40-0008. C.P~is also supported by the ANR projects Q-COAST ANR- 19-CE48-0003, ``Quantum Trajectories'' ANR- 20-CE40-0024-01, and ``Investissements d'Avenir'' ANR-11-LABX-0040 of the French National Research Agency.

\bibliographystyle{alpha}
\bibliography{bibliography}

\break
\appendix
\section{Structure of permutation operators}\label{app:permutation-operators}

This appendix is devoted to the presentation of technical results on permutations operator. Numerous results or remarks are illustrated by graphical calculus. 

\subsection{Contributions of permutations}

In Section \ref{sec:permutation-operators}, we only focus on permutations where $0$ is not a fixed point. This is motivated by the following remarks

Let $T : \mathcal{M}_d \to {(\mathcal{M}_d)}^{\otimes N}$ be a $\mathcal{U}(d)$-covariant quantum channel, and $C_T$ its Choi matrix, then from Proposition \ref{prop:decompChoi} we have for all $\rho \in \mathcal{D}_d$
\begin{align*}
	T(\rho) &= \Tr_0 \Big{[} C_T \big{(} \rho^\T \otimes I^{\otimes N} \big{)} \Big{]} \\
	&= \Tr_0 \bigg{[} \sum_{\sigma \in \mathfrak{S}_{N + 1}} \beta_\sigma \cdot \Pi_\sigma^{\To} \big{(} \rho^\T \otimes I^{\otimes N} \big{)} \bigg{]},
\end{align*}
for some $\beta_\sigma \in \mathbb{C}$ (see example \ref{ex:choiN2} for the case $N=2$).

Let us justify that $\sigma \in \mathfrak{S}_{N + 1}$ such that $\sigma(0) = 0$ do not contribute to the performance of the cloning map whereas the $(N + 1)$ cycles contribute the most to the copies. Indeed, for a permutation $\sigma \in \mathfrak{S}_{N + 1}$ such that $\sigma(0) = 0$ we have $\Pi_\sigma^{\To} = \Pi_\sigma$ and for all $\rho \in \mathcal{D}_d$
\begin{align*}
    \Tr_1 \Big{[} \Pi_\sigma^{\To} \big{(} \rho^\T \otimes I^{\otimes N} \big{)} \Big{]} &= \Tr_0 \Big{[} \Pi_\sigma \big{(} \rho^\T \otimes I^{\otimes N} \big{)} \Big{]} \\
    &= \frac{\Tr \rho}{d} \Tr_0 \Big{[} \Pi_\sigma \; I^{\otimes (N + 1)} \Big{]} \\
    &= \frac{1}{d} \cdot \Tr_0 \Pi_\sigma
\end{align*}
whose contributions consist in deteriorating the quality of the cloning, since the quantity $\Tr_1 \sigma$ does not depend on $\rho$. Now if $\sigma \in \mathfrak{S}_{N + 1}$ is a $(N + 1)$ cycle the marginals of
\begin{align*}
    T_{\scriptscriptstyle\text{perfect}} : \mathcal{M}_d &\longrightarrow {(\mathcal{M}_d)}^{\otimes N} \\
    \rho &\longmapsto \Tr_0 \Big{[} \Pi_\sigma^{\To} \big{(} \rho^\T \otimes I^{\otimes N} \big{)} \Big{]}
\end{align*}
are all equals to identity, i.e. ${(T_{\scriptscriptstyle\text{perfect}})}_i(\rho) = \rho$ for all $\rho \in \mathcal{M}_d$. This would be the case of a perfect quantum cloning map but the operator $T_{\scriptscriptstyle\text{perfect}}$ are of course not completly positive operator due to the no-cloning Theorem (see Example \ref{ex:perfect} fro computations with graphical calculus).

More generally, let $T : \mathcal{M}_d \to {(\mathcal{M}_d)}^{\otimes N}$ be a $\mathcal{U}(d)$-covariant quantum channel. From Proposition \ref{prop:decompChoi} the marginals $T_i$ are for all $\rho \in \mathcal{D}_d$
\begin{equation*}
	T_i(\rho) = p_i \! \cdot \! \rho + (1 - p_i)\frac{I}{d}
\end{equation*}
where the $p_i$ are the sums of the $\frac{N!}{2}$ coefficients $\beta_\sigma$ of the permutations contributing for this copy, that is the $\sigma$ such that $i$ is in the orbit of $0$ under the action of $\sigma$. The followings examples illustrates some results in terms of graphical calculus.  

\begin{example}[Partition of $\Sigma_{a,b}, N = 3$]
    The permutation set $\big{\{} \Pi^\To_\sigma \; \big{|} \;\sigma \in \mathfrak{S}_4, \quad \sigma(0) \neq 0 \big{\}}$ is partitioned in
    \begin{align*}
        \Sigma_{1,1} &: \Big{\{} \smallPermutationOperator[1]{1/2,2/1,3/3,4/4} , \smallPermutationOperator[1]{1/2,2/1,3/4,4/3} \Big{\}}, &
        \Sigma_{1,2} &: \Big{\{} \smallPermutationOperator[1]{1/2,2/3,3/1,4/4} , \smallPermutationOperator[1]{1/2,2/4,3/1,4/3} \Big{\}}, &
        \Sigma_{1,3} &: \Big{\{} \smallPermutationOperator[1]{1/2,2/3,3/4,4/1} , \smallPermutationOperator[1]{1/2,2/4,3/3,4/1} \Big{\}}, \\[0.5em]
        \Sigma_{2,1} &: \Big{\{} \smallPermutationOperator[1]{1/3,2/1,3/2,4/4} , \smallPermutationOperator[1]{1/3,2/1,3/4,4/2} \Big{\}}, &
        \Sigma_{2,2} &: \Big{\{} \smallPermutationOperator[1]{1/3,2/2,3/1,4/4} , \smallPermutationOperator[1]{1/3,2/4,3/1,4/2} \Big{\}}, &
        \Sigma_{2,3} &: \Big{\{} \smallPermutationOperator[1]{1/3,2/2,3/4,4/1} , \smallPermutationOperator[1]{1/3,2/4,3/2,4/1} \Big{\}}, \\[0.5em]
        \Sigma_{3,1} &: \Big{\{} \smallPermutationOperator[1]{1/4,2/1,3/2,4/3} , \smallPermutationOperator[1]{1/4,2/1,3/3,4/2} \Big{\}}, &
        \Sigma_{3,2} &: \Big{\{} \smallPermutationOperator[1]{1/4,2/2,3/1,4/3} , \smallPermutationOperator[1]{1/4,2/3,3/1,4/2} \Big{\}}, &
        \Sigma_{3,3} &: \Big{\{} \smallPermutationOperator[1]{1/4,2/2,3/3,4/1} , \smallPermutationOperator[1]{1/4,2/3,3/2,4/1} \Big{\}}.
    \end{align*}
\end{example}

\begin{example}[Partially transposed permutations, $N = 3$]
    On~$\mathfrak{S}_{4}$ the~$(0 \: i)$-partially transposed permutations are
    \begin{align*}
        \Pi^{\To}_{(0 \: 1)} &= \smallPermutationOperator[1]{1/2,2/1,3/3,4/4}, & \Pi^{\To}_{(0 \: 2)} &= \smallPermutationOperator[1]{1/3,2/2,3/1,4/4}, & \Pi^{\To}_{(0 \: 3)} &= \smallPermutationOperator[1]{1/4,2/2,3/3,4/1}.
    \end{align*}
\end{example}

\begin{example}\label{ex:choiN2}[Choi Matrix $N = 2$]
    Let $T : \mathcal{M}_d \to \mathcal{M}_d \otimes \mathcal{M}_d$ be a $\mathcal{U}(d)$-covariant quantum channel, then its Choi matrix $C_T$ is:
    \begin{equation*}
        \beta_{\id} \smallPermutationOperator[1]{1/1,2/2,3/3} + \beta_{(12)} \smallPermutationOperator[1]{1/2,2/1,3/3} + \beta_{(13)} \smallPermutationOperator[1]{1/3,2/2,3/1} + \beta_{(23)} \smallPermutationOperator[1]{1/1,2/3,3/2} + \beta_{(123)} \smallPermutationOperator[1]{1/2,2/3,3/1} + \beta_{(321)} \smallPermutationOperator[1]{1/3,2/1,3/2}
    \end{equation*}
\end{example}

\begin{example}\label{ex:perfect}[Contributions of quantum clonning]
    Let $T : \mathcal{M}_d \to \mathcal{M}_d \otimes \mathcal{M}_d$ be a $\mathcal{U}(d)$-covariant quantum channel, such that its Choi matrix $C_T$ is the partially transposed cycle~$(1 \: 2 \: 3)$:
    \begin{equation*}
        C_T = \smallPermutationOperator[1]{1/2,2/3,3/1}.
    \end{equation*}
    Then the two marginals are
    \begin{align*}
        T_1(\rho) &= \Tr_{\{0,2\}} \Big{[} (1 \: 2 \: 3)^{\To} \big{(} \rho^\T \otimes I^{\otimes N} \big{)} \Big{]} \\
        &= \Tr_2 \bigg{[} \smallUnboxedPermutationOperator[1]{1/0,2/1,3/3,4/4} \hspace{-1.6em} \smallUnboxedPermutationOperator[2]{1/1,2/3,3/4,4/2} \hspace{-1.6em} \smallUnboxedLabeledDiagram{1/1,2/2,3/3,4/4}{\rho^T}{operator}{2/2} \hspace{-1.6em} \smallUnboxedPermutationOperator[1]{1/2,2/0,3/3,4/4} \bigg{]} \\
        &= \Tr_1 \Big{[} \smallUnboxedPermutationOperator{1/1,2/2} \hspace{-1.6em} \smallUnboxedPermutationOperator{1/2,2/1} \hspace{-1.6em} \smallUnboxedLabeledDiagram{1/1,2/2}{\rho}{operator}{2/2} \Big{]} \\
        &= \smallUnboxedPermutationOperator[2]{1/1,2/0,3/2} \hspace{-1.6em} \smallUnboxedPermutationOperator{1/2,2/1,3/3} \hspace{-1.6em} \smallUnboxedLabeledDiagram{1/1,2/2,3/3}{\rho}{operator}{2/2} \hspace{-1.6em} \smallUnboxedPermutationOperator[2]{1/1,2/3,3/0} \\
        &= \smallUnboxedLabeledDiagram{1/1}{\rho}{operator}{1/1},
    \end{align*}
    and
    \begin{align*}
        T_2(\rho) &= \Tr_{\{0,1\}} \Big{[} (1 \: 2 \: 3)^{\To} \big{(} \rho^\T \otimes I^{\otimes N} \big{)} \Big{]} \\
        &= \Tr_1 \bigg{[} \smallUnboxedPermutationOperator[1]{1/0,2/1,3/3,4/4} \hspace{-1.6em} \smallUnboxedPermutationOperator[2]{1/1,2/3,3/4,4/2} \hspace{-1.6em} \smallUnboxedLabeledDiagram{1/1,2/2,3/3,4/4}{\rho^T}{operator}{2/2} \hspace{-1.6em} \smallUnboxedPermutationOperator[1]{1/2,2/0,3/3,4/4} \bigg{]} \\
        &= \Tr_0 \Big{[} \smallUnboxedPermutationOperator{1/1,2/2} \hspace{-1.6em} \smallUnboxedPermutationOperator{1/2,2/1} \hspace{-1.6em} \smallUnboxedLabeledDiagram{1/1,2/2}{\rho}{operator}{2/2} \Big{]} \\
        &= \smallUnboxedPermutationOperator[1]{1/0,2/1,3/3} \hspace{-1.6em} \smallUnboxedPermutationOperator{1/1,2/3,3/2} \hspace{-1.6em} \smallUnboxedLabeledDiagram{1/1,2/2,3/3}{\rho}{operator}{3/3} \hspace{-1.6em} \smallUnboxedPermutationOperator[1]{1/2,2/0,3/3} \\
        &= \smallUnboxedLabeledDiagram{1/1}{\rho}{operator}{1/1}.
    \end{align*}
\end{example}

\subsection{Proof of Lemma \ref*{lem:struct}} \label{sec:app-proof-lemma-permutation-operators}

In this subsection, we prove the key Lemma \ref{lem:struct}, which gives the structure of partially permutations operators and action on vectors. We illustrate as well with graphical calculus.

\begin{lemma}\label{lem:premutationForm}
    Let $1 \leq a,b \leq N$, for any $\sigma \in \Sigma_{a,b}$, there exists a unique $\hat{\sigma} \in \mathfrak{S}_{N - 1}$ such as the partially transposed permutation operator $\sigma^{\To}$ is 
    \begin{equation*}
        \Pi_\sigma^{\To} = \Pi_{(1 \: a)} \: (\omega_{(0,1)} \otimes \Pi_{\hat{\sigma}}) \; \Pi_{(1 \: b)}.
    \end{equation*}
\end{lemma}
\begin{proof}
    Let $1 \leq a,b,c,d \leq N$, and $\sigma \in \Sigma_{a,b}$, then $\Pi_{(c \: a)} \; \Pi_\sigma^{\To} \; \Pi_{(d \: b)}$ is a partially transposed permutation operator from $\Sigma_{c,d}$. Let $\sigma' \in \Sigma_{1,1}$ be the permutation such that
    \begin{equation*}
        \Pi_{(1 \: a)} \; \Pi_\sigma^{\To} \; \Pi_{(1 \: b)} = \Pi_{\sigma'}^{\To}.
    \end{equation*}
    There is a unique $\hat{\sigma} \in \mathfrak{S}_{N - 1}$, such that
    \begin{equation*}
        \Pi_{(1 \: a)} \; \Pi_\sigma^{\To} \; \Pi_{(1 \: b)} = \omega_{(0,1)} \otimes \Pi_{\hat{\sigma}},
    \end{equation*}
    where the permutation operator is $\Pi_{\hat{\sigma}}$ is
    \begin{equation*}
        \Pi_{\hat{\sigma}} = \frac{1}{d} \cdot \Tr_{0,1} \big{[} \Pi_{(1 \: a)} \; \Pi_\sigma^{\To} \; \Pi_{(1 \: b)} \big{]}.
    \end{equation*}
\end{proof}

\begin{example} \label{ex:graphical1}
    Let~$\sigma = (0 \: 3 \: 2) \in \mathfrak{S}_4$. Its partial transposition operator
    \begin{equation*}
        \Pi_\sigma^{\To} = \smallPermutationOperator[1]{1/4,2/2,3/1,4/3},
    \end{equation*}
    can be decomposed into
    \begin{align*}
        \Pi_\sigma^{\To} &= \smallUnboxedPermutationOperator[1]{1/4,2/2,3/1,4/3} \\[0.5em]
        &= \smallUnboxedPermutationOperator{1/1,2/4,3/3,4/2} \hspace{-1.6em} \smallUnboxedPermutationOperator[1]{1/2,2/1,3/4,4/3} \hspace{-1.6em} \smallUnboxedPermutationOperator{1/1,2/3,3/2,4/4} \\[0.5em]
        &= \smallUnboxedPermutationOperator{1/1,2/4,3/3,4/2} \hspace{-1.6em} \smallUnboxedLabeledDiagram[1]{1/2,2/1,3/3,4/4}{\hat{\sigma}}{operator}{3/4} \hspace{-1.6em} \smallUnboxedPermutationOperator{1/1,2/3,3/2,4/4} \\[0.5em]
        &= \Pi_{(1 \: 3)} \: \big{(} \omega_{(0,1)} \otimes \Pi_{(0 \: 1)} \big{)} \; \Pi_{(1 \: 2)}.
    \end{align*}
\end{example}

Recall that for two integers $i,j$ we denote by $(i \, : \, j)$ the permutation $\mathfrak{S}_n$ defined by
\begin{equation*}
    (i \, : \, j) =
    \begin{cases}
        \big{(} (i - 1) \: (i - 2) \: \cdots \: j \big{)} &\text{if } 0 \leq j < i \leq n \\
        \big{(} i \: (i + 1) \: \cdots \: (j - 1) \big{)} &\text{if } 0 \leq i < j \leq n \\
        (0)(1) \cdots (n - 1) &\text{otherwise}
    \end{cases}.
\end{equation*}
That is, the $|i - j|$ cycle of $\mathfrak{S}_n$ between $i$ and $j$.
\begin{example}[$n = 3$] \label{ex:graphical2}
    The non-trivial~$(i \, : \, j)$-cycles of~$\mathfrak{S}_3$ are:
    \begin{align*}
        (0 \, : \, 2) &= \smallPermutationOperator{1/2,2/1,3/3}, &
        (2 \, : \, 0) &= \smallPermutationOperator{1/2,2/1,3/3}, \\[0.5em]
        (1 \, : \, 3) &= \smallPermutationOperator{1/1,2/3,3/2}, &
        (3 \, : \, 1) &= \smallPermutationOperator{1/1,2/3,3/2}, \\[0.5em]
        (0 \, : \, 3) &= \smallPermutationOperator{1/2,2/3,3/1}, &
        (3 \, : \, 0) &= \smallPermutationOperator{1/3,2/1,3/2}.
    \end{align*}
\end{example}

\begin{lemma} \label{lem:permutAction}
    Let $1 \leq a,b,c \leq N$, and let $\sigma \in \Sigma_{a,b}$. Let $\phi \in \mathcal{H}^{\otimes (N - 1)}$, then
    \begin{align*}
        \Pi_\sigma^{\To} &(\ket{\Omega}_{(0,c)} \otimes \ket{\phi}) \\
        &= \begin{cases}
            d \cdot \ket{\Omega}_{(0,a)} \otimes \big{(} \Pi_{(0 \, : \, (a - 1))} \circ \Pi_{\hat{\sigma}}\circ \Pi_{((b - 1) \, : \, 0)} \big{)} \ket{\phi} &\text{if } b = c \\[1em]
            \ket{\Omega}_{(0,a)} \otimes \big{(} \Pi_{(0 \, : \, (a - 1))} \circ \Pi_{\hat{\sigma}} \circ \Pi_{((b - 1) \, : \, 0)} \circ \Pi_{((c - 1) \, : \, (b - 1))} \big{)} \ket{\phi} &\text{if } b \neq c.
        \end{cases}
    \end{align*}
\end{lemma}
\begin{proof}
    From Lemma \ref{lem:premutationForm} we can write $\Pi_\sigma^{\To} = \Pi_{(1 \: a)} \: (\omega_{(0,1)} \otimes \Pi_{\hat{\sigma}}) \; \Pi_{(1 \: b)}$, or equivalently
    \begin{equation*}
        \Pi_\sigma^{\To} = \Pi_{(1 \: a)} \: (I^{\otimes 2} \otimes \Pi_{\hat{\sigma}}) \, \big{(} \omega_{(0,1)} \otimes I^{\otimes (N - 1)} \big{)} \; \Pi_{(1 \: b)}.
    \end{equation*}
    Together with $\omega \ket{\Omega} = d \cdot \ket{\Omega}$, we have the following relations:
    \begin{equation*}
        \big{(} \omega_{(0,1)} \otimes I^{\otimes (N - 1)} \big{)} \; \Pi_{(1 \: b)} = \Pi_{(1 \: b)} \; \big{(} \omega_{(0,b)} \otimes I^{\otimes (N - 1)} \big{)},
    \end{equation*}
    and
    \begin{equation*}
        \Pi_{(1 \: b)} \big{(} \ket{\Omega}_{(0,b)} \otimes \ket{\phi} \big{)} = \ket{\Omega} \otimes \Pi_{((b - 1) \, : \, 0)} \ket{\phi}.
    \end{equation*}
    If $b = c$, then
    \begin{align*}
        \big{(} \omega_{(0,1)} \otimes I^{\otimes (N - 1)} \big{)} \; \Pi_{(1 \: b)} \big{(} \ket{\Omega}_{(0,c)} \otimes \ket{\phi} \big{)} &= \Pi_{(1 \: b)} \; \big{(} \omega_{(0,b)} \otimes I^{\otimes (N - 1)} \big{)} \big{(} \ket{\Omega}_{(0,c)} \otimes \ket{\phi} \big{)} \\
        &= \Pi_{(1 \: b)} \big{(} \omega_{(0,b)} \ket{\Omega}_{(0,b)} \otimes \ket{\phi} \big{)} \\
        &= d \cdot \Pi_{(1 \: b)} \big{(} \ket{\Omega}_{(0,b)} \otimes \ket{\phi} \big{)} \\
        &= d \cdot \big{(} \ket{\Omega} \otimes \Pi_{((b - 1) \, : \, 0)} \ket{\phi} \big{)},
    \end{align*}
    such that finally $\Pi_\sigma^{\To} \big{(} \ket{\Omega}_{(0,c)} \otimes \ket{\phi} \big{)} = d \cdot \Pi_{(1 \: a)} \Big{[} \ket{\Omega}_{(0,1)} \otimes \big{(} \Pi_{\hat{\sigma}} \circ \Pi_{((b - 1) \, : \, 0)} \circ \big{)} \ket{\phi} \Big{]}$.\\ If $b \neq c$, we have the following relation:
    \begin{equation*}
        \big{(} \omega_{(0,b)} \otimes I^{\otimes (N - 1)} \big{)} \big{(} \ket{\Omega}_{(0,c)} \otimes \ket{\phi} \big{)} = \ket{\Omega}_{(0,b)} \otimes \Pi_{((c - 1) \, : \, (b - 1))} \ket{\phi}.
    \end{equation*}
    Then
        \begin{align*}
        \big{(} \omega_{(0,1)} \otimes I^{\otimes (N - 1)} \big{)} \; \Pi_{(1 \: b)} \big{(} \ket{\Omega}_{(0,c)} \otimes \ket{\phi} \big{)} &= \Pi_{(1 \: b)} \; \big{(} \omega_{(0,b)} \otimes I^{\otimes (N - 1)} \big{)} \big{(} \ket{\Omega}_{(0,c)} \otimes \ket{\phi} \big{)} \\
        &= \Pi_{(1 \: b)} \big{(} \ket{\Omega}_{(0,b)} \otimes \Pi_{((c - 1) \, : \, (b - 1))} \ket{\phi} \big{)} \\
        &= \ket{\Omega}_{(0,1)} \otimes \big{(} \Pi_{((b - 1) \, : \, 0)} \circ \Pi_{((c - 1) \, : \, (b - 1))} \big{)} \ket{\phi},
    \end{align*}
    such that finally~$$\sigma^{\To}(\ket{\Omega}_{(0,c)} \otimes \ket{\phi}) = \Pi_{(1 \: a)} \Big{[} \ket{\Omega}_{(0,1)} \otimes \big{(} \Pi_{\hat{\sigma}} \circ \Pi_{((b - 1) \, : \, 0)} \circ \Pi_{((c - 1) \, : \, (b - 1))} \big{)} \ket{\phi} \Big{]}.$$
    To conclude, from the relation
    \begin{equation*}
        \Pi_{(1 \: a)} \big{(} \ket{\Omega}_{(0,1)} \otimes \ket{\phi} \big{)} = \ket{\Omega}_{(0,a)} \otimes \Pi_{(0 \, : \, (a - 1))} \ket{\phi},
    \end{equation*}
    we obtain
   \begin{align*}
        \Pi_\sigma^{\To}(\ket{\Omega}_{(0,c)} \otimes \ket{\phi}) =\\=&
        \begin{cases}
            d \cdot \ket{\Omega}_{(0,a)} \otimes \big{(} \Pi_{(0 \, : \, (a - 1))} \circ \Pi_{\hat{\sigma}}\circ \Pi_{((b - 1) \, : \, 0)} \big{)} \ket{\phi} &\text{if } b = c \\[1em]
            \ket{\Omega}_{(0,a)} \otimes \big{(} \Pi_{(0 \, : \, (a - 1))} \circ \Pi_{\hat{\sigma}} \circ \Pi_{((b - 1) \, : \, 0)} \circ \Pi_{((c - 1) \, : \, (b - 1))} \big{)} \ket{\phi} &\text{if } b \neq c.
        \end{cases}
    \end{align*}
\end{proof}
\begin{example} \label{ex:graphical3}
    Let~$\sigma = (0 \: 3 \: 1 \: 2) \in \mathfrak{S}_4$, and let~$\phi \in \mathcal{H}^{\otimes 2}$, then
    \begin{align*}
        \Pi_\sigma^{\To} \big{(} \ket{\Omega}_{(0,1)} \otimes \ket{\phi} \big{)} &= \smallUnboxedPermutationOperator[1]{1/4,2/3,3/1,4/2} \hspace{-1.6em} \smallUnboxedLabeledDiagram[1]{1/2,2/0,3/3,4/4}{\phi}{vector}{3/4} \\[0.5em]
        &= \smallUnboxedPermutationOperator[1]{1/4,2/2,3/3,4/0} \hspace{-1.6em} \smallUnboxedPermutationOperator{1/0,2/3,3/2,4/0} \hspace{-1.6em} \smallUnboxedLabeledDiagram{1/0,2/2,3/3,4/0}{\phi}{vector}{2/3} \\[0.5em]
        &= \ket{\Omega}_{(0,3)} \otimes \Pi_{(0 \: 1)} \ket{\phi}.
    \end{align*}
\end{example}

In order to simplify the next equations we use the notations~$\Pi_{\hat{\sigma}(a,b,c)}$ and~$\Pi_{b,c}$, defined by:
\begin{equation*}
    \Pi_{\hat{\sigma}(a,b,c)} = \Pi_{(0 \, : \, (a - 1))} \circ \Pi_{\hat{\sigma}} \circ \Pi_{b,c},
\end{equation*}
and
\begin{equation*}
    \Pi_{b,c} = 
    \begin{cases}
        \Pi_{((b - 1) \, : \, 0)} &\text{if } b = c \\
        \Pi_{((b - 1) \, : \, 0)} \circ \Pi_{((c - 1) \, : \, (b - 1))} &\text{if } b \neq c
    \end{cases}.
\end{equation*}

The following lemma concerns the image and the kernel of operators  $\Pi_\sigma^{\To}$.
\begin{lemma} \label{lem:ImKer}
    For any $1 \leq a,b \leq N$ and any $\sigma \in \Sigma_{a,b}$, 
    \begin{align*}
	    \Ima \Pi_\sigma^{\To} &= \Span \Big{\{} \ket{\Omega}_{(0,a)} \otimes \ket{\phi} \; \Big{|} \; \phi \in \mathcal{H}^{\otimes (N - 1)} \Big{\}} \\
	    \Ker \Pi_\sigma^{\To} &\supseteq \Span {\Big{\{} \ket{\Omega}_{(0,k)} \otimes \ket{\phi} \; \Big{|} \; 1 \leq k \leq N, \; \phi \in \mathcal{H}^{\otimes (N - 1)} \Big{\}}}^\perp.
    \end{align*}
\end{lemma}
\begin{proof}
    Let~$1 \leq a,b \leq N$, let~$\sigma \in \Sigma_{a,b}$, and let~$v \in \Ima \Pi_\sigma^{\To}$, then~$v = \ket{\Omega}_{(0,a)} \otimes \ket{\phi}$ with~$\phi \in \mathcal{H}^{\otimes (N - 1)}$, from Lemma \ref{lem:premutationForm}. Let~$\phi \in \mathcal{H}^{\otimes (N - 1)}_d$, and define~$u = \frac{1}{d} \ket{\Omega}_{(0,b)} \otimes \Pi^{\shortminus 1}_{\hat{\sigma},a,b,c} \ket{\phi}$. Then from Lemma \ref{lem:permutAction},
    \begin{equation*}
         \Pi_\sigma^{\To}(u) = \ket{\Omega}_{(0,a)} \otimes \ket{\phi}.
    \end{equation*}
    Since on finite-dimensional spaces, the kernel of a linear map~$M$ is equal to the orthogonal complement of the image of adjoint~$M^*$, i.e.~$\Ker M = {\big{(} \Ima M^* \big{)}}^\perp$, we have
    \begin{align*}
        \Ker \Pi_\sigma^{\To} &= \Span {\Big{\{} \ket{\Omega}_{(1,b)} \otimes \ket{\phi} \; \Big{|} \; \phi \in \mathcal{H}^{\otimes (N - 1)} \Big{\}}}^\perp \\
        &\supseteq \Span {\Big{\{} \ket{\Omega}_{(1,k)} \otimes \ket{\phi} \; \Big{|} \; 1 \leq k \leq N,\; \phi \in \mathcal{H}^{\otimes (N - 1)} \Big{\}}}^\perp
    \end{align*}
\end{proof}

\subsection{Partial trace of partially transposed permutation operators}

\begin{lemma} \label{lem:permuPartialTrace}
    Let~$1 \leq a,b \leq N$, then
    \begin{align*}
        \Tr_{\scriptscriptstyle [\![ 0,N ]\!] \setminus \{0\}} \bigg{[} \sum_{\sigma \in \Sigma_{a,a}} \Pi_\sigma^{\To} \bigg{]} &= (N - 1)! \; \Tr \Big{[} P^+_{\mathfrak{S}_{N - 1}} \Big{]} \cdot I \\
        \Tr_{\scriptscriptstyle [\![ 0,N ]\!] \setminus \{0\}} \bigg{[} \sum_{\sigma \in \Sigma_{a,b}} \Pi_\sigma^{\To} \bigg{]} &= \frac{(N - 1)!}{d} \Tr \Big{[} P^+_{\mathfrak{S}_{N - 1}} \Big{]} \cdot I,
    \end{align*}
    where the orthogonal projector onto the trivial representation subspace of~$\mathfrak{S}_{N - 1}$ is defined by~$P^+_{\mathfrak{S}_{N - 1}} = \frac{1}{(N - 1)!} \sum_{\sigma \in \mathfrak{S}_{N - 1}} \Pi_\sigma$.
\end{lemma}
\begin{proof}
    For the first equation we have from Lemma \ref{lem:premutationForm}
    \begin{align*}
        \Tr_{\scriptscriptstyle [\![ 0,N ]\!] \setminus \{0\}} &\bigg{[} \sum_{\sigma \in \Sigma_{a,a}} \Pi_\sigma^{\To} \bigg{]} = \Tr_{\scriptscriptstyle [\![ 0,N ]\!] \setminus \{0\}} \bigg{[} \sum_{\sigma \in \mathfrak{S}_{N - 1}} \Pi_{(1 \: a)} \: (\omega \otimes \Pi_{\hat{\sigma}}) \; \Pi_{(1 \: a)} \bigg{]} \\
        &= \sum_{\sigma \in \mathfrak{S}_{N - 1}} \Tr_{\scriptscriptstyle [\![ 0,N ]\!] \setminus \{0\}} \Big{[} \Pi_{(1 \: a)} \: (\omega \otimes \Pi_{\hat{\sigma}}) \; \Pi_{(1 \: a)} \Big{]} \\
        &= \sum_{\sigma \in \mathfrak{S}_{N - 1}} \Tr_{\scriptscriptstyle [\![ 0,N ]\!] \setminus \{0\}} \Big{[} (I \otimes \Pi_{(0 \: (a - 1))}) \: (\omega \otimes \Pi_{\hat{\sigma}}) \; (I \otimes \Pi_{(0 \: (a - 1))}) \Big{]} \\
        &= \sum_{\sigma \in \mathfrak{S}_{N - 1}} \sum_{i,j=1}^d \Tr_{\scriptscriptstyle [\![ 0,N ]\!] \setminus \{0\}} \Big{[} (I \otimes \Pi_{(0 \: (a - 1))}) \: (\ketbra{i}{j} \otimes \ketbra{i}{j} \otimes \Pi_{\hat{\sigma}}) \; (I \otimes \Pi_{(0 \: (a - 1))}) \Big{]} \\
        &= \sum_{\sigma \in \mathfrak{S}_{N - 1}} \sum_{i,j=1}^d \Tr_{\scriptscriptstyle [\![ 0,N ]\!] \setminus \{0\}} \Big{[} (I \ketbra{i}{j} I) \otimes \big{(} \Pi_{(0 \: (a - 1))} \: (\ketbra{i}{j} \otimes \Pi_{\hat{\sigma}}) \: \Pi_{(0 \: (a - 1))} \big{)} \Big{]} \\
        &= \sum_{\sigma \in \mathfrak{S}_{N - 1}} \sum_{i,j=1}^d \Tr \Big{[} \Pi_{(0 \: (a - 1))} \: (\ketbra{i}{j} \otimes \Pi_{\hat{\sigma}}) \; \Pi_{(0 \: (a - 1))} \Big{]} \cdot \ketbra{i}{j} \\
        &= \sum_{\sigma \in \mathfrak{S}_{N - 1}} \sum_{i,j=1}^d \Tr \Big{[} (\ketbra{i}{j} \otimes \Pi_{\hat{\sigma}}) \; (\Pi_{(0 \: (a - 1))} \circ \Pi_{(0 \: (a - 1))}) \Big{]} \cdot \ketbra{i}{j} \\
        &= \sum_{\sigma \in \mathfrak{S}_{N - 1}} \sum_{i,j=1}^d \braket{i}{j} \otimes \Tr \big{[} \Pi_{\hat{\sigma}} \big{]} \cdot \ketbra{i}{j} \\
        &= \sum_{\sigma \in \mathfrak{S}_{N - 1}} \Tr \big{[} \Pi_{\hat{\sigma}} \big{]} \cdot I, \\
    \end{align*}
    where~$\Pi_{(0 \: (a - 1))}$ is a permutation operator of~$\mathfrak{S}_N$. For the second equation, since the partial transposition is a linear operator, we have the equality
    \begin{equation*}
        \Tr_{\scriptscriptstyle [\![ 0,N ]\!] \setminus \{0\}} \bigg{[} \sum_{\sigma \in \Sigma_{a,b}} \Pi_\sigma^{\To} \bigg{]} = {\Bigg{(} \Tr_{\scriptscriptstyle [\![ 0,N ]\!] \setminus \{0\}} \bigg{[} \sum_{\sigma \in \Sigma_{a,b}} \Pi_\sigma \bigg{]} \Bigg{)}}^{\To}.
    \end{equation*}
    For any~$\sigma \in \Sigma_{a,b}$ the relation
    \begin{equation*}
        (\bar{U} \otimes U^{\otimes N}) \: \Pi_\sigma^{\To} \: (\bar{U} \otimes U^{\otimes N}) = \Pi_\sigma^{\To}
    \end{equation*}
    holds for all~$U \in \mathcal{U}(d)$, such that~$\Tr_{\scriptscriptstyle [\![ 0,N ]\!] \setminus \{0\}} \Big{[} \sum_{\sigma \in \Sigma_{a,b}} \Pi_\sigma^{\To} \Big{]}$ is a multiple of the identity, and
    \begin{equation*}
        \Tr_{\scriptscriptstyle [\![ 0,N ]\!] \setminus \{0\}} \bigg{[} \sum_{\sigma \in \Sigma_{a,b}} \Pi_\sigma^{\To} \bigg{]} = \Tr_{\scriptscriptstyle [\![ 0,N ]\!] \setminus \{0\}} \bigg{[} \sum_{\sigma \in \Sigma_{a,b}} \Pi_\sigma \bigg{]} = c \cdot I,
    \end{equation*}
    with~$c = \frac{1}{d} \Tr \Big{[} \sum_{\sigma \in \Sigma_{a,b}} \Pi_\sigma \Big{]}$. For a permutation~$\sigma \in \mathfrak{S}_n$ we write~$\# \sigma$ the number of disjoint cycles of~$\sigma$. Then
    \begin{align*}
        \Tr \Big{[} \sum_{\sigma \in \Sigma_{a,b}} \Pi_\sigma \Big{]} &= \sum_{\sigma \in \Sigma_{a,b}} d^{\# \sigma} \\
        &= \sum_{\sigma \in \Sigma_{a,a}} d^{\# [\sigma \circ (a \: b)]}.
    \end{align*}
    Let~$\sigma \in \Sigma_{a,a}$ and write~$\sigma = c_1 \circ \cdots \circ c_k$ the decomposition of~$\sigma$ in disjoint cycles~$c_i = (s^i_1 \: \cdots \: s^i_{l_i})$. If it exists~$i \in [k]$ such that~$a,b \in c_i$ then~$c_i \circ (a \: b)$ can be decomposed in two disjoint cycles. Otherwise, if there exist~$i,j \in [k]$ such that~$a \in c_i$ and~$b \in c_j$ then~$c_i \circ c_j \circ (a \: b)$ can be decomposed in one disjoint cycle. Finally
    \begin{equation*}
        \# \big{[} \sigma \circ (a \: b) \big{]} =
        \begin{cases}
            \# \sigma + 1 &\text{if } \exists i\in [k] \text{ s.t. } a,b \in c_i \\
            \# \sigma - 1 &\text{otherwise}
        \end{cases}.
    \end{equation*}
    But since~$\sigma \in \Sigma_{a,a}$ and~$b \neq a$, in the decomposition of~$\sigma$ in disjoint cycles,~$a$ is in the cycle~$(0 \: a)$, and~$\# \big{[} \sigma \circ (a \: b) \big{]} = \# \sigma - 1$. That is
    \begin{equation*}
        \Tr \Big{[} \sum_{\sigma \in \Sigma_{a,b}} \Pi_\sigma \Big{]} = \frac{1}{d} \Tr \Big{[} \sum_{\sigma \in \Sigma_{a,a}} \Pi_\sigma \Big{]}.
    \end{equation*}
\end{proof}

\subsection{Scalar product of vectors in \texorpdfstring{$\mathcal{V}^\lambda$}{V lambda}}\label{app:scalarproduct}

\begin{lemma}
    Let~$\lambda_1, \lambda_2$ be two distinct irreducible representations of~$\mathfrak{S}_{N - 1}$, let~$\mathbf{v}_1, \mathbf{v}_2$ be two normalized vectors in the irreducible subspace~$\lambda_1$ and~$\lambda_2$ of~$\mathfrak{S}_{N - 1}$, and let~$1 \leq k,l \leq N$. Then
    \begin{equation*}
        \braket{\Omega_{(0,k)} \otimes\mathbf{v}_1}{\Omega_{(0,l)} \otimes\mathbf{v}_2 } =
        \begin{cases}
            d \cdot \braket{\mathbf{v}_1}{\mathbf{v}_2} &\text{if } k = l \\
            \bra{\mathbf{v}_1} \Pi_{((l - 1) \, : \, (k - 1))} \ket{\mathbf{v}_2} &\text{if } k \neq l.
        \end{cases}
    \end{equation*}
    As a consequence
    \begin{align*}
        \braket{\Omega_{(0,k)} \otimes\mathbf{v}_1}{\Omega_{(0,k)} \otimes\mathbf{v}_1} &= d \\
        \braket{\Omega_{(0,k)} \otimes\mathbf{v}_1}{\Omega_{(0,l)} \otimes\mathbf{v}_1} &= \bra{\mathbf{v}_1} \Pi_{((l - 1) \, : \, (k - 1))} \ket{\mathbf{v}_1},
    \end{align*}
    and
    \begin{equation*}
        \braket{\Omega_{(0,k)} \otimes\mathbf{v}_1}{\Omega_{(0,l)} \otimes\mathbf{v}_2} = 0.
    \end{equation*}
    In particular the vector spaces $\Span \mathcal{V}^\lambda$ are in orthogonal direct sum.
\end{lemma}

\begin{proof}
    Since both~$\mathbf{v}_1$ and~$\mathbf{v}_2$ are normalized, we have~$\braket{\mathbf{v}_i}{\mathbf{v}_i} = 1$ for all~$i \in \{1, 2\}$. Then from~$\braket{\Omega}{\Omega} = d$ we have
    \begin{align*}
       \braket{\Omega_{(0,k)} \otimes\mathbf{v}_1}{\Omega_{(0,k)} \otimes\mathbf{v}_1} &= \braket{\Omega}{\Omega} \otimes \braket{\mathbf{v}_1}{\mathbf{v}_1} \\
        &= d.
    \end{align*}
    When~$k \neq l$, from the equation~$\bra{\Omega_{(0, k)} \otimes \phi} \ket{\Omega_{(0, l)} \otimes \psi} = \bra{\phi} \Pi_{((l - 1) \, : \, (k - 1))} \ket{\psi}$ for any~$\phi, \psi \in \mathcal{H}^{\otimes (N - 1)}$, we have 
    \begin{equation*}
        \braket{\Omega_{(0,k)} \otimes\mathbf{v}_1}{\Omega_{(0,l)} \otimes\mathbf{v}_1} = \bra{\mathbf{v}_1} \Pi_{((l - 1) \, : \, (k - 1))} \ket{\mathbf{v}_1}.
    \end{equation*}
    When we look at~$\lambda_1, \lambda_2$, two distinct irreducible representations of~$\mathfrak{S}_{N - 1}$, the scalar product becomes
    \begin{equation*}
        \braket{\Omega_{(0,k)} \otimes\mathbf{v}_1}{\Omega_{(0,l)} \otimes\mathbf{v}_2 } =
        \begin{cases}
            \braket{\Omega}{\Omega} \otimes \braket{\mathbf{v}_1}{\mathbf{v}_2} &\text{if } k = l \\
            \bra{\mathbf{v}_1} \Pi_{((l - 1) \, : \, (k - 1))} \ket{\mathbf{v}_2} &\text{if } k \neq l.
        \end{cases}
    \end{equation*}
    Note that~$\Pi_{((l - 1) \, : \, (k - 1))} \ket{\mathbf{v}_2}$ lives in the irreducible subspace of~$\mathfrak{S}_{N - 1}$ associated to~$\lambda_2$, since the irreducible representations are stable by any representation of permutations. Then both scalar product are null because the irreducible spaces are orthogonal each other.
\end{proof}
\begin{remark}
    In the particular case of the trivial representation, we have for all~$1 \leq k,l \leq N$ and all normalized~$\mathbf{v}$ in the irreducible subspace~$\vee^{(N - 1)}(\mathcal{H})$,
    \begin{align*}
        \braket{\Omega_{(0,k)} \otimes\mathbf{v}}{\Omega_{(0,k)} \otimes\mathbf{v}} &= d \\
        \braket{\Omega_{(0,k)} \otimes\mathbf{v}}{\Omega_{(0,l)} \otimes\mathbf{v}} &= 1.
    \end{align*}
\end{remark}

From Lemma \ref{lem:permutAction}, the action of the permutation operators $\Pi^{\To}_{\sigma_{a,b}}$ on the vectors $\ket{\Omega}_{(0,k)} \otimes \ket{\mathbf{v}}$, for all $1 \leq k \leq N$, is
\begin{equation} \label{eq:actionOnVLambda}
    \Pi^{\To}_{\sigma_{a,b}} \big{(} \ket{\Omega}_{(0,k)} \otimes \ket{\mathbf{v}} \big{)} =
    \begin{cases}
        d \cdot \ket{\Omega}_{(1,a)} \otimes \Pi^\lambda_{\hat{\sigma},a,b,b} \ket{\mathbf{v}} &\text{ if } b = k\\
        \ket{\Omega}_{(1,a)} \otimes \Pi^\lambda_{\hat{\sigma},a,b,k} \ket{\mathbf{v}} &\text{ if } b \neq k,
    \end{cases}
\end{equation}
with~$\Pi^\lambda_\sigma = P^\lambda \: \Pi_\sigma \: P^\lambda$, and~$P^\lambda$ the projector onto the irreducible subspace~$\lambda$.

For any permutation operator~$\Pi^{\To}_{\sigma_{a,b}}$ we know from Lemma \ref{lem:ImKer} that the~$V^{\lambda}$'s form a generators of the orthogonal complement of the kernel of~$\Pi^{\To}_{\sigma_{a,b}}$. However, in general they do not define a basis. We can extract a subset of linearly independent vectors to form a basis. In this basis~$\Pi^{\To}_{\sigma_{a,b}}$ can be block diagonalized due to Eq.~\eqref{eq:actionOnVLambda}, since for all irreducible representation~$\lambda$, and any vector~$v$ in~$\mathcal{V}^{\lambda}$, we have 
\begin{equation*}
    \Pi^{\To}_{\sigma_{a,b}} (v) \in \Span\mathcal{V}^{\lambda}.
\end{equation*}

\begin{example}[$N = 3$]
    Let~$x_1, x_2, x_3 \in \mathbb{R}_+$, and define
    \begin{equation*}
        S = x_1 \cdot \Pi^{\To}_{(0 \: 1)} + x_2 \cdot \Pi^{\To}_{(0 \: 2)} + x_3 \cdot \Pi^{\To}_{(0 \: 3)},
    \end{equation*}
    that is
    \begin{equation*}
        S = x_1 \smallPermutationOperator[1]{1/2,2/1,3/3,4/4} + x_2 \smallPermutationOperator[1]{1/3,2/2,3/1,4/4} + x_3 \smallPermutationOperator[1]{1/4,2/2,3/3,4/1}.
    \end{equation*}
    Then an eigenvector for the largest eigenvalue of~$R$ is
    \begin{equation*}
        \chi = \beta_1 \cdot \ket{\Omega}_{(0,1)} \otimes \ket{\mathbf{v}} + \beta_2 \cdot \ket{\Omega}_{(0,2)} \otimes \ket{\mathbf{v}} + \beta_3 \cdot \ket{\Omega}_{(0,3)} \otimes \ket{\mathbf{v}},
    \end{equation*}
    that is
    \begin{equation*}
        \chi = \beta_1 \smallUnboxedPermutationOperator[1]{1/2,2/0,3/3,4/4} \hspace{-1.6em} \smallUnboxedLabeledDiagram{1/0,2/0,3/3,4/4}{\mathbf{v}}{vector}{3/4} + \beta_2 \smallUnboxedPermutationOperator[1]{1/3,2/2,3/1,4/4} \hspace{-1.6em} \smallUnboxedLabeledDiagram{1/0,2/2,3/0,4/4}{\mathbf{v}}{vector}{2/4} + \beta_3 \smallUnboxedPermutationOperator[1]{1/4,2/2,3/3,4/1} \hspace{-1.6em} \smallUnboxedLabeledDiagram{1/0,2/2,3/3,4/0}{\mathbf{v}}{vector}{2/3},
    \end{equation*}
    for some coefficients~$\beta_i$, and~$\mathbf{v} = \displaystyle\sum_{1 \leq i \leq j < d} \ket{ij}$.
\end{example}

\section{Some results in linear algebra}

In the proofs, we use simple linear algebra results that are not mainstream. We gather them in this appendix for the convenience of the reader. 

\begin{lemma}\label{lem:norm}
    For all~$1 \leq i, j \leq n$, let $M_{ij}$ be some matrices in~$\mathcal{M}_m$, and define the block matrix
    \begin{equation*}
        M = \sum_{1 \leq i, j \leq n} M_{ij} \otimes \ketbra{i}{j},
    \end{equation*}
    and the matrix~$M'= \sum_{1 \leq i, j \leq n} \norm{M_{ij}} \cdot \ketbra{i}{j}$. We then have
    \begin{equation*}
        \norm{M} \leq \norm{M'}.
    \end{equation*}
\end{lemma}
\begin{proof}
    Consider $v$ and $w$ two unit vectors such that $\norm{M} = \bra{v} M \ket{w}$, we can decompose $v=\displaystyle\sum_{i=1}^n \vert v_i\rangle\otimes\vert i\rangle$ with $v_i\in\mathbb C^m$ and we define $v'=\displaystyle\sum^n_{i=1} \Vert v_i\Vert \cdot \vert i\rangle$. Since we consider the Euclidian norm, we have obviously $\Vert v'\Vert=\Vert v\Vert=1$. We define~$w'$ similarly, and then have
    \begin{align*}
        \big{|} \! \bra{v} M \ket{w} \! \big{|} &= \bigg{|} \sum_{ij} \bra{v_i} M_{ij} \ket{w_j} \bigg{|} \\
        &\leq \sum_{ij} \norm{v_i} \: \norm{M_{ij}} \: \norm{w_j}\quad (\text{Cauchy-Schwarz}) \\
        &= \sum_{ij} \norm{v_i} \: \norm{w_j} \: \big{\langle} i \big{|} M^{'} \big{|} j \big{\rangle} \\
        &=\big{\langle} v' \big{|} M^{'} \big{|} w' \big{\rangle}\\
        &\leq \Vert v'\Vert \, \Vert w'\Vert \, || M^{'} ||\\
        &\leq || M^{'} ||.
    \end{align*}
    which ends the proof.
\end{proof}

The foolowing lemma is usefull in the paper when we relate the spectrum of $S$ with the one of $\tilde S$

\begin{lemma}\label{lem:spectrum-inclusion}
    Let $A\in\mathcal M$ be an operator, let $n\geq m$ let~$\{ f_1, \ldots, f_n \}$ be a set of vectors such that $\mathbb C^m=Vect\{ f_1, \ldots, f_n \}$. We suppose that for all $i=1,\ldots,n$ there exists $a_{ij}$, $j=1,\ldots,n$ such that
    $$Af_i=\sum_{i,j=1}^n a_{ij}f_j$$
Denote $\tilde A$ the $n\times n$ matrix with coefficients $(a_{ij})$. 
    Then~$$\spec{A} \subseteq \spec{\tilde{A}}.$$
\end{lemma}

\begin{proof} Let $F$ be the matrix composed of the vectors $f_i$, $i=1,\ldots,n$ written in column. Then $F$ is a $m\times n$ matrix of rank $m$ since $\mathbb C^m=Vect\{ f_1, \ldots, f_n \}$. By definition of $\tilde A$ we have $$AF=F\tilde A$$
Furthermore for all scalar $\alpha$
\begin{align*}
        (A - \alpha \, I_m) F &= A F - \alpha \, I_m \, F \\
        &= F \tilde{A} - \alpha \, F\\
        &=F(\tilde A-\alpha I_n)
    \end{align*}
Now using the fact that $rank(AB)\leq min(rank(A),rank(B))$ for all matrices $A$ and $B$, if $\alpha\in Spec A$ since $\rank(F)=m$
$$\rank(F(\tilde A-\alpha I_n))=\rank((A - \alpha \, I_m) F)<m.$$
Now since $n\geq m$ this necessarily implies that $\rank(\tilde A-\alpha I_n)<n$ which says that $\alpha\in Spec \tilde A$
\end{proof}
\begin{remark}
    Note that the matrix $\tilde A$ is not uniquely defined and depends on a choice of the way of writing $Af_i=\sum_{i,j=1}^n a_{ij}f_j$ but the results concerning the inclusion is true whatever if the form of the matrix $\tilde A$. Of course the other inclusion is not true in general unless $n=m$ and $\{ f_1, \ldots, f_n \}$ is linearily independent (but this is trivial in this case). 
\end{remark}
\begin{remark}
    In the lemma the condition $Af_i=\sum_{i,j=1}^n a_{ij}f_j$ means that $A$ leaves $V=Vect\{ f_1, \ldots, f_n \}$ invariant and we can rephrase this lemma by considering the restriction to $A$ on $V$ and then saying that $\spec{A_{\vert V}} \subseteq \spec{\tilde{A}}.$ which is what we use in the body of the paper.
\end{remark}

\section{The \texorpdfstring{$\mathcal Q$}{Q}-norm}\label{app:norm}
\begin{lemma}\label{lemma:qnormapp0}
For all $x\in \mathbb R^N$, we have
$$\frac 1d\Vert x\Vert_1\leq\lambda_{\max}(S_x)\leq d\Vert x\Vert_1.$$
In particular, the $\|\cdot\|_{\mathcal Q}$ quantity is non-negative. The first inequality is saturated if and only if the matrix $S_x$ is a (non-negative) multiple of the identity. The second inequality is saturated if and only if $x$ has at most one non-zero entry.
\end{lemma}
\begin{proof}
For the first inequality, note that
$$\lambda_{\max}(S_x) \geq \frac{\Tr S_x}{d^{N+1}} = \frac{ \sum^N_{i=1} d |x_i| \cdot d^{N-1}}{d^{N+1}} = \frac{\|x\|_1}{d}.$$
The inequality above is saturated if and only if the eigenvalues of $S_x$ are identical. 

For the second inequality, we use the subadditivity of the $\lambda_{\max}$ functional: 
$$\lambda_{\max}(S_x) \leq \sum^N_{i=1} \lambda_{\max}\left( |x_i| \cdot \omega_{(0,i)} \otimes I^{\otimes (N - 1)} \right) = \sum^N_{i=1} d |x_i| = d\|x\|_1.$$
The inequality above is saturated if and only if the matrices $|x_i| \cdot \omega_{(0,i)} \otimes I$ have a common largest eigenvector, which can happen only if the support of $x$ has size 0 or 1. 
\end{proof}
\begin{lemma}\label{lemma:qnormapp1}
    For all~$x, y \in \mathbb{R}^N_+$,~$\norm{x + y}_{\mathcal{Q}} \leq \norm{x}_{\mathcal{Q}} + \norm{y}_{\mathcal{Q}}$.
\end{lemma}
\begin{proof}
This follows from the subadditivity of the $\lambda_{\max}$ functional: for $x, y \in \mathbb{R}^N_+$ we have
    \begin{align*}
        \lambda_{\max} \bigg{[} \sum^{N}_{i = 1} &(x_i + y_i) \cdot \omega_{(0,i)} \otimes I^{\otimes (N - 1)} \bigg{]} \\
        &\leq \lambda_{\max} \bigg{[} \sum^{N}_{i = 1} x_i \cdot \omega_{(0,i)} \otimes I^{\otimes (N - 1)} \bigg{]} + \lambda_{\max} \bigg{[} \sum^{N}_{i = 1}  y_i \cdot \omega_{(0,i)} \otimes I^{\otimes (N - 1)} \bigg{]}.
    \end{align*}
\end{proof}
\begin{lemma}\label{lem:rect}\label{lemma:qnormapp2}
    For all~$t \in [0,1]^N$ and~$x \in \mathbb{R}^N_+$,~$\norm{t \cdot x}_{\mathcal{Q}} \leq \norm{x}_{\mathcal{Q}}$.
\end{lemma}
\begin{proof}
    Let us show that for~$0 \leq x \leq y$ (meaning that $x_i\leq y_i$, $i=1,\ldots,N$), we have~$\norm{x}_{\mathcal{Q}} \leq \norm{y}_{\mathcal{Q}}$. Let~$\chi = \sum_k \beta_k \ket{\Omega_{(0,k)}}\otimes\ket{\mathbf v}$ be a normalized eigenvector corresponding to the the largest eigenvalue~$\lambda_{\max}(S_x)$, from Theorem \ref{thm:largest-eigenspace}, that is
    \begin{equation*}
        \lambda_{\max}(S_x) = \Big{\langle} \chi \Big{|} \sum^{N}_{i = 1} x_i \cdot \big{(} \omega_{(0,i)} \otimes I^{\otimes (N - 1)} \big{)} \Big{|} \chi \Big{\rangle}.
    \end{equation*}
    Note that due to the Perron–Frobenius Theorem, $\beta_i > 0$ for all $i=1,\ldots,N$. Let us point out two facts. First using Lemma \ref{lem:scalarProduct}
    \begin{align*}
        \bra{\chi} \omega_{(0,i)} \otimes I^{\otimes (N - 1)} \ket{\chi} &= \sum^N_{k = 1} \sum^N_{l = 1} \beta_k \beta_l \cdot \bra{\Omega_{(0,k)} \otimes \mathbf{v}} \big{(} \omega_{(0,i)} \otimes I^{\otimes (N - 1)} \big{)} \ket{\Omega_{(0,l)} \otimes \mathbf{v}} \\
        &= \sum^N_{k = 1} \beta_k \cdot \Big{\langle} \Omega_{(0,k)} \otimes \mathbf{v} \Big{|} d \beta_i \cdot \Omega_{(0,i)} \otimes \mathbf{v} + \sum^N_{\substack{l = 1 \\ l \neq i}} \beta_l \cdot \Omega_{(0,l)} \otimes \mathbf{v} \Big{\rangle} \\
        &= \sum^N_{k = 1} \beta_k \cdot \Big{\langle} \Omega_{(0,k)} \otimes \mathbf{v} \Big{|} (d - 1) \beta_i \cdot \Omega_{(0,i)} \otimes \mathbf{v} + \sum^N_{l = 1} \beta_l \cdot \Omega_{(0,l)} \otimes \mathbf{v} \Big{\rangle} \\
        &= {\bigg{(} (d - 1) \beta_i + \sum_k \beta_k \bigg{)}}^2
    \end{align*}
    Second
    \begin{align*}
        1 &=\braket{\chi} \\
        &= \sum_{i\neq j}\beta_i\beta_j+d\sum_i\beta_i^2\\
        &= \left(\sum_{i=1}^N\beta_i\right)^2+(d-1)\sum_{i=1}^N\beta_i^2
    \end{align*}
    Now using that $\left(\sum_{i=1}^N\beta_i\right)^2\geq \sum_{i=1}^N\beta_i^2$ we get that
  $1\leq d\left(\sum_{i=1}^N\beta_i\right)^2$ then, since $\beta_i\geq0$
  $$\bra{\chi} \omega_{(0,i)} \otimes I^{\otimes (N - 1)} \ket{\chi}\geq \left(\sum_{i=1}^N\beta_i\right)^2\geq\frac 1d$$  
This way since $y\geq x$, we have that $\lambda_{max}(S_y)\geq \langle\chi,S_y\chi\rangle$ then
    \begin{align*}
        \lambda_{max}(S_y)-\lambda_{max}(S_x)&\geq\langle \chi\vert S_y-S_x\vert \chi \rangle\\
        &=\sum_{i=1}^N(y_i-x_i)\langle \chi\vert \omega_{(0,i)} \otimes I^{\otimes (N - 1)}\vert \chi \rangle\\
        &=\sum_{i=1}^N(y_i-x_i){\bigg{(} (d - 1) \beta_i + \sum_k \beta_k \bigg{)}}^2\\
        &\geq \sum_{i=1}^N(y_i-x_i)\frac 1d
     \end{align*}
     Then
     $$\norm{y}_{\mathcal{Q}}=\lambda_{max}(S_y)-\frac 1d \Vert y\Vert\geq \lambda_{max}(S_x)-\frac 1d \Vert x\Vert=\norm{x}_{\mathcal{Q}}$$
which was the desired result.     
    \end{proof}
    
\section{Optimal cloning}\label{app:opt}

\begin{lemma}\label{lemm:appchoi0} Let some reals ${(\beta_i)}_{1 \leq i \leq N}$. The operator 
$$\widetilde{C}_{T_\beta}=\sum_{\substack{1 \leq a,b \leq N \\ \sigma \in \Sigma_{a,b}}} \frac{\beta_a \beta_b}{(N - 1)!} \Pi_\sigma^{\To}$$
such that there exists $a$ and $ b$ with $\beta_a\beta_b\neq 0$ is an orthogonal projection if and only if the condition \eqref{projectionConditionEquation}
\begin{equation*} 
    (d - 1) \sum^N_{i = 1} \beta^2_i + {\bigg{(} \sum^N_{i = 1} \beta_i \bigg{)}}^2 = 1
\end{equation*}
is satisfied.

Assume \eqref{projectionConditionEquation} is satisfied. The operator $C_{T_\beta}$ is a positive operator, such that $\Tr_{\scriptscriptstyle [\![ 0,N ]\!] \setminus \{0\}} C_{T_\beta} = I$ 
\end{lemma}

\begin{proof}
Let some reals ${(\beta_i)}_{1 \leq i \leq N}$, then
\begin{align*}
    {\big{(} \widetilde{C}_{T_\beta} \big{)}}^2 &= \sum_{\substack{1 \leq a,b \leq N \\ 1 \leq c,d \leq N}} \: \sum_{\substack{\sigma \in \Sigma_{a,b} \\ \tau \in \Sigma_{c,d}}} \frac{\beta_a \beta_b \beta_c \beta_d}{{(N-1)!}^2} \: \Pi_\sigma^{\To} \Pi_\tau^{\To} \\
    &= \sum_{1 \leq a,d \leq N} \bigg{[} d \sum_{\substack{1 \leq b,c \leq N \\ b = c}} \sum_{\sigma \in \Sigma_{a,d}} \frac{\beta_a \beta_b \beta_c \beta_d}{{(N-1)!}^2} \: \Pi_\sigma^{\To} + \sum_{\substack{1 \leq b,c \leq N \\ b \neq c}} \sum_{\sigma \in \Sigma_{a,d}} \frac{\beta_a \beta_b \beta_c \beta_d}{{(N-1)!}^2} \: \Pi_\sigma^{\To} \bigg{]}.
\end{align*}
The projection condition ${\big{(} \widetilde{C}_{T_\beta} \big{)}}^2 = \widetilde{C}_{T_\beta}$ is equivalent to
\begin{align*}
    \beta_a \beta_b &= d \sum^N_{i = 1} \beta_a \beta_i \beta_i \beta_b + \sum_{\substack{1 \leq i,j \leq N \\ i \neq j}} \beta_a \beta_i \beta_j \beta_b \\
    &= \beta_a \Bigg{(} (d - 1) \sum_{i=1}^N\ \beta^2_i + {\bigg{(} \sum_i \beta_i \bigg{)}}^2 \Bigg{)} \beta_b.
\end{align*}

Since \eqref{projectionConditionEquation} is assumed to be satisfied then $\widetilde{C}_{T_\beta}$ is an orthogonal projection and hence a positive operator. This way $C_{T_\beta}$ is positive. We have now from Lemma \ref{lem:permuPartialTrace}
\begin{align*}
    \Tr_{\scriptscriptstyle [\![ 0,N ]\!] \setminus \{0\}} \big{(} \widetilde{C}_{T_\beta} \big{)} &= \Tr_{\scriptscriptstyle [\![ 0,N ]\!] \setminus \{0\}} \bigg{[} \sum_{\substack{1 \leq a,b \leq N \\ \sigma \in \Sigma_{a,b}}} \frac{\beta_a \beta_b}{(N - 1)!} \sigma^{\To} \bigg{]} \\
    &= \Tr \Big{[} P^+_{\mathfrak{S}_{N - 1}} \Big{]} \bigg{(} \sum_i \beta_i \beta_i + \frac{1}{d} \sum_{i \neq j} \beta_i \beta_j \bigg{)} \cdot I \\
    &= \frac{\Tr P^+_{\mathfrak{S}_{N - 1}}}{d} \Bigg{[} (d - 1) \sum_i \beta^2_i + {\bigg{(} \sum_i \beta_i \bigg{)}}^2 \Bigg{]} \cdot I \\
    &= \frac{\Tr P^+_{\mathfrak{S}_{N - 1}}}{d} \cdot I
\end{align*}
Then
\begin{align*}
    \Tr_{\scriptscriptstyle [\![ 0,N ]\!] \setminus \{0\}} (C_{T_\beta}) &= \Tr_{\scriptscriptstyle [\![ 0,N ]\!] \setminus \{0\}} \Big{[} \frac{d}{\Tr P^+_{\mathfrak{S}_N}} \frac{N + d - 1}{N} \cdot \widetilde{C}_{T_\beta} \Big{]} \\
    &= \frac{\Tr P^+_{\mathfrak{S}_{N - 1}}}{\Tr P^+_{\mathfrak{S}_N}} \frac{N + d - 1}{N} \cdot I \\
    &= I,
\end{align*}
where the last equation comes from~$\Tr \big{[} P^+_{\mathfrak{S}_{N - 1}} \big{]} = \frac{N}{N + d - 1} \Tr \big{[} P^+_{\mathfrak{S}_{N}} \big{]}$.
\end{proof}

\begin{lemma}\label{lemm:appchoi1}
For any~$1 \leq a,b \leq N$, and any $\sigma \in \Sigma_{a,b}$, the permutation operator~$\sigma^{\To}$ is the Choi matrix a linear map~$T_{\mu,\nu}: \mathcal{M}_d \to {(\mathcal{M}_d)}^{\otimes N}$ defined by
\begin{equation*}
    T_{\mu,\nu}(X) = \Pi_\mu \big{(} X \otimes I^{\otimes (N - 1)} \big{)} \Pi_\nu
\end{equation*}
for some permutations $\mu$ and $\nu$ in $\mathfrak{S}_N$ such that $\mu(0) = a - 1$ and $\nu(b - 1) = 0$.

As a consequence for some fixed $1 \leq a,b \leq N$,
    \begin{equation*}
        \sum_{\sigma \in \Sigma_{a,b}} \Pi_\sigma^{\To} = \frac{1}{(N - 1)!} C_{T_{a,b}}
    \end{equation*}
    where
    \begin{equation*}
        T_{a,b}(X) = \sum_{\substack{\scriptstyle \mu,\nu \scriptstyle \in \mathfrak{S}_N \\ \scriptstyle \mu(0) \scriptstyle = a - 1 \\ \scriptstyle \nu(b - 1) \scriptstyle = 0}} \Pi_\mu \big{(} X \otimes I^{\otimes (N - 1)} \big{)} \Pi_\nu.
    \end{equation*}

\end{lemma}
\begin{proof}
    Let~$1 \leq a,b \leq N$, and $\sigma \in \Sigma_{a,b}$. then from Lemma \ref{lem:premutationForm} we have
    \begin{equation*}
        \Pi_\sigma^{\To} = \Pi_{(1 \: a)} \: (\omega \otimes \Pi_{\hat{\sigma}}) \; \Pi_{(1 \: b)},
    \end{equation*}
    such that on~$X \in \mathcal{M}_d$, the partial trace yields
    \begin{align*}
        \Tr_0 \big{[} &\Pi_\sigma^{\To} (X^\T \otimes I^{\otimes N}) \big{]} = \Tr_0 \big{[} \Pi_{(1 \: a)} \: (\omega \otimes \Pi_{\hat{\sigma}}) \; \Pi_{(1 \: b)} \; (X^\T \otimes I^{\otimes N}) \big{]} \\
        &= \Tr_0 \big{[} (I \otimes \Pi_{(0 \: (a-1))}) \: (\omega \otimes \Pi_{\hat{\sigma}}) \; (I \otimes \Pi_{(0 \: (b-1))}) \; (X^\T \otimes I^{\otimes N}) \big{]} \\
        &= \sum_{i,j \in [d]} \Tr_0 \big{[} (I \otimes \Pi_{(0 \: (a - 1))}) \: (\ketbra{i}{j} \otimes \ketbra{i}{j} \otimes \Pi_{\hat{\sigma}}) \; (I \otimes \Pi_{(0 \: (b - 1))}) \; (X^\T \otimes I^{\otimes N}) \big{]} \\
        &= \sum_{i,j \in [d]} \Tr \big{[} \ketbra{i}{j} X^\T\big{]} \cdot \Pi_{(0 \: (a - 1))} \: (\ketbra{i}{j} \otimes \Pi_{\hat{\sigma}}) \; \Pi_{(0 \: (b - 1))} \\
        &= \sum_{i,j \in [d]} \bra{i} X \ket{j} \cdot \Pi_{(0 \: (a - 1))} \: (\ketbra{i}{j} \otimes \Pi_{\hat{\sigma}}) \; \Pi_{(0 \: (b - 1))} \\
        &= \Pi_{(0 \: (a - 1))} \: (X \otimes \Pi_{\hat{\sigma}}) \; \Pi_{(0 \: (b - 1))} \\
        &= \Pi_{(0 \: (a - 1))} \: (X \otimes I^{\otimes (N - 1)}) \; (I \otimes \Pi_{\hat{\sigma}}) \; \Pi_{(0 \: (b - 1))}.
    \end{align*}
    And the result hold for~$\Pi_\mu = \Pi_{(0 \: (a - 1))}$ and~$\Pi_\nu = (I \otimes \Pi_{\hat{\sigma}}) \circ \Pi_{(0 \: (b - 1))}$.
\end{proof}

\begin{lemma}\label{lemm:appchoi2}
    Let an orthogonal projector~$\widetilde{C}_{T_\beta}$ for some reals~${(\beta_i)}_{1 \leq i \leq N}$ satisfying Eq.~\eqref{projectionConditionEquation}. Then for any~$1 \leq i \leq N$, we have
    \begin{equation*}
        \widetilde{C}_{T_\beta} \: \omega_{(0,i)} = \frac{1}{(N - 1)!} \sum_{\substack{1 \leq a \leq N \\ \sigma \in \Sigma_{a,i}}} \beta_a \bigg{(} (d - 1) \beta_i + \sum_{1 \leq b \leq N} \beta_b \bigg{)} \Pi_\sigma^{\To}.
\end{equation*}

Let Choi matrix~$C_{T_\beta}$ for some reals~${(\beta_i)}_{1 \leq i \leq N}$ satisfying Eq.~\eqref{projectionConditionEquation}. Then for any~$1 \leq i \leq N$, we have
    \begin{equation*}
        \Tr \Big{[} C_{T_\beta} \: \omega_{(0,i)} \Big{]} = d {\bigg{(} (d - 1) \beta_k + \sum^N_{j = 1} \beta_j \bigg{)}}^2.
    \end{equation*}    
\end{lemma}
\begin{proof}
    Since we are summing on all permutations~$\Pi^{\To}_{\sigma_{a,b}}$, a direct calculation gives us
    \begin{align*}
        \widetilde{C}_{T_\beta} \: \omega_{(0,i)} &= \frac{1}{(N - 1)!} \sum_{\substack{1 \leq a,b \leq N \\ \sigma \in \Sigma_{a,b}}} \beta_a \beta_b \cdot \Pi_\sigma^{\To} \: \omega_{(0,i)} \\
        &= \frac{1}{(N - 1)!} \sum_{1 \leq a \leq N} \beta_a \bigg{(} (d - 1) \sum_{\sigma \in \Sigma_{a,i}} \beta_i \cdot \sigma^{\To} + \sum_{\substack{1 \leq b \leq N \\ \sigma \in \Sigma_{a,i}}} \beta_b \cdot \Pi_\sigma^{\To} \bigg{)} \\ 
        &= \frac{1}{(N - 1)!} \sum_{\substack{1 \leq a \leq N \\ \sigma \in \Sigma_{a,i}}} \beta_a \bigg{(} (d - 1) \beta_i + \sum_{1 \leq b \leq N} \beta_b \bigg{)} \Pi_\sigma^{\To}.
    \end{align*}

And then we look at the action of the Choi matrix~$C_{T_\beta}$ onto some~$\omega_{(0,i)}$, and take the trace.

    The Lemma \ref{lem:permuPartialTrace} yields
    \begin{align*}
        \Tr \Big{[} C_{T_\beta} \: \omega_{(0,i)} \Big{]} &= \frac{d (N + d - 1)}{N \: \Tr \! \Big{[} P^+_{\mathfrak{S}_N} \Big{]}} \sum^N_{a = 1} \frac{\beta_a}{(N - 1!)} \bigg{(} (d - 1) \beta_i + \sum^N_{b = 1} \beta_b \bigg{)} \Tr \bigg{[} \sum_{\sigma \in \Sigma_{a,i}} \Pi_\sigma^{\To} \bigg{]} \\
        &= \frac{d (N + d - 1)}{N \: \Tr \! \Big{[} P^+_{\mathfrak{S}_N} \Big{]}} \bigg{(} (d - 1) \beta_i + \sum^N_{b = 1} \beta_b \bigg{)} \bigg{(} (d - 1) \beta_i + \sum^N_{a = 1} \beta_a \bigg{)} \Tr \Big{[} P^+_{\mathfrak{S}_{N - 1}} \Big{]} \\
        &= \frac{d (N + d - 1)}{N \: \Tr \! \Big{[} P^+_{\mathfrak{S}_N} \Big{]}} {\bigg{(} (d - 1) \beta_i + \sum^N_{j = 1} \beta_j \bigg{)}}^2 \Tr \Big{[} P^+_{\mathfrak{S}_{N - 1}} \Big{]} \\
        &= d {\bigg{(} (d - 1) \beta_k + \sum^N_{j = 1} \beta_j \bigg{)}}^2.
    \end{align*}
\end{proof}

\end{document}